\documentclass[11pt]{article}
\pdfoutput=1
\usepackage[margin=1in]{geometry}
\usepackage[utf8]{inputenc}
\usepackage{microtype}
\usepackage[all=normal,paragraphs=tight,bibliography=tight,floats=tight]{savetrees}
\usepackage[pdftex]{graphicx}
\usepackage[english]{babel}
\usepackage[cmex10]{amsmath}
\usepackage{amsthm}
\usepackage{amssymb}
\usepackage{xspace}
\usepackage{cite}
\usepackage{comment}
\usepackage{letltxmacro}
\usepackage{enumerate}
\usepackage{url}
\usepackage{textcomp}
\usepackage{wrapfig}
\usepackage{float}
\usepackage{caption}
\usepackage{subcaption}
\usepackage{tikz}
\usepackage{boxedminipage}
\usepackage{marvosym}

\newtheorem{theorem}{Theorem}
\newtheorem{lemma}[theorem]{Lemma}
\newtheorem{claim}[theorem]{Claim}
\newtheorem{corollary}[theorem]{Corollary}
\newtheorem{proposition}[theorem]{Proposition}

\theoremstyle{definition}
\newtheorem{definition}[theorem]{Definition}

\def\cqedsymbol{\ifmmode$\lrcorner$\else{\unskip\nobreak\hfil
\penalty50\hskip1em\null\nobreak\hfil$\lrcorner$
\parfillskip=0pt\finalhyphendemerits=0\endgraf}\fi} 

\newcommand{\cqed}{\renewcommand{\qed}{\cqedsymbol}}

\newcommand{\executeiffilenewer}[3]{%
\ifnum\pdfstrcmp{\pdffilemoddate{#1}}%
{\pdffilemoddate{#2}}>0%
{\immediate\write18{#3}}\fi%
} 
\newcommand{\Oh}{\ensuremath{\mathcal{O}}}
\newcommand{\Ohtilde}{\ensuremath{\widetilde{\mathcal{O}}}}

\newcommand{\Ff}{\mathcal{F}}
\newcommand{\Tt}{\mathcal{T}}

\newcommand{\parent}{\mathrm{parent}}

\newcommand{\reach}{\mathsf{reach}}
\newcommand{\Cc}{\mathcal{C}}
\newcommand{\ctw}{c_{\mathrm{tw}}}
\newcommand{\Nats}{\mathbb{N}}

\DeclareMathOperator{\distance}{dist}
\DeclareMathOperator{\distball}{Ball}

\newcommand{\Ii}{\mathcal{I}}
\newcommand{\Far}{\mathsf{Far}}
\newcommand{\patsiz}{\Pi}
\newcommand{\grasiz}{\Gamma}
\newcommand{\dstpot}{\Phi}
\newcommand{\dist}{\mathrm{dist}}
\newcommand{\pin}{\mathsf{in}}
\newcommand{\pout}{\mathsf{out}}
\newcommand{\cc}{\mathsf{cc}}
\newcommand{\Monster}{\mathsf{LB}}
\newcommand{\Comps}{\mathsf{CC}}
\newcommand{\Prv}{\mathsf{Prv}}
\newcommand{\Pub}{\mathsf{Pub}}

\newcommand{\trms}{T}
\newcommand{\light}{\trms^{\mathsf{li}}}
\newcommand{\heavy}{\trms^{\mathsf{he}}}
\newcommand{\ghost}{R}
\newcommand{\torso}[2]{#1\langle #2\rangle}

\newcommand{\torsoo}{\mathrm{torso}}
\newcommand{\adh}{\sigma}

\newcommand{\myparagraph}[1]{\medskip\noindent{\bf{#1}}}

\title{Subexponential parameterized algorithms for planar and apex-minor-free graphs via low treewidth pattern covering\thanks{%
The research leading to these results has received funding from the European Research Council under the European Union's Seventh Framework Programme (FP/2007-2013) / ERC Grant Agreements no. 267959 (F. V. Fomin),
no.~280152 (D. Marx), and no.~306992 (S. Saurabh).
The work of D.~Marx is also supported by Hungarian Scientific Research Fund (OTKA) grant NK105645.
The research of Ma. Pilipczuk is supported by Polish National Science Centre grant UMO-2013/09/B/ST6/03136.
The research of Mi. Pilipczuk is supported by Polish National Science Centre grant UMO-2013/11/D/ST6/03073.
Mi. Pilipczuk is also supported by the Foundation for Polish Science (FNP) via the START stipend programme.
Part of the research of F. Fomin, D. Marx, and Ma. Pilipczuk has been done when they were participating in the ``Fine-grained complexity and algorithm design'' program
at the Simons Institute for Theory of Computing in Berkeley.}}
\author{
  Fedor V. Fomin\thanks{
    Department of Informatics, University of Bergen, Norway, \texttt{fomin@ii.uib.no}.
  }
  \and
  Daniel Lokshtanov\thanks{
    Department of Informatics, University of Bergen, Norway, \texttt{daniello@ii.uib.no}.
  }
  \and
  D\'aniel Marx\thanks{
    Institute for Computer Science and Control, Hungarian Academy of Sciences (MTA SZTAKI), Hungary, \texttt{dmarx@cs.bme.hu}.
  }
  \and
  Marcin Pilipczuk\thanks{
    Institute of Informatics, University of Warsaw, Poland, \texttt{marcin.pilipczuk@mimuw.edu.pl}.
  }
  \and
  Micha\l{} Pilipczuk\thanks{
    Institute of Informatics, University of Warsaw, Poland, \texttt{michal.pilipczuk@mimuw.edu.pl}.
  }
  \and 
  Saket Saurabh\thanks{
    Institute of Mathematical Sciences, India, \texttt{saket@imsc.res.in}, and
    Department of Informatics, University of Bergen, Norway, \texttt{Saket.Saurabh@ii.uib.no}.
  }
}

\date{}

\begin{document}

\begin{titlepage}
\def\thepage{}
\thispagestyle{empty}
\maketitle

\begin{abstract}
We prove the following theorem.
Given a planar graph $G$ and an integer $k$, it is possible in polynomial time to randomly sample a subset $A$ of vertices of $G$ with the following properties:
\begin{itemize}
\item $A$ induces a subgraph of $G$ of treewidth $\Oh(\sqrt{k}\log k)$, and 
\item for every connected subgraph $H$ of $G$ on at most $k$ vertices, the probability that $A$ covers the whole vertex set of $H$ is at least $(2^{\Oh(\sqrt{k}\log^2 k)}\cdot n^{\Oh(1)})^{-1}$, 
where $n$ is the number of vertices of $G$.
\end{itemize}
Together with standard dynamic programming techniques for graphs of bounded treewidth, this result gives 
a versatile technique for obtaining (randomized) subexponential parameterized algorithms
for problems on planar graphs, usually with running time bound $2^{\Oh(\sqrt{k} \log^2 k)} n^{\Oh(1)}$.
The technique can be applied to problems expressible as searching for a small, connected pattern with a prescribed property in a large host graph; examples of such problems include
{\sc{Directed $k$-Path}}, {\sc{Weighted $k$-Path}}, {\sc{Vertex Cover Local Search}}, and {\sc{Subgraph Isomorphism}}, among others.
Up to this point, it was open whether these problems can be solved in subexponential parameterized time on planar graphs, because they are not amenable to the classic technique of bidimensionality.
Furthermore, all our results hold in fact on any class of graphs that exclude a fixed apex graph as a minor, in particular on graphs embeddable in any fixed surface.

\end{abstract}
\end{titlepage}

\newcommand{\pname}{\textsc}
\newcommand{\ProblemFormat}[1]{\pname{#1}}
\newcommand{\ProblemIndex}[1]{\index{problem!\ProblemFormat{#1}}}
\newcommand{\ProblemName}[1]{\ProblemFormat{#1}\ProblemIndex{#1}{}\xspace}

 
\newcommand{\probVC}{\ProblemName{Vertex Cover}}
\newcommand{\probPlVC}{\ProblemName{Planar Vertex Cover}}
\newcommand{\probFVS}{\ProblemName{Feedback Vertex Set}}
\newcommand{\probIFVS}{\ProblemName{Independent Feedback Vertex Set}}
\newcommand{\probCycPacking}{\ProblemName{Cycle Packing}}
\newcommand{\probPlFVS}{\ProblemName{Planar Feedback Vertex Set}}
\newcommand{\probIS}{\ProblemName{Independent Set}}
\newcommand{\probDS}{\ProblemName{Dominating Set}}
\newcommand{\probPlDS}{\ProblemName{Planar Dominating Set}}
\newcommand{\probCDS}{\ProblemName{Connected Dominating Set}}
\newcommand{\probOCT}{\ProblemName{Odd Cycle Transversal}}
\newcommand{\probHam}{\ProblemName{Hamiltonian Cycle}}
\newcommand{\probkCycle}{\ProblemName{$k$-Cycle}}
\newcommand{\probkDWCycle}{\ProblemName{Directed Weighted $k$-Cycle}}
\newcommand{\probPlHam}{\ProblemName{Planar Hamiltonian Cycle}}
\newcommand{\probHamP}{\ProblemName{Hamiltonian Path}}
\newcommand{\probKPath}{\ProblemName{Longest Path}}
\newcommand{\probDKPath}{\ProblemName{Directed Longest Path}}
\newcommand{\probWKPath}{\ProblemName{Weighted Longest Path}}
\newcommand{\probDWKPath}{\ProblemName{Directed Weighted Longest Path}}
\newcommand{\probthreeCol}{\ProblemName{$3$-Coloring}}
\newcommand{\probqCol}{\ProblemName{$q$-Coloring}}
\newcommand{\probSI}{\ProblemName{Subgraph Isomorphism}}
\newcommand{\probEDS}{\ProblemName{Edge Dominating Set}}
\newcommand{\probDirEMWay}{\ProblemName{Directed Edge Multiway Cut}}
\newcommand{\probEMCut}{\ProblemName{Edge Multicut}}
\newcommand{\probMulticut}{\ProblemName{Multicut}}
\newcommand{\probDirEMCut}{\ProblemName{Directed Edge Multicut}}
\newcommand{\probSkewEMCut}{\ProblemName{Skew Edge Multicut}}
\newcommand{\probDFVS}{\ProblemName{Directed Feedback Vertex Set}}
\newcommand{\probDFAS}{\ProblemName{Directed Feedback Arc Set}}
\newcommand{\probTw}{\ProblemName{Treewidth}}
\newcommand{\probSteiner}{\ProblemName{Steiner Tree}}
\newcommand{\probSteinerF}{\ProblemName{Steiner Forest}}
\newcommand{\probCVC}{\ProblemName{Connected Vertex Cover}}
\newcommand{\probCFVS}{\ProblemName{Connected Feedback Vertex Set}}
\newcommand{\probBisec}{\ProblemName{Minimum Bisection}}
\newcommand{\probMaxBisec}{\ProblemName{Maximum Bisection}}
\newcommand{\probMaxLeafSubtree}{\ProblemName{Max Leaf Subtree}}
\newcommand{\probTSP}{\ProblemName{TSP}}
\newcommand{\probSubTSP}{\ProblemName{Subset TSP}}
\newcommand{\probMaxLeafOut}{\ProblemName{Max Leaf Outbranching}}
\newcommand{\probLSVC}{\ProblemName{LS Vertex Cover}}
\newcommand{\probLSDS}{\ProblemName{LS Dominating Set}}

\newcommand{\NP}{$\mathsf{NP}$}

\section{Introduction}\label{sec:intro}
Most of natural \NP-hard problems on graphs remain \NP-hard even when the input graph is restricted to be planar.
However, it was realized already in the dawn of algorithm design that the planarity of the input can be exploited algorithmically.
Using the classic planar separator theorem of Lipton and Tarjan~\cite{LiptonT80}, one can design algorithms working in subexponential time, usually of the form $2^{\Oh(\sqrt{n})}$ or $2^{\Oh(\sqrt{n}\log n)}$,
for a wide variety of problems that behave well with respect to separators; such running time cannot be achieved on general graph unless the Exponential Time Hypothesis (ETH) fails~\cite{ImpagliazzoP01}.
From the modern perspective, the planar separator theorem implies that a planar graph on $n$ vertices has treewidth $\Oh(\sqrt{n})$, and the obtained tree decomposition can be used to
run a Divide\&Conquer algorithm or a dynamic programming subroutine. 
The idea of exploiting short separators plays a crucial role in modern algorithm design on planar graphs and related graph classes, including polynomial-time, approximation, and parameterized algorithmic paradigms.

Let us take a closer look at the area of parameterized complexity. While for most parameterized \NP-hard problems the exponential dependence on the parameter is the best we can hope for, under ETH, 
there are plenty of problems admitting subexponential parameterized algorithms when restricted to planar graphs.
This was first observed in $2000$ by Alber et al.~\cite{AlberBFKN02}, who obtained a subexponential parameterized algorithm  for deciding whether a given $n$-vertex planar 
graph contains a dominating set of size $k$ in time $2^{\Oh(\sqrt{k})} n^{\Oh(1)}$. 
It turned out that the phenomenon is much more general. A robust framework explaining why various problems like 
\probFVS, \probVC,  \probDS or \probKPath admit subexponential parameterized algorithms on planar graphs, 
usually with running times of the form $2^{\Oh(\sqrt{k})}\cdot n^{\Oh(1)}$ or $2^{\Oh(\sqrt{k}\log k)}\cdot n^{\Oh(1)}$,
is provided by the bidimensionality theory of Demaine et al.~\cite{DemaineFHT05}. 
Roughly speaking, the idea is to exploit the relation between the treewidth of a planar graph and the size of a grid minor that can be found in it.
More precisely, the following win/win approach is implemented. Either the treewidth of the graph is $\Oh(\sqrt{k})$ and the problem can be solved using dynamic programming on a tree decomposition,
or a $c\sqrt{k}\times c\sqrt{k}$ grid minor can be found, for some large constant $c$, which immediately implies that we are working with a yes- or with a no-instance of the problem.
Furthermore, it turns out that for a large majority of problems, the running time yielded by bidimensionality are essentially optimal under ETH: no $2^{o(\sqrt{k})}\cdot n^{\Oh(1)}$-time algorithm can be expected.
We refer to the survey \cite{Demaine:2008mi}, as well as the textbook 
\cite[Chapter~7]{cygan2015parameterized} for an overview of bidimensionality and its applications.

While the requirement that the problem can be solved efficiently on bounded treewidth graphs is usually not restrictive, the assumption that uncovering {\em{any}} large grid minor provides
a meaningful insight into the instance considerably limits the applicability of the bidimensionality methodology.
Therefore, while bidimensionality can be extended to more general classes, like excluding some fixed graph as a 
minor \cite{DemaineFHT05,DemaineH08}, map graphs \cite{DemaineFHT05talg}, or unit disk graphs~\cite{FominLS12},
there are many problems that are ``$\varepsilon$-close'' to being bidimensional, and yet their parameterized complexity remains open for years.

One example where such a situation occurs is the \probDKPath problem.
While the existence of a $\sqrt{k}\times \sqrt{k}$ grid minor in an undirected graph immediately implies the existence of an undirected path on $k$ vertices, the same principle
cannot be applied in the directed setting: even if we uncover a large grid minor in the underlying undirected graph, there is no guarantee that a long directed path can be
found, because we do not control the orientation of the arcs. Thus, \probDKPath can be solved in time $2^{\Oh(k)}n^{\Oh(1)}$ on general graphs using color coding~\cite{AlonYZ}, 
but no substantially faster algorithms on planar graphs were known. On the other hand, no $2^{o(\sqrt{k})}\cdot n^{\Oh(1)}$-time algorithm on planar graphs can be expected under ETH, which leaves a large
gap between the known upper and lower bounds. Closing this embarrassing gap was mentioned as an open problem in \cite{borradaile_et_al:DR:2014:4427,Dorn:2013wd,Tazari2010,Tazari12}.
A similar situation happens for \probWKPath, where we are looking for a $k$-path of minimum weight in an edge-weighted planar graph; the question
about the complexity of this problem was raised in~\cite{borradaile_et_al:DR:2014:4427,SauT10}.
Another example is \probkCycle: deciding whether a given planar graph contains a cycle of length exactly $k$. While the property of admitting a cycle of length {\em{at least $k$}}
is bidimensional, and therefore admits a $2^{\Oh(\sqrt{k})}n^{\Oh(1)}$-time algorithm on planar graphs, the technique fails for the variant when we ask for length {\em{exactly $k$}}.
This question was asked in~\cite{borradaile_et_al:DR:2014:4427}. 
We will mention more problems with this behavior later on.

The theme of ``subexponential algorithms beyond bidimensionality'' has recently been intensively investigated, with various success.
For a number of specific problems such algorithms were found; 
these include planar variants of \probSteiner parameterized by the size of the tree~\cite{pilipczuk_et_al:LIPIcs:2013:3947,PilipczukPSL14}, \probSubTSP~\cite{KleinM14}, or \probMaxLeafOut~\cite{Dorn:2013wd}.
In all these cases, the algorithms were technically very involved and depended heavily on the combinatorics of the problem at hand.
A more systematic approach is offered by the work of Dorn et al.~\cite{Dorn:2013wd}  and  Tazari~\cite{Tazari2010,Tazari12}, who
obtained ``almost" subexponential algorithm for   \probDKPath on planar, and more generally, apex-minor-free graphs. More precisely, they proved that 
for any $\varepsilon>0$, there is $\delta$ such that the  \probDKPath  problem 
 is solvable in time $\Oh((1+\varepsilon)^k \cdot n^\delta)$ on planar directed graphs, and more generally, on directed graphs whose underlying undirected graph excludes a fixed apex graph as a minor.
This technique can be extended to other problems that can be characterized as searching for a connected pattern in a large host graph, 
which suggests that some more robust methodology is hiding just beyond the frontier of our understanding.

\myparagraph{Main result.} In this paper, we introduce a versatile technique for solving such problems in subexponential parameterized time, by proving the following theorem.

\begin{theorem}\label{thm:maintheorem}
Let $\Cc$ be a class of graphs that exclude a fixed apex graph as a minor.
Then there exists a randomized polynomial-time algorithm that, given an $n$-vertex graph $G$ from $\Cc$ and an integer $k$,
samples a vertex subset $A\subseteq V(G)$ with the following properties:
\begin{enumerate}[(P1)]
\item The induced subgraph $G[A]$ has treewidth $\Oh(\sqrt{k}\log k)$.
\item For every vertex subset $X\subseteq V(G)$ with $|X|\leq k$ that induces a connected subgraph of $G$, 
the probability that $X$ is covered by $A$, that is $X\subseteq A$, is at least $(2^{\Oh(\sqrt{k}\log^2 k)}\cdot n^{\Oh(1)})^{-1}$.
\end{enumerate}
\end{theorem}

Here, by an apex graph we mean a graph that can be made planar by removing one vertex. Note that Theorem~\ref{thm:maintheorem} in particular applies to planar graphs, and to graphs
embeddable in a fixed surface. 

%
%
%
 
\myparagraph{Applications.} 
Similarly as in the case of bidimensionality, Theorem~\ref{thm:maintheorem} provides a simple recipe for obtaining subexponential parameterized algorithms: 
check how fast the considered problem can be solved on graphs of bounded treewidth, and then combine the treewidth-based algorithm with Theorem~\ref{thm:maintheorem}. 
We now show how Theorem~\ref{thm:maintheorem} can be used to obtain randomized subexponential parameterized algorithms for a variety of problems on apex-minor-free classes;
for these problems, the existence of such algorithms so far was open even for planar graphs. We only list the most interesting examples to showcase possible applications.

\medskip
\noindent\emph{Directed and weighted paths and cycles.}  
As mentioned earlier, the question about the existence of subexponential parameterized algorithms for \probDKPath and \probWKPath on planar graphs 
was asked in~\cite{borradaile_et_al:DR:2014:4427,Dorn:2013wd,Tazari2010,Tazari12}. 
Let us observe that on a graph of treewidth\footnote{For \probDKPath we speak about the treewidth of the underlying undirected graph.} $t$, both 
\probDKPath and \probWKPath, as well as their different combinations, like finding a maximum or minimum weight directed path or cycle on $k$ vertices, 
are solvable in time $2^{\Oh(t\log{t})}n^{\Oh(1)}$ by the standard dynamic programming; see e.g.~\cite[Chapter~7]{cygan2015parameterized}. 
This running time can be improved to single-exponential time $2^{\Oh(t)}n^{\Oh(1)}$~\cite{BodlaenderCKN15,DornFT12,FominLS14}.

In order to obtain a subexponential parameterized algorithm for, say, \probDKPath on planar directed graphs, we do the following.
Let $G$ be the given planar directed graph, and let $U(G)$ be its underlying undirected graph. 
Apply the algorithm of Theorem~\ref{thm:maintheorem} to $U(G)$, which in polynomial time samples a subset $A\subseteq V(U(G))$ such that $G[A]$ has treewidth at most $\Oh(\sqrt{k}\log k)$,
and the probability that $A$ covers some directed $k$-path in $G$, provided it exists, is at least $(2^{\Oh(\sqrt{k}\log^2 k)}\cdot n^{\Oh(1)})^{-1}$.
Then, we verify whether $G[A]$ admits a directed $k$-path using standard dynamic programming in time $2^{\Oh(\sqrt{k}\log k)}\cdot n^{\Oh(1)}$.
Provided some directed $k$-path exists in the graph, this algorithm will find one with probability at least $(2^{\Oh(\sqrt{k}\log^2 k)}\cdot n^{\Oh(1)})^{-1}$.
Thus, by making $2^{\Oh(\sqrt{k}\log^2 k)}\cdot n^{\Oh(1)}$ independent runs of the algorithm, we can reduce the error probability to at most $1/2$.
All in all, the obtained algorithm runs in time $2^{\Oh(\sqrt{k}\log^2 k)}\cdot n^{\Oh(1)}$ and can only report false negatives with probability at most $1/2$

Note that in order to apply the dynamic programming algorithm, 
we need to construct a suitable tree decomposition of $G[A]$. However, a variety of standard algorithms, e.g., the classic $4$-approximation of Robertson and Seymour~\cite{RobertsonS-GMXIII},
can be used to construct such an approximate tree decomposition within the same asymptotic running time.
Actually, a closer look into the proof of Theorem~\ref{thm:maintheorem} reveals that the algorithm can construct, within the same running time, 
a tree decomposition of $G[A]$ certifying the upper bound on its treewidth.

Observe that the same approach works also for any apex-minor-free class $\Cc$, and can be applied also to \probWKPath and  \probkCycle. 
We obtain the following corollary. 
\begin{corollary}\label{thm:subexpdirectedpath_apps1}
Let $\Cc$ be a class of graphs that exclude some fixed apex graph as a minor.  
Then problems: \probWKPath, \probkCycle, and \probDKPath, are all solvable in randomized time $2^{\Oh(\sqrt{k} \log^2 k)}\cdot n^{\Oh(1)}$ on graphs from $\Cc$.
In case of \probDKPath, we mean that the underlying undirected graph of the input graph belongs to $\Cc$.
\end{corollary}

Note here that the approach presented above works in the same way for various combinations and extensions of problems in Corollary~\ref{thm:subexpdirectedpath_apps1}, 
like weighted, colored, or directed variants, possibly with some constraints on in- and out-degrees, etc.
In essence, the only properties that we need is that the sought pattern persists in the subgraph induced by the covering set $A$, 
and that it can be efficiently found using dynamic programming on a tree decomposition.
To give one more concrete example, Sau and Thilikos in~\cite{SauT10} studied the problem of finding a connected $k$-edge subgraph with all vertices of degree at most some integer $\Delta$; 
for $\Delta=2$ this corresponds to finding a $k$-path or a $k$-cycle. 
For fixed $\Delta$ they gave a subexponential algorithm on (unweighted) graphs excluding some fixed graph as a minor, and asked if the weighted version of this problem can be solved in subexponential parameterized time. 
Theorem~\ref{thm:maintheorem} immediately implies that for fixed $\Delta$, the weighted variant of the problem is solvable in randomized time $2^{\Oh(\sqrt{k} \log^2 k)}\cdot n^{\Oh(1)}$ on apex-minor-free graphs. 

 \medskip
\noindent\emph{Subgraph Isomorphism.} \probSI  is a fundamental problem, where we are given two graphs: 
an $n$-vertex \emph{host} graph $G$ and a $k$-vertex \emph{pattern} graph $P$. 
The task is to decide whether $P$ is isomorphic to a subgraph of $G$.
Eppstein~\cite{Eppstein99} gave an algorithm solving \probSI on planar graphs in time $k^{\Oh(k)} n$, which was subsequently improved by
Dorn~\cite{dorn:LIPIcs:2010:2460} to $2^{\Oh(k)} n$. 
The first implication of our main result for \probSI concerns the case when the maximum degree of $P$ is bounded by a constant.
 Matou\v{s}ek and Thomas~\cite{MatousekT92} proved that if a tree decomposition of the host graph $G$ of width $t$ is given,
and the pattern graph $P$ is connected and of maximum degree at most some constant $\Delta$,
then deciding whether $P$ is isomorphic to a subgraph of $G$ can be done in time $\Oh(k^{t+1} n)$. 
By combining this with Theorem~\ref{thm:maintheorem} as before, we obtain the following.

\begin{corollary}\label{thm:subgraphisodegree_apps1}
Let $\Cc$ be a class of graphs that exclude some fixed apex graph as a minor, and let $\Delta$ be a fixed constant.
Then, given a connected graph $P$ with at most $k$ vertices and maximum degree not exceeding $\Delta$, and a graph $G\in \Cc$ on $n$ vertices,
it is possible to decide whether $P$ is isomorphic to a subgraph of $G$ in randomized time $2^{\Oh(\sqrt{k} \log^2 k)}\cdot n^{\Oh(1)}$.
\end{corollary}
 
In a very recent, unpublished work, Bodlaender et al.~\cite{Bod-private} proved  that \probSI on planar graphs cannot be solved in time $2^{o(n/\log{n})}$ unless ETH fails.
The lower bound  of Bodlaender et al. holds for two very special cases. 
The first case is when the pattern graph $P$ is a tree and has only one vertex of super-constant degree. 
Then second case is  when $P$ is not connected, but its maximum degree is a constant. 
Thus, the results of Bodlaender et al. show that both the connectivity and the bounded degree constraints on pattern $P$ in Corollary~\ref{thm:subgraphisodegree_apps1} 
are necessary to keep the  square root dependence on  $k$ in the exponent. 
However, a possibility of solving \probSI in time $2^{\Oh(k/\log{k})}\cdot n^{\Oh(1)}$, which is still parameterized subexponential, is not ruled out by the work of Bodlaender et al. 
Interestingly enough, Bodlaender et al.~\cite{Bod-private}, also give a matching dynamic programming algorithm that can be combined with our theorem.

\begin{theorem}[\cite{Bod-private}]\label{thm:bod-private}
Let $H$ be a fixed graph, and let us fix any $\varepsilon>0$.
Given a pattern graph $P$ on at most $k$ vertices, and an $H$-minor-free host graph $G$ of treewidth at most $\Oh(k^{1-\varepsilon})$,
it is possible to decide whether $P$ is isomorphic to a subgraph of $G$ in time  $2^{\Oh(k/\log{k})}\cdot n^{\Oh(1)}$.
\end{theorem}

We remark that Bodlaender et al.~\cite{Bod-private} claim only a subexponential exact algorithm with running time $2^{\Oh(n/\log{n})}$, but the dynamic programming subroutine underlying this result actually
gives the stronger algorithmic result as stated above.

By combining Theorem~\ref{thm:bod-private} with Theorem~\ref{thm:maintheorem} in the same way as before, we obtain the following.  
\begin{corollary}\label{thm:SItwo}
Let $\Cc$ be a class of graphs that exclude some fixed apex graph as a minor.
Then, given a connected graph $P$ with at most $k$ vertices, and a graph $G\in \Cc$ on $n$ vertices,
it is possible to decide whether $P$ is isomorphic to a subgraph of $G$ in randomized time $2^{\Oh(k/\log{k})}\cdot n^{\Oh(1)}$.
\end{corollary}
Let us stress here that the lower bounds of Bodlaender et al.~\cite{Bod-private} show that the running time given by Corollary~\ref{thm:SItwo} is tight: 
no $2^{o(k/\log{k})}\cdot n^{\Oh(1)}$-time algorithm can be expected under ETH.

  \medskip
\noindent\emph{Local search.} Fellows et al.~\cite{FellowsFLRSV12} studied the following parameterized local search problem on apex-minor-free graphs. 
In the \probLSVC problem we are given an $n$-vertex graph $G$, a vertex cover $S$ in $G$, and an integer $k$. 
The task is to decide whether $G$ contains a vertex cover $S'$, such that $|S'|< |S|$ and the Hamming distance $|S\triangle S'|$ between sets $S$ and $S'$ is at most $k$. 
In other words, for a given vertex cover, we ask if there is a smaller vertex cover which is $k$-close to the given one, in terms of Hamming (edit) distance. 
Fellows et al.~\cite{FellowsFLRSV12} gave an algorithm solving \probLSVC in time $2^{\Oh(k)}\cdot n^{\Oh(1)}$ on planar graphs.
The question whether this can be improved to subexponential parameterized time was raised in~\cite{demaine_et_al:DR:2013:4013,FellowsFLRSV12}.

The crux of the approach of  Fellows et al.~\cite{FellowsFLRSV12} is the following observation. 
If there is solution to  \probLSVC, then there is a solution $S'$, such that  $S\triangle S'$ induces a connected subgraph in $G$.
Since $S\triangle S'$ contains at most $k$ vertices and is connected, 
our Theorem~\ref{thm:maintheorem} can be used to sample a vertex subset $A$ that induces a subgraph of treewidth $\Oh(\sqrt{k}\log k)$ and covers $S\triangle S'$ with high probability.
Thus, by applying the same principle of independent repetition of the algorithm, we basically have search for suitable sets $S\setminus S'$ and $S'\setminus S$ in the subgraph of $G$ induced by $A$. 
We should be, however, careful here: there can be edges between $A$ and its complement, and these edges also need to be covered by $S'$, so we cannot just restrict our attention to $G[A]$.
To handle this, we apply the following preprocessing. 
For every vertex $v\in A$, if $v$ is adjacent to some vertex outside of $A$ that is not included $S$, then $v$ must be in $S$ and needs also to remain in $S'$. 
Hence, we delete all such vertices from $G[A]$, and it is easy to see that now the problem boils down to looking for feasible $S\setminus S'$ and $S'\setminus S$ within the obtained induced subgraph.
This can be easily done in time $2^{\Oh(t)}\cdot n^{\Oh(1)}$, where $t\leq \Oh(\sqrt{k}\log k)$ is the treewidth of this subgraph; hence we obtain the following:
\begin{corollary}\label{thm:ls}
Let $\Cc$ be a class of graphs that exclude some fixed apex graph as a minor.
Then \probLSVC on graphs from $\Cc$ can be solved in randomized time $2^{\Oh(\sqrt{k} \log^2 k)}\cdot n^{\Oh(1)}$.
\end{corollary}


\medskip
\noindent\emph{Steiner tree.} \probSteiner is a fundamental network design problem: 
for a graph $G$ with a prescribed set of terminal vertices $S$, and an integer $k$, 
we ask whether there is a tree on at most $k$ edges that spans all terminal vertices. 
Pilipczuk et al.~\cite{pilipczuk_et_al:LIPIcs:2013:3947} gave an algorithm for this problem with running time $2^{\Oh((k\log{k})^{2/3})}\cdot n$ on planar graphs and on graphs of bounded genus. 
With much more additional work, the running time was improved to $2^{\Oh(\sqrt{k\log{k}})} \cdot n$ in~\cite{PilipczukPSL14}. 

Again, by combining the standard dynamic programming solving \probSteiner on graphs of treewidth $t$ in time $2^{\Oh(t\log{t})}n^{\Oh(1)}$, see e.g. \cite{ChimaniMZ12}, with Theorem~\ref{thm:maintheorem},
we immediately obtain the following. 

\begin{corollary}\label{thm:steiner}
Let $\Cc$ be a class of graphs that exclude some fixed apex graph as a minor.
Then \probSteiner on graphs from $\Cc$ can be solved in randomized time $2^{\Oh(\sqrt{k} \log^2 k)}\cdot n^{\Oh(1)}$.
\end{corollary}

Contrary to the much more involved algorithm of Pilipczuk et al.~\cite{PilipczukPSL14}, the algorithm above can equally easily handle various variants of the problem.
For instance, we can look for a Steiner tree on $k$ edges that minimizes the total weight of the edges, or we can ask for a Steiner out-branching in a directed graph, or 
we can put additional constraints on vertex degrees in the tree, and so on.

\myparagraph{Outline.} In Section~\ref{sec:overview} we give an informal overview of the proof of Theorem~\ref{thm:maintheorem} for the case of planar graphs. We try rather to focus on main intuitions
and concepts, than describe technical details necessary for a formal reasoning. Then, in Sections~\ref{sec:prelims} and~\ref{sec:aux-tools}
we recall the standard concepts and introduce auxiliary technical results. The full proof of Theorem~\ref{thm:maintheorem} is contained in Section~\ref{sec:proof}.
We conclude in Section~\ref{sec:conc} by listing open problems raised by our work. The appendix outlines possible extensions
to multiple components of the pattern and arbitrary proper minor-closed graph classes.

\section{Overview of the proof of Theorem~\ref{thm:maintheorem}}\label{sec:overview}

We now give an informal overview of the proof of Theorem~\ref{thm:maintheorem} in the case of planar graphs.
In fact, the only two properties of planar graphs which are essential to the proof are (a) planar graphs are minor-closed, and (b) they have locally bounded treewidth
by a linear function, that is, there exists a constant $\ctw$ such that every planar graph of radius $k$ has treewidth at most $\ctw \cdot k$.
In fact, for planar graphs one can take $\ctw = 3$~\cite{RobertsonS3}, and as shown in~\cite{DemaineH04}, the graph classes satisfying both (a) and (b) are exactly graph classes excluding a fixed apex graph as a minor.
However, in a planar graph we can rely on some topological intuition, making
the presentation more intuitive.
In the description we assume familiarity with tree decompositions; see Section~\ref{sec:prelims} for a formal definition.

\myparagraph{Locally bounded treewidth of planar graphs.}
As a warm-up, let us revisit a proof that planar graphs have locally bounded treewidth. 
The considered proof yields a worse constant than $\ctw = 3$, but it is insightful for our argumentation.
Let $G_0$ be a graph of radius $k$, that is, there exists a root vertex $r_0$
such that every vertex of $G_0$ is within distance at most $k$ from $r_0$.

As with most proofs showing that a graph in question has bounded treewidth, we will recursively construct a tree decomposition of bounded width. 
To this end, we need to carefully define the state of the recursion. 
We do it as follows: the recursive step aims at decomposing a subgraph $G$ of the input graph $G_0$, with some chosen set of terminals $T \subseteq V(G)$ on the outer face of $G$. 
The terminals $T$ represent connections between $G$ and the rest of $G_0$. 
In order to be able to glue back the obtained decompositions from the recursive step, our goal is to provide a tree decomposition of $G$ with $T$ contained in the root bag of the decomposition, 
so that later we can connect this bag to decompositions of other pieces of the graph that also contain the vertices of $T$. 
During the process, we keep the invariant that $|T| \leq 8(k+2)$, allowing us to bound the width of the decomposition.
Furthermore, the assumption that $G_0$ is of bounded radius projects onto the recursive subinstances by the following invariant: every vertex of $G$ is within distance at most $k$ from some terminal.

In the recursive step, if $T = V(G)$, $|T| < 8(k+2)$, or
$G$ is not connected, then we can perform some simple steps that we do not discuss here. The interesting case is when $|T| = 8(k+2)$.

We partition $T$ along the outer face into four parts of size $2(k+2)$ each, called north, east, south, and west terminals. 
We compute minimum vertex cuts between the north and the the south terminals, and between the east and the west terminals.
If, in any of these directions, a cut $W$ of size strictly smaller than $2(k+2)$ is found, then we can make a divide-and-conquer step:
for every connected component $D$ of $G-(T \cup W)$ we recurse on the graph $G[N[D]]$ with terminals $N(D)$, obtaining a tree decomposition $\mathcal{T}_D$.
Finally, we attach all the obtained tree decompositions below a fresh root node with bag $T \cup W$,
which is of size less than $10(k+2)$.

The crux is that such a separator $W$ is always present in the graph.
Indeed, otherwise there would exist $2(k+2)$ disjoint paths between the north and the south terminals
and $2(k+2)$ disjoint paths between the east and the west terminals. 
Consider the region bounded by the two middle north-south and the two middle east-west paths:
the vertices contained in this region are within distance larger than $k$ from the outer face, on which all terminals lie (see Fig.~\ref{fig-over:ctw}).
This contradicts our invariant.

\begin{figure}[tb]
\centering
\begin{subfigure}{.49\textwidth}
\centering
\includegraphics[width=.5\linewidth]{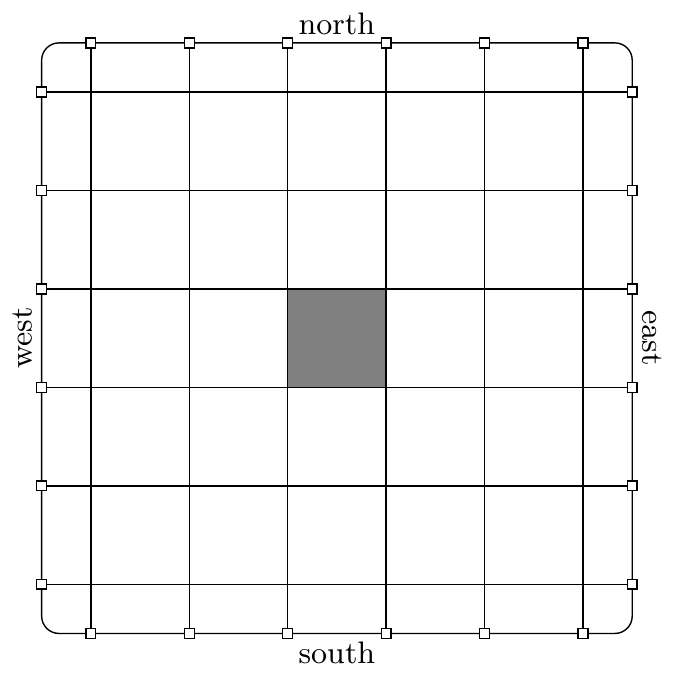}
\caption{}
\label{fig-over:ctw}
\end{subfigure}%
\begin{subfigure}{.49\textwidth}
 \centering
 \includegraphics[width=.6\linewidth]{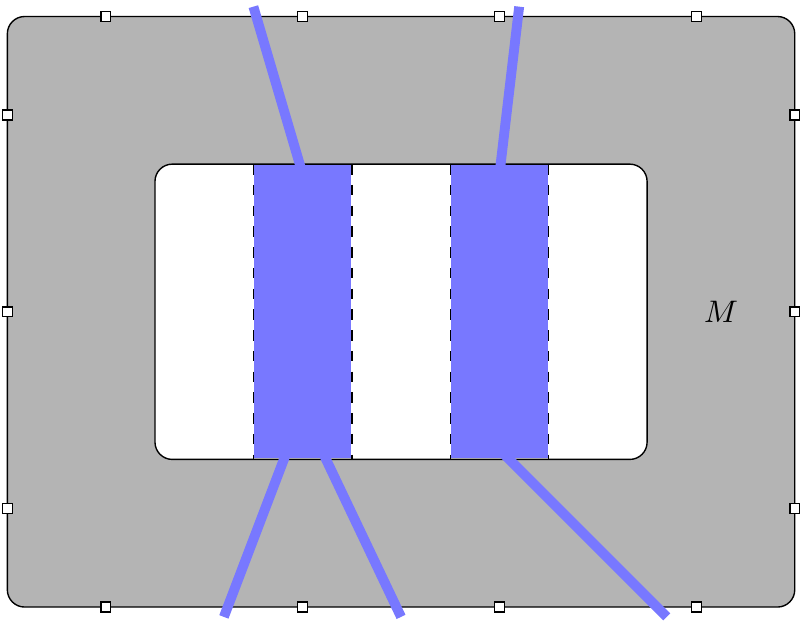}
 \caption{}\label{fig-over:std-sep}
\end{subfigure}\\[2mm]
\begin{subfigure}{.49\textwidth}
 \centering
 \includegraphics[width=.6\linewidth]{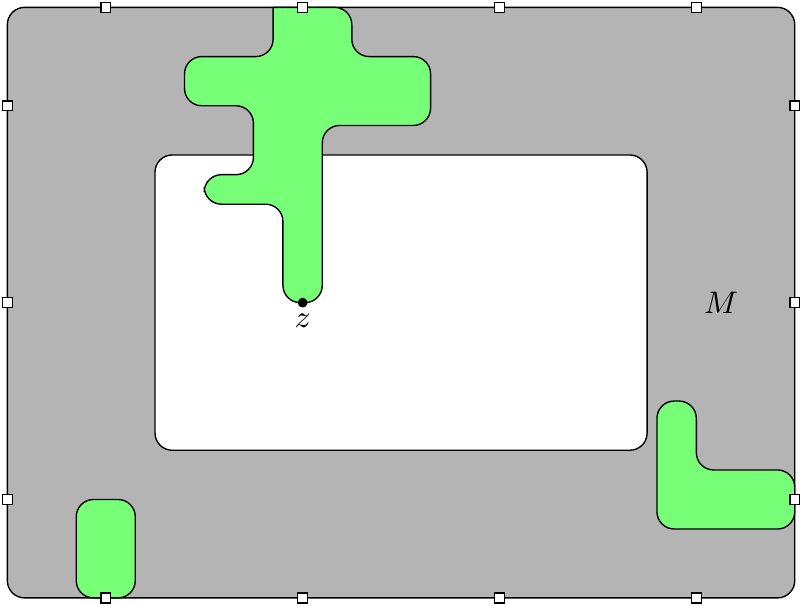}
 \caption{}\label{fig-over:stretched}
\end{subfigure}%
\begin{subfigure}{.49\textwidth}
 \centering
 \includegraphics[width=.6\linewidth]{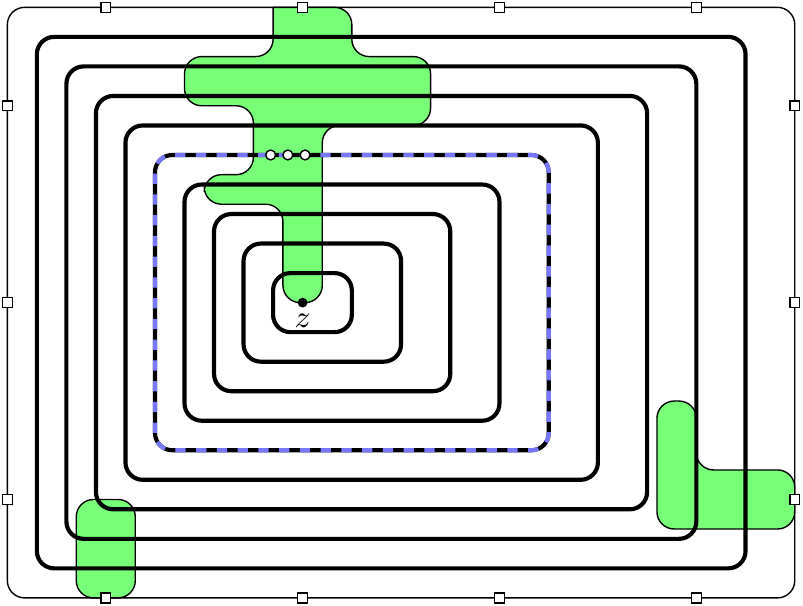}
 \caption{}\label{fig-over:cycles}
\end{subfigure}\\[2mm]
\begin{subfigure}{.98\textwidth}
 \centering
 \includegraphics[width=.7\linewidth]{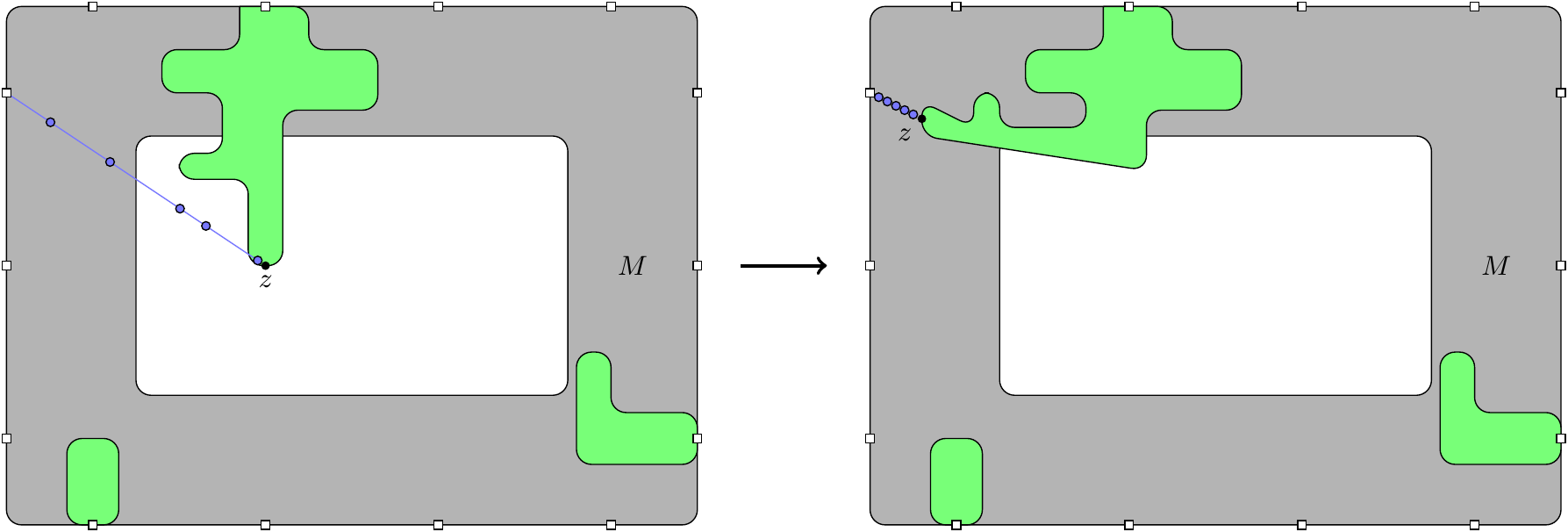}
 \caption{}\label{fig-over:contract}
\end{subfigure}%
\caption{Illustrations for Section~\ref{sec:overview}.
(a) The proof that planarity implies locally bounded treewidth.
The vertices in the gray area are too far from the terminals.
(b) Standard partitioning step. The margin is gray and the islands
are separated by dashed lines. The blue separator consists of $\Ohtilde(\sqrt{k})$ islands and
vertices of the margin. If the blue islands are disjoint from the solution, we delete them and obtain a balanced separator of size $\Ohtilde(\sqrt{k})$.
(c) The situation if the standard partitioning step cannot be applied: we have
a component of the pattern (green) stretched between a light terminal and a vertex
$z$ inside the margin.
(d) Chain of $z$-$\light$ separators: a sparse separator
that partitions the pattern in a balanced fashion is highlighted.
(e) Contraction of a path $P_i$ (blue) onto
its public vertices (blue circles). A significant number of
the vertices of the pattern become much closer to the light terminals.}
\label{fig-over}
\vspace{-0.3cm}
\end{figure}

\myparagraph{Our recursion.}
In our case, we use a similar, but much more involved recursion scheme. In the recursive step,
we are given a \emph{minor} $G$ of the input graph $G_0$, with some \emph{light}
terminals $\light \subseteq V(G)$ on the outer face, and some \emph{heavy} terminals $\heavy \subseteq V(G) \setminus \light$ lying anywhere in the graph.
As before, the terminals represent connections to the other parts of the graph.
We require that the terminals $\trms := \light \cup \heavy$
need to be contained in the root bag of the tree decomposition that is going to be constructed in this recursion step. 
Moreover, we maintain an invariant that $|\trms| = \Ohtilde(\sqrt{k})$, in order to bound the width of the decomposition. 
The graph $G$ is a minor of the input graph $G_0$, since
we often prefer to contract some edges instead of deleting them; thus we maintain
some distance properties of $G$.

In our recursion, the light terminals originate from cutting the graph
in a similar fashion as in the proof that planar graphs have locally bounded treewidth, presented above.
Hence, we keep the invariant that light terminals lie on the outer face. We sometimes need to cut deeply
inside~$G$. The produced terminals are heavy, but every such step 
corresponds to a significant progress in detecting the pattern, and hence
such steps will be rare. In every such step, we artificially 
provide connectivity of the subinstances through the heavy terminals;
this is technical and omitted in this overview.

Recall that our goal is to preserve a connected $k$-vertex pattern from 
the input graph. Here, the pattern can become disconnected by recursing on subsequent separations, but such cuttings will always be along light terminals.
Therefore, we define a {\em{pattern}} in a subinstance $(G,\light,\heavy)$ solved in the recursion as 
a set $X \subseteq V(G)$ of size at most $k$ such that every connected component
of $G[X]$ contains a light terminal.
Hence, compared to the presented proof of planarity implying locally bounded treewidth,
we aim at more restricted width of the decomposition, namely $\Ohtilde(\sqrt{k})$, 
but we can contract or delete parts of $G$, as long as the probability of spoiling a fixed, but unknown $k$-vertex pattern $X$ remains inverse subexponential in $k$.

\myparagraph{Clustering.}
Upon deleting a vertex or an edge,
some distance properties that we rely on can be broken. We need
the following sampling procedure that partitions the graph into connected
components of bounded radii, such that the probability of spoiling a particular
pattern is small.

\begin{theorem}\label{thm:clust-over}
There exists a randomized polynomial-time algorithm that, given a graph $G$ on $n>1$ vertices and a positive integer $k$, returns a vertex subset $B\subseteq V(G)$ with the following properties:
\begin{enumerate}[(a)]
\item The radius of each connected component of $G[B]$ is less than $9k^2\log n$.
\item For each vertex subset $X\subseteq V(G)$ of size at most $k$, the probability that $X \subseteq B$ is at least $1-\frac{1}{k}$.
\end{enumerate}
\end{theorem}
\begin{proof}[Proof sketch]
Start with $H := G$ and iteratively, as long as $V(H) \neq \emptyset$, perform the following procedure. Pick arbitrary $v \in V(H)$, and choose a radius $r$ as follows.
Start with $r = 1$ and iteratively, given current radius $r$, with probability $p := (2k^2)^{-1}$ accept $r$, 
and with probability $1-p$ increase it by one and continue (i.e., choose $r$ according to the geometric distribution with success probability $p$). 
Given an accepted radius $r$, put all vertices within distance \emph{less} than $r$ from $v$ into $B$, and delete from $H$ all vertices within distance \emph{at most} $r$ from $v$.

Since the procedure performs at most $n$ steps, by union bound
the probability of some radius exceeding $9k^2 \lg n$ is at most $\frac{1}{2k}$.
Fix a vertex $x \in V(G)$. We have $x \notin B$ only if at some point $x \in V(H)$
and the distance between $v$ and $x$ in $H$ is exactly $r$ when the radius $r$ gets accepted. However, in this case, if the radius $r$ is increased, $x$ is put into $B$ regardless of subsequent random choices. Consequently, for a fixed vertex $x \in V(G)$, the probability that $x \notin B$ is at most $p$. By union bound, the probability that $X \not\subseteq B$ is at most $|X|p$, which is at most $\frac{1}{2k}$ for $|X| \leq k$.
\end{proof}

\myparagraph{Standard divide\&conquer step.}
We would like to apply a similar divide\&conquer step as in the
presented proof that planar graphs have locally bounded treewidth. The problem is
that we can only afford a separator $W$ of size $\Ohtilde(\sqrt{k})$, however the radius of the graph can be much larger.

Let us define the {\em{margin}} $M$ to be the set of vertices within distance
at most $2000 \sqrt{k} \lg k = \Ohtilde(\sqrt{k})$ from any light terminal.
Intuitively, our case should be easy if every vertex of the pattern $X$
is in the margin: we could then just throw away all vertices of $G-M$
and use locally bounded treewidth, as the light terminals lie on the outer face. However, we cannot just branch (guess)
whether this is the case: the information that $G-M$ contains a vertex of the pattern is not directly useful.

Instead, we make a localized analog of this guess: we identify a relatively compact set of vertices of $G-M$ that prohibit us from making
a single step of the recursion sketched above.
First, we apply the clustering procedure 
(Theorem~\ref{thm:clust-over})
to the graph $G-M$, so that we can assume that every connected component of $G-M$,
henceforth called \emph{an island}, is of radius
bounded polynomially in $k$ and $\lg n$. Second, we construct an auxiliary
graph $H$ by contracting every island $C$ into a single vertex $u_C$.
Note that now in $H$ every vertex is within distance at most 
$2000\sqrt{k} \lg k + 1 = \Ohtilde(\sqrt{k})$ from a light terminal.
Thus $H$ has treewidth $\Ohtilde(\sqrt{k})$.
By standard arguments,
we can find a balanced separator $W_H$ in $H$, that is, a separator
of size $\Ohtilde(\sqrt{k})$ such that every connected component of $H-W_H$, after lifting it back to $G$ by reversing contractions,
contains (a) at most $|\trms|/2$ terminals from $G$, and 
(b) at most $|V(G) \setminus \trms|/2$ non-terminal vertices of $G$.

The separator $W_H$ can be similarly lifted to a separator $W$ in $G$ that corresponds
to $\Ohtilde(\sqrt{k})$ vertices of $M$ and $\Ohtilde(\sqrt{k})$ islands.
Now it is useful to make a guess if some vertex of an island in $W$
(i.e., a vertex of $W \setminus M$) belongs to the solution:
if this is not the case, we can delete the whole $W \setminus M$ from the graph, and apply the procedure recursively 
to connected components of $G-W$; if this is the case, the area to search
for such a vertex of the pattern is limited to $\Ohtilde(\sqrt{k})$ components
of radius polynomial in $k$ and $\lg n$.
Therefore, with some probability $q$ we decide to assume that $W\setminus M$ contains
a vertex of the pattern, and with the remaining probability $1-q$ we decide
that this is not the case. In the latter case, we remove $W \setminus M$ from the
graph, and recurse using $W \cap M$, which has size $\Ohtilde(\sqrt{k})$, as a separator; see Fig.~\ref{fig-over:std-sep}.
 The fact that
every connected component of $G-W$ contains at most $|\trms|/2$ terminals
allows us to keep the invariant that $|\trms| = \Ohtilde(\sqrt{k})$.

Let us now analyze what probability $q$ we can afford.
Observe that in every subinstance solved recursively, the number of nonterminal
vertices is halved. Thus, every vertex $x$ of the pattern $X$ is contained in $G$
only in $\Oh(\lg n)$ subinstances in the whole recursion tree; here we exclude the subinstances where $x$ is a light terminal, because
then its treatment is determined by the output specification of the recursive procedure. Consequently, we care about correct
choices only in $\Oh(k \lg n)$ steps of the recursion. 
In these steps, we want not to make a mistake during the clustering procedure
($1/k$ failure probability) and to correctly guess that $W \setminus M$
is disjoint with the pattern, provided this is actually the case ($q$ failure probability).
Thus, if we put $q = 1/k$, then the probability that we succeed in all $\Oh(k \lg n)$ steps we care about is inverse-polynomial in $n$; this is sufficient for our~needs.

\myparagraph{Island with a vertex of the pattern.}
We are left with the second case, where some island $C \subseteq W$
intersects the pattern. We have
$q = 1/k$ probability of guessing correctly that this is the case,
and independently we have $(1-1/k)$ probability
of not making a mistake in the clustering step.

The bound on the radii of the islands, as well as the fact that only
$\Ohtilde(\sqrt{k})$ islands are contained in $W$, allow us to localize
this vertex of the pattern even closer.
Recall that the radius of each island is bounded by $9k^2\lg n$.
For the rest of this overview we assume that $\lg n$ is bounded polynomially in $k$, and hence the radius of each island is polynomial in $k$.
Intuitively, this is because if, say, we had $\lg n > 100 \cdot k^{100}$, then $n > 2^{100 k^{100}}$ and
with $2^{100 k^{100}}$ allowed factor in the running time bound we can apply
a variety of other algorithmic techniques. 
More formally, we observe that $(\lg n)^{\Ohtilde(\sqrt{k})}$ is bounded by $2^{\Ohtilde(\sqrt{k})}\cdot n^{o(1)}$, which is sufficient 
to make sure that all the experiments whose success probability depend on $\lg n$ succeed simultaneously with probability at least $(2^{\Ohtilde(\sqrt{k})}\cdot n^{o(1)})^{-1}$.

We first guess (by sampling at random) an island $C \subseteq W$ that contains a vertex of the pattern.
Then, we pick an arbitrary vertex
$z \in C$, and guess (by sampling at random) the distance $d$ in $C$ between $z$ and the closest vertex
of the pattern in $C$. By contracting all vertices within distance less than $d$
from $z$, with success
probability inverse-polynomial in $k$, we arrive at the following situation:
(see Fig.~\ref{fig-over:stretched})
\begin{equation*}
\textrm{we have a vertex }z \notin M\textrm{ such that either }z\textrm{ or a neighbor of }z\textrm{ belongs to the pattern }X.
\end{equation*}

\myparagraph{Chain of separators.}
Hence, one of the components of $G[X]$ is stretched across the margin $M$,
between a light terminal on the outer face and the vertex $z$ inside the margin.
Our idea now is to use this information to cut $X$ in a balanced fashion.
Note that we have already introduced an inverse-polynomial in $k$ multiplicative factor in the success probability.
Hence, to maintain the overall inverse-subexponential dependency on $k$ in the success
probability, we should aim at a progress that will allow us
to bound the number of such steps by $\Ohtilde(\sqrt{k})$.

Unfortunately, it is not obvious how to find such a separation.
It is naive to hope for a $z$-$\light$ separator of size $\Ohtilde(\sqrt{k})$,
and a larger separator seems useless, if there is only one. However, we can aim at a Baker-style argument: if we find a chain of $p$ pairwise disjoint $z$-$\light$ separators $C_1,C_2,\ldots,C_p$ 
(see Fig.~\ref{fig-over:cycles}), we could guess a ``sparse one'', and separate along it.
Since the separators are pairwise disjoint, there exists a ``sparse'' separator $C_i$ with at most $k/p$ vertices of the pattern. 
On the other hand, since the pattern contains a component stretched from $z$ to a light terminal, every $C_i$ intersects the pattern. If we omit the first and the last $p/4$ separators, and look for a sparse separator in between with at most $2k/p$
by means of guessing only $2k/p$ vertices of $X$.
If all separators are of size bounded polynomially in $k$,
the optimal choice is $p \sim k^{2/3}$, which 
leads to success probability inverse in~$2^{\Ohtilde(k^{2/3})}$.

However, we can apply a bit smarter counting argument. Take $p = 120 \sqrt{k} \lg k$.
Look at $C_{p/2}$ and assume that at most half of the vertices of $X$
lie on the side of $C_{p/2}$ with separators $C_i$ for $i < p/2$;
the other case is symmetric.
The crucial observation is the following:
there exists an index $i \leq p/2$ such that
if $|C_i \cap X| = \alpha$ then $C_i$ partitions $X$ into two parts
of size at least $\alpha \sqrt{k}/10$ each. Indeed, otherwise
we have that for every $i \leq p/2$ it holds that
$$|X \cap C_i| \geq \frac{10}{\sqrt{k}} \cdot \sum_{j < i} |X \cap C_j|.$$
This implies $|X \cap \bigcup_{j \leq i} C_j| \geq (1+10/\sqrt{k})^i$, and $|X \cap \bigcup_{j \leq p/2} C_{j}| > k$
for $p = 120\sqrt{k}\lg k$.

Hence, we guess (by sampling at random) such an index $i$, the value of $\alpha = |X \cap C_i|$, and
the set $X \cap C_i$. If the size of $C_i$ is bounded polynomially in $k$,
with success probability $k^{-\Oh(\alpha)}$ we partition the pattern
into two parts of size at least $\alpha \sqrt{k}/10$ each.
A simple amortization argument shows that all these guessings incur only 
the promised $2^{-\Ohtilde(\sqrt{k})}$ multiplicative factor in the overall success probability.
Furthermore, as such a step creates $\alpha$ heavy terminals, it can be easily seen that the total number
of heavy terminals will never grow beyond $\Ohtilde(\sqrt{k})$.

However, the above argumentation assumes we are given such a chain of separators
$C_i$: they are not only pairwise disjoint, but also of size polynomial in $k$.
Let us now inspect how to find them.

\myparagraph{Duality.}
In the warm-up proof of planar graphs having locally bounded treewidth, 
the separator $W$ is obtained from the classic Menger's maximum flow/minimum cut duality. 
Here, we aim at a chain of separators, but we require that their sizes
are polynomial in $k$. It turns out that we can find such a chain by formulating a maximum flow of minimum cost problem, 
and extracting the separator chain in question from the optimum solution to its (LP) dual.

\begin{theorem}\label{thm:duality-over}
There is a polynomial-time algorithm that given
a connected graph $G$, a pair $s,t\in V(G)$ of different vertices, and positive integers $p,q$,
outputs one of the following structures in $G$:
\begin{enumerate}[(a)]
\item A chain $(C_1,\ldots,C_p)$ of $(s,t)$-separators with $|C_j|\leq 2q$ for each $j\in [p]$.
\item A sequence $(P_1,\ldots,P_q)$ of $(s,t)$-paths with $|(V(P_i)\cap \bigcup_{i'\neq i} V(P_{i'}))\setminus \{s,t\}|\leq 4p$ for each~$i\in [q]$.
\end{enumerate}
\end{theorem}
\begin{proof}[Proof sketch.]
We formulate the second outcome as a maximum flow of minimum cost problem in an auxiliary
graph, where every vertex $v \in V(G) \setminus \{s,t\}$ is duplicated into two copies: one of cost $0$ and capacity $1$, and one of cost $1$ and infinite capacity. We ask for a minimum-cost flow of size $2q$ from $s$ to $t$.
If the cost of such flow is at most $2pq$, the projection onto $G$ of the
cheapest $q$ flow paths gives the second output. Otherwise, we read the desired chain of separators $(C_1,\ldots,C_p)$ as distance layers from $s$ in the graph with distances
imposed by the solution to the dual linear program.
\end{proof}

Since all light terminals lie on the outer face, we can attach an auxiliary root
vertex $r_0$ adjacent to all light terminals, and apply Theorem~\ref{thm:duality-over} to $(s,t) = (r_0,z)$, $p = 120\sqrt{k} \lg k$ and $q = \textrm{poly}(k)$.
If the algorithm of Theorem~\ref{thm:duality-over} returns a chain of separators,
we proceed as described before. Thus, we are left with the second output:
$q = \textrm{poly}(k)$ nearly-disjoint paths from $\light$ to $z$.

\myparagraph{Nearly-disjoint paths.}
The vertex set of every path $P_i$ can be partitioned into \emph{public} vertices $\Pub(P_i)$, the ones used also by other paths, and the remaining \emph{private} vertices $\Prv(P_i)$. We have $|\Pub(P_i)| \leq 4p = 480\sqrt{k}\lg k$,
and the sets $\Prv(P_i)$ are pairwise disjoint.
We can assume $q > k$, so there exists a path $P_i$ such that $\Prv(P_i)$
is disjoint with the pattern $X$.
By incurring an additional $1/k$ multiplicative factor in the success probability,
we can guess, by sampling at random, such~index~$i$.

How can we use such a path $P_i$? Clearly, we can delete the private vertices
of $P_i$, because they can be assumed not to be used by the patter $X$. 
However, we choose a different way: we contract them onto neighboring public vertices along $P_i$, reducing $P_i$ to a path with vertex set $\Pub(P_i)$. 
Observe that by this operation the vertex $z$ changes its location in $G$:
from a vertex deeply inside $G$, namely not within the margin $M$, it is moved
to a place within distance $|\Pub(P_i)| \leq 480\sqrt{k} \lg k$ from the light terminals,
which is less than a quarter of the width of the margin; see Fig.~\ref{fig-over:contract}.

Furthermore, by the connectivity assumptions on the pattern $X$, the vertex $z$
drags along a number of vertices of $X$ that are close to it. More precisely, if $Q$
is a path in $G[X]$ connecting $z$ or a neighbor of $z$ with a light terminal, then
the first $500 \sqrt{k}\lg k$ vertices on $Q$ are moved from being within distance
at least $1500\sqrt{k}\lg k$ from all light terminals, to being within distance
at most $1000 \sqrt{k}\lg k$ from some light terminal. 
Hence, if we define that a vertex $x \in X$ is \emph{far} if it is within
distance larger than $1000\sqrt{k}\lg k$ (i.e., half of the width of the margin)
from all light terminals, and \emph{close} otherwise, then
by contracting the private vertices of $P_i$ as described above, at least $500\sqrt{k}\lg k$
vertices of $k$ change their status from far to close. 

By a careful implementation of all separation steps, we can ensure that no close vertex of $X$ becomes far again. 
Consequently, we ensure that the above step can happen only $\Ohtilde(\sqrt{k})$ times. 
Since the probability of succeeding in all guessings within this step is inverse-polynomial in $k$, this incurs only a $2^{-\Ohtilde(\sqrt{k})}$ multiplicative factor in the overall success probability.

This finishes the overview of the proof of Theorem~\ref{thm:maintheorem}. We invite the reader to the next sections for a fully formal proof, which is moreover conducted for an arbitrary apex-minor-free class.

\section{Preliminaries}\label{sec:prelims}

\myparagraph{Notation.} We use standard graph notation; see, e.g.,~\cite{cygan2015parameterized} for a reference.
All graphs considered in this paper are undirected and simple (without loops or multiple edges), unless explicitly stated.
For a vertex $u$ of a graph $G$, by $N_G(u)=\{v\colon uv\in E(G)\}$ and $N_G[u]=\{u\}\cup N_G(u)$ we denote the open and closed neighborhood of $u$, respectively.
Similarly, for a vertex subset $X\subseteq V(G)$, by $N_G[X]=\bigcup_{u\in X} N_G[u]$ and $N_G(X)=N_G[X]\setminus X$ we denote the closed and open neighborhood of $X$, respectively.
The subscript is dropped whenever it is clear from the context.

For an undirected graph $G$ and an edge $uv\in E(G)$, by {\em{contracting}} $uv$ we mean the following operation:
remove $u$ and $v$ from the graph, and replace them with a new vertex that is adjacent to exactly those vertices that were neighbors of $u$ or $v$ in $G$.
Note that this definition preserves the simplicity of the graph. 
By {\em{contracting $v$ onto $u$}} we mean the operation of contracting the edge $uv$ and renaming the obtained vertex as $u$.
More generally, if $X$ is a subset of vertices with $G[X]$ being connected, and $u\notin X$ is such that $u$ has a neighbor in $X$, 
then by {\em{contracting $X$ onto $u$}} we mean the operation of exhaustively contracting a neighbor of $u$ in $X$ onto $u$ up to the point when $X$ becomes empty.
Note that due to the connectivity of $G[X]$ such outcome will always be achieved.

We say that $H$ is a {\em{minor}} of $G$ if $H$ can be obtained from $G$ by means of vertex deletions, edge deletions, and edge contractions.
An {\em{apex graph}} is a graph that can be made planar by removing one of its vertices.

For a positive integer $k$, we denote $[k]=\{1,\ldots,k\}$. We denote $\lg x = \log_2 x$. Notation $\log$ is used only under the $\Oh(\cdot)$-notation, 
where multiplicative constants are hidden anyways. We also denote $\exp[t]=e^t$, where $e$ is the base of the natural logarithm. 

\myparagraph{Tree decompositions.} Let $G$ be an undirected graph. A {\em{tree decomposition}} $\Tt$ of $G$ is a rooted tree $T$ with a {\em{bag}} $\beta(x)\subseteq V(G)$ associated with each its node $x$,
which satisfies the following conditions:
\begin{enumerate}[(T1)]
\item For each $u\in V(G)$ there is some $x\in V(T)$ with $u\in \beta(x)$.
\item For each $uv\in E(G)$ there is some $x\in V(T)$ with $\{u,v\}\subseteq \beta(x)$.
\item For each $u\in V(G)$, the node subset $\{x\in V(T)\colon u\in \beta(x)\}$ induces a connected subtree of $T$.
\end{enumerate}
The {\em{width}} of a tree decomposition $\Tt$ is $\max_{x\in V(T)} |\beta(x)|-1$, and the {\em{treewidth}} of $G$ is equal to the minimum possible width of a tree decomposition of $G$.
We assume reader's familiarity with basic combinatorics of tree decompositions, and hence we often omit a formal verification that some constructed object is indeed a tree decomposition of some graph.

The following fact about the existence of balanced separators in graphs of bounded treewidth is well-known, see e.g.~\cite[Lemma 7.19]{cygan2015parameterized}.

\begin{lemma}\label{lem:balsep}
Let $G$ be a graph of treewidth at most $t$ and let $\mathbf{w}\colon V(G)\to \mathbb{R}_{\geq 0}$ be a nonnegative weight function on vertices of $G$. 
Then there exists a subset of vertices $X\subseteq V(G)$ of size at most $t+1$ such that for every connected component $C$ of $G-X$ the following holds:
$$\mathbf{w}(V(C))\leq \mathbf{w}(V(G))/2,$$
where $\mathbf{w}(A)=\sum_{u\in A}\mathbf{w}(u)$ for any vertex subset $A$.
\end{lemma}

\myparagraph{Locally bounded treewidth and apex-minor-freeness.} For the whole proof we fix a class $\Cc$ of graphs that satisfies the following properties:
\begin{enumerate}[(a)]
\item\label{pr:minor-closed} $\Cc$ is closed under taking minors.
\item\label{pr:loc-bnd-tw} $\Cc$ has locally bounded treewidth. That is, there exists a function $f\colon \Nats\to \Nats$ such that any connected graph $G$ from $\Cc$ of radius $r$ has treewidth bounded by $f(r)$.
\end{enumerate}
Demaine and Hajiaghayi~\cite{DemaineH04} proved that any class $\Cc$ that satisfies properties~\eqref{pr:minor-closed} and~\eqref{pr:loc-bnd-tw} is actually apex-minor-free. 
More precisely, there exists some apex graph $H$ such that no graph from $\Cc$ admits $H$ as a minor.
However, Grohe~\cite{Grohe03} proved that every apex-minor-free class of graphs has locally bounded treewidth, where moreover the function $f$ is linear.
As observed by Demaine and Hajiaghayi~\cite{DemaineH04}, these fact combined yield the following.

\begin{proposition}\label{prop:local-tw}
If a class of graphs $\Cc$ satisfies properties~\eqref{pr:minor-closed} and~\eqref{pr:loc-bnd-tw}, then there is a constant $\ctw$ such that
for every connected graph $G$ from $\Cc$ of radius $r$, the treewidth of $G$ is bounded by $\ctw\cdot r$.
\end{proposition}

The results of Grohe~\cite{Grohe03} in particular show that if $\Cc$ is a class of graphs that exclude a fixed apex graph as a minor, 
then the closure of $\Cc$ under taking minors satisfies properties~\eqref{pr:minor-closed} and~\eqref{pr:loc-bnd-tw}, and hence also the property implied by Proposition~\ref{prop:local-tw}.
Hence, it suffices to prove Theorem~\ref{thm:maintheorem} for any graph class $\Cc$ that satisfies properties~\eqref{pr:minor-closed} and~\eqref{pr:loc-bnd-tw}.
From now on, we fix the constant $\ctw$ yielded by Proposition~\ref{prop:local-tw} for the class $\Cc$. To avoid corner cases, we assume without loss of generality that $\ctw\geq 10$.

\section{Auxiliary tools}\label{sec:aux-tools}

In this section we introduce auxiliary technical tools that will be needed in the proof: a clustering procedure that reduces the radius of the graph, and a duality result concerning almost disjoint paths and
chains of separators between a pair of vertices.

\subsection{Clustering procedure}\label{sec:clustering}

The following result was stated as Theorem~\ref{thm:clust-over} in Section~\ref{sec:overview}.

\begin{theorem}\label{thm:clust}
There exists a randomized polynomial-time algorithm that, given a graph $G$ on $n>1$ vertices and a positive integer $k$, returns a vertex subset $B\subseteq V(G)$ with the following properties:
\begin{enumerate}[(a)]
\item The radius of each connected component of $G[B]$ is less than $9k^2\lg n$.
\item For each vertex subset $X\subseteq V(G)$ of size at most $k$, the probability that $X \subseteq B$ is at least $1-\frac{1}{k}$.
\end{enumerate}
\end{theorem}
\begin{proof}
Consider the following iterative procedure which constructs a set $B_0$. We start with $V_0 = V(G)$.
In the step $i$, given a set $V_{i-1} \subseteq V(G)$, we terminate the procedure if $V_{i-1} = \emptyset$.
Otherwise we pick an arbitrary vertex $v_i \in V_{i-1}$ and we randomly select a radius $r_i$ according to the geometric distribution with success probability $p := \frac{1}{2k^2}$.
Given $v_i$ and $r_i$, we insert $\distball_{G[V_{i-1}]}(v_i,r_i-1)$ into $B_0$ and define $V_i := V_{i-1} \setminus \distball_{G[V_{i-1}]}(v_i,r_i)$.
That is, we delete from the graph vertices within distance at most $r_i$ from $v_i$ in $G[V_{i-1}]$, and put the vertices within distance \emph{less} than $r_i$ into $B_0$.
Finally, at the end of the procedure, if any of the selected radii $r_i$ is larger than $9k^2 \lg n$, we return $B = \emptyset$, and otherwise we return $B=B_0$.

Clearly, the procedure runs in polynomial time, as at every step at least the vertex $v_i$ is removed from $V_{i-1}$, and hence at most $n$ iterations are executed.
For the radii of the connected components of $G[B_0]$, note that the fact that we insert into $B_0$ vertices within distance \emph{less} than $r_i$ from $v_i$, but delete
from $V_i$ vertices within distance \emph{at most} $r_i$ from $v_i$, ensures that at every step $i$ of the iteration we have $N_G[B_0] \cap V_i = \emptyset$. Consequently,
$\distball_{G[V_{i-1}]}(v_i,r_i-1)$ induces a connected component of $G[B_0]$, and is of radius less than $r_i$. 
Since we return $B = \emptyset$ instead of $B=B_0$ if any of the selected radii $r_i$ exceeds $9k^2 \lg n$, the upper bound on the radii of the connected components of $G[B]$ follows.

It remains to argue that any fixed $k$-vertex subset $X \subseteq V(G)$ survives in $B$ with high probability. 
First, note that for fixed $i$ we have that 
$$\mathbb{P}(r_i > 9k^2 \lg n) \leq \left(1-p\right)^{9k^2 \lg n} \leq e^{-4.5 \lg n} < n^{-3} < \frac{1}{2kn};$$
here, we use the fact that $p=\frac{1}{2k^2}$ and the inequality $1-x\leq e^{-x}$.
Thus, as there are at most $n$ iterations, with probability less than $\frac{1}{2k}$ the algorithm returns $B = \emptyset$ because of some $r_i$ exceeding the limit of $9k^2 \lg n$.

Second, we analyze the probability that $X\subseteq B_0$. Let us fix some $x \in X$.
The only moment where the vertex $x$ could be deleted from the graph, but not put into $B_0$, is when $x$ is within distance exactly $r_i$ from the vertex $v_i$ in an iteration $i$.
It is now useful to think of the choice of $r_i$ in the iteration $i$ as follows: we start with $r_i := 1$ and then, iteratively, with probability $p$ accept the current radius,
and with probability $1-p$ increase the radius $r_i$ by one and repeat.
However, in the aforementioned interpretation of the geometric distribution, when $r_i = \distance_{G[V_{i-1}]}(v_i,x)$, with probability $p$ the radius $r_i$ is accepted
(and $x$ is deleted but not put in $B_0$), but with probability $(1-p)$ the radius $r_i$ is increased, and the vertex $x$ is included in the ball
$\distball_{G[V_{i-1}]}(v_i,r_i-1) \subseteq B_0$. Consequently, the probability that a fixed vertex $x \in X$ is not put into $B_0$ is at most $p$.
By the union bound, we infer that the probability that $X \not\subseteq B_0$ is at most $kp = \frac{1}{2k}$. Together with the $\frac{1}{2k}$ upper bound on the probability
that the maximum radius among $r_i$ exceeds $9k^2\lg n$, which results in putting $B=\emptyset$ instead of $B=B_0$, we have that the probability that $X\not\subseteq B$ is at most $\frac{1}{k}$.
This concludes the proof.
\end{proof}

\subsection{Duality}\label{sec:duality}

We start with a few standard definitions.

\begin{definition}
For a graph $H$ and its vertex $u$, by $\reach(u,H)$ we denote the set of vertices of $H$ reachable from $u$ in $H$.
Suppose $G$ is a connected graph, and $s,t$ are its different vertices.
An {\em{$(s,t)$-separator}} is a subset $C$ of vertices of $G$ such that $s,t \notin C$ and $t\notin \reach(s,G-C)$.
An $(s,t)$-separator is {\em{minimal}} if no its proper subset is also an $(s,t)$-separator.
\end{definition}

\begin{definition}
A sequence $(C_1,C_2,\ldots,C_k)$ of minimal $(s,t)$-separators is called an {\em{$(s,t)$-separator chain}} if all of them are pairwise disjoint and for each $1\leq j<j'\leq k$, the following holds:
$$C_j\subseteq \reach(s,G-C_{j'})\qquad\textrm{and}\qquad C_{j'}\subseteq \reach(t,G-C_j).$$
\end{definition}

We now state and prove the main duality result, that is, Theorem~\ref{thm:duality-over} from Section~\ref{sec:overview}.

\begin{theorem}\label{thm:duality}
There is a polynomial-time algorithm that given
a connected graph $G$, a pair $s,t\in V(G)$ of different vertices, and positive integers $p,q$,
outputs one of the following structures in $G$:
\begin{enumerate}[(a)]
\item An $(s,t)$-separator chain $(C_1,\ldots,C_p)$ with $|C_j|\leq 2q$ for each $j\in [p]$.
\item A sequence $(P_1,\ldots,P_q)$ of $(s,t)$-paths with $|(V(P_i)\cap \bigcup_{i'\neq i} V(P_{i'}))\setminus \{s,t\}|\leq 4p$ for each~$i\in [q]$.
\end{enumerate}
\end{theorem}

\begin{proof}
Our approach is as follows: we formulate the second output as a min-cost max-flow problem in an auxiliary graph $H$. If the cost of the computed flow
is not too large, a simple averaging argument yields the desired paths $P_i$ from the flow paths. If the cost is large, 
we look at the dual of the min-cost max-flow problem, expressed as a linear program, which is in fact a distance LP. Then we read the separators
$C_i$ as layers of distance from the vertex $s$. Let us now proceed with formal argumentation.

We define a graph $H$ as follows. Starting with $H := G$, we replace every vertex $v \in V(G) \setminus \{s,t\}$ with two copies $v_0$ and $v_1$:
the copy $v_0$ has capacity $1$ and cost $0$, while the copy $v_1$ has infinite capacity and cost $1$. The vertex $s$ is a source of size $2q$ and
cost $0$, and the vertex $t$ is a sink of size $2q$ and cost $0$. The edges of $H$ are defined naturally: every edge $uv$ of $G$ gives rise to up to four
edges in $H$, between the copies of $u$ and the copies of $v$. 

In the graph $H$, we ask for a minimum-cost vertex-capacitated flow from $s$ to $t$ of size $2q$. Clearly, such a flow exists for connected graphs $G$,
as every vertex $v_1$ is of infinite capacity. Since all the costs and capacities are integral or infinite, in polynomial time we can find a minimum-cost
solution that decomposes into $2q$ flow paths $P_1',P_2',\ldots,P_{2q}'$ that carry unit flow each.  Let $C$ be the total cost of this flow.
Every path $P_i'$ induces a walk $P_i$ in $G$: whenever $P_i'$ traverses a vertex $v_0$ or $v_1$, the path $P_i$ traverses the corresponding vertex $v$.
By shortcutting if necessary, we may assume that each $P_i$ is a path.

In this proof, we consider every path $P$ from $s$ to $t$ (either in $G$ or in $H$) as oriented from $s$ towards $t$;
thus, the notions of a predecessor/successor on $P$ or the relation of lying before/after on $P$ are well-defined.

Let us define the cost of a path $P_i$, denoted $c(P_i)$, as the number of internal vertices that $P_i$ shares with other paths. That is,
\begin{eqnarray*}
c(P_i)& := & |(V(P_i) \cap \bigcup_{j \neq i} V(P_j)) \setminus \{s,t\}|.
\end{eqnarray*} 
Observe that, due to the capacity constraints, if a vertex $v \notin \{s,t\}$
lies on $h > 1$ paths $P_{i_1},P_{i_2},\ldots,P_{i_h}$, then all but one of the paths $P_{i_j}'$ have to use the vertex $v_1$, inducing total cost
$h-1$ for the minimum-cost flow. 
Since $h-1\geq h/2$ for $h>1$, it follows that $$\sum_{i=1}^{2q} c(P_i) \leq 2C.$$
We infer that if $C \leq 2pq$, then $\sum_{i=1}^{2q} c(P_i) \leq 4pq$. A standard averaging argument now shows that
for at least $q$ paths $P_i$ we have $c(P_i) \leq 4p$. This yields the second desired output.

Thus we are left with the case when $C > 2pq$. Our goal is to find a separator chain suitable as the first desired output.
To this end, we formulate the discussed minimum-cost flow problem as a linear program, and we analyze its dual.
The precise formulations can be found in Figures~\ref{fig:duality-lp1} and~\ref{fig:duality-lp2}.

\begin{figure}[tb]
\begin{center}
\begin{minipage}{0.9\textwidth}
\begin{align*}
	\quad \min \quad & \sum_{v \in V(G) \setminus \{s,t\}}\, \sum_{a \in N_H(v_1)} f(v_1,a) &\\
	\text{s.t.} \quad & \sum_{b \in N_H(a)} f(a,b) - f(b,a) = 0 & \forall\ a \in V(H) \setminus \{s,t\}\\
	& \sum_{a \in N_H(s)} f(s,a)-f(a,s) = 2q & \\
	& \sum_{a \in N_H(t)} f(t,a)-f(a,t) = -2q & \\
	& \sum_{a \in N_H(v_0)} f(v_0,a) \leq 1 & \forall\ v \in V(G) \setminus \{s,t\} \\
        & f(a,b) \geq 0 & \forall\ ab \in E(H)\\
\end{align*}
\end{minipage}
\end{center}
\caption{The minimum-cost flow problem used in the proof of Theorem~\ref{thm:duality}.
\label{fig:duality-lp1}
In the flow problem, the variables $f(a,b)$ correspond to the amount of flow pushed from $a$ to $b$ along an edge $ab$.}
\end{figure}

\begin{figure}[tb]
\begin{center}
\begin{minipage}{0.9\textwidth}
\begin{align*}
	\quad \max \quad & 2q(y_t-y_s) - \sum_{v \in V(G) \setminus \{s,t\}} z_v &\\
	\text{s.t.} \quad & y_{v_0} \leq y_a + z_v & \forall\  v \in V(G) \setminus \{s,t\}, a \in N_H(v_0) \\
        & y_{v_1} \leq y_a + 1 & \forall\ v \in V(G) \setminus \{s,t\}, a \in N_H(v_1) \\
        & y_s \leq y_a & \forall\ a \in N_H(s) \\
        & y_t \leq y_a & \forall\ a \in N_H(t) \\
        & z_v \geq 0 & \forall\ v \in V(G) \setminus \{s,t\}\\
\end{align*}
\end{minipage}
\end{center}
\caption{The dual of the minimum-cost flow problem from Figure~\ref{fig:duality-lp1}, used in the proof of Theorem~\ref{thm:duality}.\label{fig:duality-lp2}}
\end{figure}

In the dual formulation, the value $y_a-y_s$ can be interpreted as a distance of $a$ from $s$, where 
traveling through a vertex $v_1$ costs $1$ and traveling through a vertex $v_0$ costs $z_v$.
The goal is to maximize the distance from $s$ to $t$ with weight $2q$,
while paying as little as possible in the sum $\sum_{v \in V(G) \setminus \{s,t\}} z_v$.

Let $\{z_v\colon v\in V(G);\ y_a\colon a\in V(H)\}$ be an optimum solution to the dual LP. Since the primal program is a minimum-cost flow
problem with integral coefficients, in polynomial time we can find such values $z_v$ and $y_a$ that are additionally integral.
Observe that the dual is invariant under adding a constant to every variable $y_a$; hence, we can assume $y_s = 0$.
Since traveling through a vertex $v_1$ incurs distance $1$, and $v_0$ is a twin of $v_1$, the optimum solution
never uses values $z_v$ greater than $1$; hence, $z_v \in \{0,1\}$ for every $v \in V(G) \setminus \{s,t\}$.

On one hand, $C$, as the optimum value of both the primal and the dual LP, is assumed to be larger than $2pq$.
On the other hand, we have that $z_v \geq 0$ for every $v \in V(G) \setminus \{s,t\}$.
We infer that
$$2q(y_t - y_s) \geq C > 2pq,$$
and hence $y_t > p$. We define for every $1 \leq j \leq p$ the set
$$C_j = \{v \in V(G): z_v = 1 \wedge y_{v_0} = j\}.$$
We claim that $C_1,C_2,\ldots,C_p$ is the desired separator chain.
Clearly, the sets are pairwise disjoint and do not contain neither $s$ nor $t$.
We now show that they form a separator chain.

\begin{claim}
For each $1\leq j\leq p$, the set $C_j$ is an $(s,t)$-separator.
\end{claim}
\begin{proof}
Consider a path $P$ from $s$ to $t$.
Let $P'$ be the corresponding path in $H$ that traverses a vertex $v_0$ whenever $v \in V(G) \setminus \{s,t\}$ lies on $P$.
Since $z_v \in \{0,1\}$ for every $v \in V(G) \setminus \{s,t\}$, we have $y_b \leq y_a+1$ for every $ab \in E(H)$.
As $y_t > p$, there exists a vertex $b$ on $P'$ with $y_b = j$; let $b$ be the first such vertex and let $a$ be its predecessor on $P'$.
Note that $b \neq s$ as $y_s=0$ and $b \neq t$ as $y_t > p$, hence $b = v_0$ for some $v \in V(G) \setminus \{s,t\}$.
Since $b$ is the first vertex on $P'$ with $y_b = j$, we have $y_a = j-1$, and, consequently, $z_v = 1$.
Thus $v \in C_j$. Since the choice of $P$ is arbitrary, $C_j$ is an $(s,t)$-separator, as desired.
\cqed\end{proof}

Consider a path $P_i'$. By complementary slackness conditions, whenever the path $P_i'$ traverses an edge $ab \in E(H)$ from $a$ to $b$, the corresponding
distance inequality of the dual LP is tight: $y_b = y_a + 1$ if $b = v_1$ for some $v$, $y_b = y_a + z_v$ if $b = v_0$, and
$y_b = y_a$ if $b = s$ or $b=t$. Thus, $P_i'$ is a shortest path from $s$ to $t$ in the graph $H$ with vertex weights $0$ for $s,t$, $z_v$ for every $v_0$
and $1$ for every $v_1$. In particular, always $y_a \leq y_b \leq y_a+1$ for $b$ being a successor of $a$ on $P_i'$.

\begin{claim}
For each $1\leq j\leq p$ and each $1\leq i\leq 2q$, we have that $|V(P_i) \cap V(C_j)| = 1$.
\end{claim}
\begin{proof}
We first prove that the cardinality of this intersection is at most $1$.
Assume $b \in P_i'$ such that $b = v_0$ or $b = v_1$ for some $v \in C_j$ and let $a$ be the predecessor of $b$ on $P_i'$.
Since $z_v = 1$, we have $y_b = y_a + 1 = j$, that is, $b$ must be the first vertex on $P_i'$
with $y_b = j$. Consequently, there exists at most one vertex $b$ on $P_i'$ that projects to a vertex of $C_j$ in $G$, which proves that $|V(P_i) \cap V(C_j)| \leq 1$.

To prove the converse inequality, we show that the projection onto $G$ of the first vertex $b$ on $P_i'$ such that $y_b = j$
belongs to $C_j$. First, observe that such a vertex $b$ exists since $y_s = 0$, $y_t > p$, and $y_b \leq y_a + 1$ for every $ab \in E(H)$.
Let $b = v_0$ or $b=v_1$ for some $v \in V(G) \setminus \{s,t\}$; we claim that $z_v = 1$, $y_{v_0} = j$, and hence $v \in C_j$.
If $b = v_0$, the claim is immediate as $y_b = y_a + 1$ for the predecessor $a$ of $b$ on $P_i'$.
By contradiction, let us assume $b=v_1$ but $z_v = 0$.
Then by replacing $b=v_1$ with $v_0$ on $P_i'$ we obtain a shorter path from $s$ to $t$ in $H$, contradiction to the fact
that $P_i'$ is a shortest path from $s$ to $t$.
Thus $z_v = 1$ and, consequently, $y_{v_0} = y_{v_1} = y_a+1$, so $v$ belongs to $C_j$, as claimed.
\cqed\end{proof}

Let us denote the unique vertex of $V(P_i) \cap V(C_j)$ as $w_{i,j}$.
Note that $w_{i,j}$ and $w_{i',j}$ may coincide for different indices $i,i'$.
By the complementary slackness conditions again, if $z_v = 1$ for a vertex $v \in V(G) \setminus \{s,t\}$, then there exists a flow path $P_i'$ that passes through
$v_0$. Consequently, for every $1 \leq j \leq p$ and every $v \in C_j$, there exists a flow path $P_i'$ that passes through $v_0$. It follows that
$v = w_{i,j}$, and thus $C_j = \{w_{i,j} : 1 \leq i \leq 2q\}$. In particular, we have $|C_j| \leq 2q$ for every $1 \leq j \leq p$. 
Moreover, since each vertex of $C_j$ lies on some path $P_i$, and is the unique vertex of $V(P_i) \cap V(C_j)$, we infer the $C_j$ is a minimal separator.

We are left with verifying the inclusion of reachability sets.
Since $P_i'$ is a shortest path from $s$ to $t$ in $H$, we have that the vertex $w_{i,j}$ lies before the vertex $w_{i,j'}$ on the path $P_i$, whenever $j < j'$. 
As the separators $C_j$ and $C_{j'}$ respectively consist only of vertices $w_{i,j}$ and $w_{i,j'}$, already the reachability within paths $P_i$ certifies that
that $C_j\subseteq \reach(s,G-C_{j'})$ and $C_{j'}\subseteq \reach(t,G-C_j)$ for every $1 \leq j < j' \leq p$, as requested.
This concludes the proof of the theorem.
\end{proof}

\section{Proof of the main result}\label{sec:proof}

In this section we give a formal proof of our main result, i.e. Theorem~\ref{thm:maintheorem}. 
Throughout the proof we assume that $k\geq \max(10,2^{\ctw})$, because otherwise the result is trivial:
as $k$ is bounded by a constant, we can just sample $k$ vertices of the graph uniformly and independently at random, and return them as $A$. Hence, we may assume that $k\geq 10$ and $\lg k\geq \ctw$.

\subsection{Recursive scheme and potentials}\label{ss:proof:recur}

Let $G_0\in \Cc$ be the input graph, and let $k$ be the requested upper bound on the sizes of patterns $X$ that we need to cover.
The algorithm constructs the set $A$ by means of a recursive procedure that roughly partitions the graph into smaller and smaller pieces, at each point making some random decisions.
For the analysis, we fix some pattern $X$, that is, a subset $X$ of vertices such that $G[X]$ is connected and $|X|\leq k$.
Recall that our goal is to construct $A$ in such a manner that the probability that $X$ is covered by $A$ is at least the inverse of $2^{\Oh(\sqrt{k}\log^2 k)}\cdot n^{\Oh(1)}$.
The steps taken by the algorithm obviously will not depend on $X$, 
but at each random step we argue about the {\em{success probability}}: the probability that the taken decision is
compliant with the target pattern $X$, i.e., leads to its coverage.

As usual with recursive algorithms, we need to consider a more general problem, which will be supplied with a few potential measures.
Formally, an instance of the general problem is a tuple consisting of:
\begin{enumerate}[(i)]
\item A connected graph $G$ that is a minor of the original graph $G_0$; hence in particular $G\in \Cc$.
\item A specified vertex $r\in V(G)$ called the {\em{root}}.
\item Two disjoint vertex subsets $\light,\heavy\subseteq V(G)$, called {\em{light terminals}} and {\em{heavy terminals}}, respectively. 
We require that the root vertex is a light terminal, that is, $r\in \light$. By $\trms:=\light\cup \heavy$ we will denote the sets of {\em{terminals}}.
\item A subset of non-terminals $\ghost\subseteq V(G)\setminus \trms$, called {\em{ghost vertices}}.
\item A nonnegative integer $\lambda$, called the {\em{credit}}.
\end{enumerate}
Intuitively, terminals represent the boundary via which the currently considered piece communicates with the rest of the original graph, whereas
ghost vertices represent maximal connected parts of the original graph lying outside of the currently considered piece, each contracted to one vertex.
The reader can think of ghost vertices as emulation of hyperedges on their neighborhoods, or rather as emulation of the torso operation that would turn their neighborhoods into cliques.
This torso operation cannot be performed directly because of the necessity of staying within class $\Cc$, and therefore we resort to introducing ghost vertices.

In the course of the algorithm, we shall maintain the following invariants:
\begin{enumerate}[(a)]
\item\label{inv:dist} Every light terminal is at distance at most $3$ from the root.
\item\label{inv:terminals} It holds that $|\trms|\leq 16014\ctw\sqrt{k}\lg k+\lambda$.
\end{enumerate}
We say that a subset $X\subseteq V(G)\setminus \ghost$ is a {\em{pattern}} in instance $\Ii$ if 
\begin{enumerate}[(i)]
\item The root $r$ is contained in $X$;
\item Every vertex of $X$ can be reached from $r$ by a path that traverses only vertices of $X\cup \ghost$.
\item $|X|\leq k-10\sqrt{k}\cdot \lambda$.
\end{enumerate}
In particular, every pattern has at most $k$ vertices, but we may consider only smaller patterns if the credit is positive.
Note that the ghost vertices provide free connectivity for the pattern.

For a graph $G$ equipped with ghost vertices $\ghost$, by $\torso{G}{\ghost}$ we define the graph obtained from $G$ by {\em{eliminating}} each ghost vertex, that is, 
removing it and turning its neighborhood into a clique.
Note that $\torso{G}{\ghost}$ does not necessarily belong to $\Cc$. 
For two vertices $x,y$ in $G$, we define the {\em{distance function}} $\dist_G(x,y)$ as the distance between $x$ and $y$ as the minimum possible number of non-ghost vertices
on a path connecting $x$ and $y$, minus~$1$.
In other words, when a graph is equipped with ghost vertices, non-ghost vertices have cost~$1$ of traversing them, whereas ghost vertices have cost~$0$.
Note that if $x$ and $y$ are non-ghost vertices, then $\dist_G(x,y)$ is equal to the (normal) distance between $x$ and $y$ in $\torso{G}{\ghost}$.

For an instance $\Ii$, we define the subset $\Far_\Ii(X)$ of {\em{far}} vertices as follows:
\begin{equation*}
\Far_\Ii(X):=\{u\in X\ \colon\ \dist_G(u,r_0)>1000\sqrt{k}\lg k\}.
\end{equation*}
If a vertex of $X$ is not far, it is said to be {\em{close}}. 
Obviously, by invariant~\eqref{inv:dist} we have that no far vertex is a light terminal, that is, $\Far_\Ii(X)\cap \light=\emptyset$.
For a pattern $X$, we define the following potentials.

\smallskip
\begin{center}
\begin{tabular}{ p{5cm} p{7cm} }
{\em{Pattern potential}} & $\patsiz_\Ii(X):=|X\setminus \light|$\\[0.2cm]
{\em{Graph potential}} & $\grasiz_\Ii:=|V(G)\setminus (\light\cup \ghost)|$\\[0.2cm]
{\em{Distance potential}} & $\dstpot_\Ii(X):=|\Far_\Ii(X)|$ 
\end{tabular}
\end{center}

\smallskip

\noindent We drop the subscript $\Ii$ whenever the instance $\Ii$ is clear from the context.

Our goal in the general problem is to compute a subset of non-ghost vertices $A\subseteq V(G)\setminus \ghost$ with the following properties:
\begin{enumerate}[(1)]
\item It holds that $\light\subseteq A$, and the graph $G[A]$ admits a tree decomposition of width at most $24022\ctw\sqrt{k}\lg k$, where $\trms\cap A$ is contained in the root bag.
\item For every pattern $X$ in instance $\Ii$, we require that
\begin{equation}\label{eq:main-prob}
\mathbb{P}(X\subseteq A)\geq \exp\left[-c_1\cdot\frac{\lg k+\lg\lg n}{\sqrt{k}}\cdot (\patsiz(X)\lg\patsiz(X)+\dstpot(X))\right]\cdot \left(1-\frac{1}{k}\right)^{c_2\patsiz(X)\lg \grasiz},
\end{equation}
for some constants $c_1,c_2$, where $n$ is the total number of vertices of the graph.
\end{enumerate}
The constants $c_1,c_2$ will be fixed while explaining the proof; actually, we will fix $c_1=2$ and $c_2=2$.
For convenience, by $$\Monster(n,\patsiz(X),\grasiz,\dstpot(X))$$ we denote the right-hand side of~\eqref{eq:main-prob}, regarded as a function of the potentials and the number of vertices $n$.
Note that both $X$ and $A$ reside in the graph with the ghost vertices removed. 
The intuition is that ghost vertices do not belong to the piece of the graph that we are currently decomposing, but we cannot forget about them completely because 
they provide connectivity for the pattern.

\paragraph*{Applying the general problem.}
We now argue that an algorithm for the problem stated above implies the result claimed in Theorem~\ref{thm:maintheorem}. For this, we need the following simple claim.

\begin{claim}\label{cl:loglog}
The following holds:
$$2^{\sqrt{k}\lg k\lg\lg n}\leq 2^{\sqrt{k}\lg^2 k}\cdot n^{o(1)}$$
\end{claim}
\begin{proof}
The left hand size is equal to $(\lg n)^{\sqrt{k}\lg k}$. Suppose first that $n\leq 2^k$. Then
$$(\lg n)^{\sqrt{k}\lg k}\leq k^{\sqrt{k}\lg k}=2^{\sqrt{k}\lg^2 k}.$$
Suppose second that $n>2^k$. Then
$$(\lg n)^{\sqrt{k}\lg k}\leq (\lg n)^{\sqrt{\lg n}\lg\lg n}=2^{\sqrt{\lg n}\cdot (\lg \lg n)^2}=n^{o(1)}.$$
Hence in both cases we are done.
\cqed\end{proof}

In order to obtain an algorithm as stated in Theorem~\ref{thm:maintheorem}, we can sample one vertex $r$ of the input graph $G_0$, define $G$ to be the connected component of $G_0$ that contains $r$,
and apply the algorithm for the general problem to the instance $\Ii_0:=(G,r,\{r\},\emptyset,\emptyset,0)$. Fix some subset $X\subseteq V(G_0)$ with $G_0[X]$ being connected and $|X|\leq k$. 
Conditioned on the event that $r\in X$, which happens with probability at least $1/|V(G_0)|$, the algorithm for the general problem returns a suitable subset $A$ that covers $X$
with probability lower bounded by $\Monster(n,\patsiz_{\Ii_0}(X),\grasiz_{\Ii_0},\dstpot_{\Ii_0}(X))$, where $n=|V(G)|$. Observe that $\patsiz_{\Ii_0}(X)\leq k$, $\dstpot_{\Ii_0}(X)\leq k$ and $\grasiz_{\Ii_0}\leq n$. 
Thus, we have that
\begin{eqnarray*}
\mathbb{P}(X\subseteq A) & \geq & \exp\left[-c_1\sqrt{k}\lg k(\lg k+\lg \lg n)\right]\cdot \left(1-\frac{1}{k}\right)^{c_2\cdot k\lg n}\\
& \geq & \exp\left[-c_1\sqrt{k}\lg^2 k-c_1\sqrt{k}\lg k\lg\lg n - \frac{2}{k}\cdot c_2k\lg n\right]\\
& \geq & \exp\left[-2c_1\sqrt{k}\lg^2 k-o(\lg n) - 2c_2\lg n\right]\\
& = & \left[2^{\Oh(\sqrt{k} \log^2 k)}\cdot n^{\Oh(1)}\right]^{-1}.
\end{eqnarray*}
Here, the second inequality follows from the fact that $1-x\geq e^{-2x}$ for $x\in [0,1/2]$, which can be checked by simple differentiation.
The third inequality follows from Claim~\ref{cl:loglog}.
Together with the $1/n$ probability that indeed $r\in X$,
we obtain the success probability as required in Theorem~\ref{thm:maintheorem}.

\subsection{Solving the general problem}

We now proceed to the description of the recursive procedure for the general problem. 
Throughout the description we fix the considered instance $\Ii_0=(G,r,\light,\heavy,\ghost,0)$ and denote $n=|V(G)|$.
We assume that $\Ii_0$ satisfies the invariants stated in the previous section.
Moreover, we will assume the following:
\begin{enumerate}[(a)]
\item\label{inv:ghost} Ghost vertices are pairwise non-adjacent and no ghost vertex is adjacent to the root.
\item\label{inv:credit} It holds that $\lambda\leq \sqrt{k}/10$.
\item\label{inv:base} There is at least one vertex that is neither a terminal nor a ghost vertex.
\item\label{inv:pat} We restrict our attention only to patterns $X$ for which it holds that $\patsiz_{\Ii_0}(X)>0$.
\end{enumerate}
Note that assumption~\eqref{inv:credit}, together with the invariant about the number of terminals, in particular implies that $|\trms|\leq 16015\ctw\sqrt{k}\lg k$.
We now explain how the assumptions above can be guaranteed.

For~\eqref{inv:ghost}, observe that whenever two ghost vertices are adjacent,
we can contract the edge connecting them without changing the family of patterns in the instance, or any of its potentials. 
The same happens when a ghost vertex is adjacent to the root: we can contract it onto the root.
Hence, whenever~\eqref{inv:ghost} is not satisfied,
we can just apply these operations exhaustively, and work further on the instance obtained in this manner.

For~\eqref{inv:credit}, observe that if the credit exceeds $\sqrt{k}/10$ in any considered instance, then the upper bound on the sizes of considered patterns becomes negative,
and the instance of the general problem may be solved trivially by returning $A=\trms$ with a trivial tree decomposition consisting of one root node with bag equal to $A$.
This decomposition has width at most $|\trms|\leq 16014\ctw\sqrt{k}\lg k+\lambda$; however, we so far do not have any concrete upper bound on $\lambda$. 
Nevertheless, it will be the case throughout the algorithm that we will apply the algorithm only to instances with credit at most $\sqrt{k}/5$.
Hence, it is the case that this decomposition has width at most $16014\ctw\sqrt{k}\lg k+\sqrt{k}/5<24022\ctw\sqrt{k}\lg k$.

For~\eqref{inv:base}, observe that in this case we may again just output $A=\trms$ with a trivial tree decomposition consisting of one root node with bag equal to $A$.
Again, this decomposition has width at most $|\trms|\leq 16014\ctw\sqrt{k}\lg k+\lambda<24022\ctw\sqrt{k}\lg k$. This provides the base case for our recursion.

Finally, for~\eqref{inv:pat}, if $\patsiz_{\Ii_0}(X)=0$, then every vertex of $X$ is a light terminal, so in particular $X\subseteq \light$. 
Hence, whatever choice the algorithm does, as long as it outputs a vertex subset $A$ that contains $\light$, we will have that $X\subseteq A$ with probability $1$.
Therefore, we can focus only on the case when $\patsiz_{\Ii_0}(X)$ is positive.

Having these assumptions in place, let us define
$$M:=\{u\in V(G)\colon \dist_G(r,u)\leq 2000\sqrt{k}\lg k\}.$$
Obviously, $M$ can be computed in linear time using a breadth-first search from $r$ (here we need a trivial modification to traverse ghost vertices at cost $0$).
Every connected component of $G-M$ will be called an {\em{island}}.
Note that, by the definition of $M$, every vertex of an island that neighbors some vertex of $M$ cannot be a ghost vertex.
Hence, in particular every island contains some non-ghost vertex.

The first step of the algorithm is to apply the clustering procedure of Theorem~\ref{thm:clust} for parameter $k$ to all the islands.
More precisely, we apply the algorithm of Theorem~\ref{thm:clust} to the graph $K:=\torso{(G-M)}{\ghost\setminus M}$, that is, to $G-M$ with all the ghost vertices eliminated.
This algorithm works in randomized polynomial time, and returns a subset $B'\subseteq V(K)$ with the following properties:
\begin{itemize}
\item each connected component of $K[B']$ has radius at most $9k^2\lg n$, and 
\item with probability at least $1-1/k$ we have that $X\setminus M$ is contained in $B'$.
\end{itemize}
Now define $B$ to be $B'$ extended by adding all ghost vertices of $\ghost\setminus M$ that have at least one neighbor in $B'$. 
Then each connected component of $G[B]$ has radius at most $9k^2\lg n$, computed according to the distance measure $\dist_G$ that treats ghost vertices as traversed for free.
Moreover, with probability at least $1-1/k$ we have that $X\setminus M\subseteq B$.
We henceforth assume that this event happens, i.e., indeed $X\subseteq M\cup B$, keeping in mind the multiplicative factor of $1-1/k$ in the success probability.

Thus, we can restrict ourselves to the graph $G'$ defined as the connected component of $G[M\cup B]$ that contains $r$; 
indeed, since each vertex of $X$ can be reached from $r$ using only ghost vertices and other vertices of $X$, by the construction of $B$ it must be entirely
contained in this connected component. Note that $G[M]$ is connected by definition, so $M\subseteq V(G')$.
Observe that as we remove only vertices at distance larger than $2000\sqrt{k}\lg k$ from $r$ in $G$, every vertex of $X$ that was close or far, remains close or far respectively.
These observations are formalized in the following claim, whose proof follows as explained above.
\begin{claim}
Consider a new instance $$\Ii=(G',r,\light\cap V(G'),\heavy\cap V(G'),\ghost\cap V(G'),\lambda).$$ Provided $X\subseteq M\cup B$, $X$ remains a pattern in $\Ii_0$ and $\Far_{\Ii}(X)=\Far_{\Ii_0}(X)$.
\end{claim}
From now on we focus on the instance $\Ii$. Observe that, by the claim above, the following holds:
\begin{eqnarray*}
\patsiz_{\Ii_0}(X)=\patsiz_{\Ii}(X),\qquad \grasiz_{\Ii_0}\geq \grasiz_{\Ii},\qquad \dstpot_{\Ii_0}(X)=\dstpot_{\Ii}(X),
\end{eqnarray*}
so this restriction can only lower the potentials and make a better lower bound on the success probability.
By slightly abusing the notation, we redefine the islands to be the connected components of $G'-M$.
Also, if in $\Ii$ there is no non-ghost, non-terminal vertex, then we conclude by outputting $A=\trms$. 
Hence from now on we assume that $\Ii$ has a non-ghost, non-terminal vertex, which in particular implies that $\grasiz_\Ii>0$.

Construct an auxiliary graph $H$ from $G'$ by contracting each island $C\in \cc(G'-M)$ into a single (non-ghost) vertex $u_C$; let $I=\{u_C\colon C\in \cc(G'-M)\}$.
By $\iota\colon V(G')\to V(H)$ we denote the function that assigns to each vertex of $G'$ its image under the contraction; i.e., $\iota$ is identity on $M$ and $\iota(V(C))=\{u_C\}$ for each island $C$.
Obviously $H$ is a minor of $G'$, which in turn is a subgraph of $G$, and hence $H\in \Cc$. 
Moreover, each vertex of $H$ is at distance at most $2000\sqrt{k}\lg k+1\leq 2001\sqrt{k}\lg k$ from $r$.
Hence, if we measure the distance in $H$ normally, without eliminating the ghost vertices, then each vertex is at distance at most $4002\sqrt{k}\lg k$; this follows from the invariant that no two ghost vertices
are adjacent. From Proposition~\ref{prop:local-tw} we infer that the treewidth of $H$ is at most $4002\ctw\cdot \sqrt{k}\lg k$.

Define two weight functions $\mathbf{w}_1(u)$ and $\mathbf{w}_2(u)$ on $V(H)$ as follows.
For a vertex $u\notin I$, we put $\mathbf{w}_1(u)=1$ if $u\in \trms$ and $\mathbf{w}_1(u)=0$ otherwise.
However, if $u=u_C$ for some island $C$, then we put $\mathbf{w}_1(u_C)=|V(C)\cap \trms|$.
Similarly, for $u\notin I$, we put $\mathbf{w}_2(u)=1$ if $u\in V(G)\setminus (\light\cup \ghost)$ and $\mathbf{w}_2(u)=0$ otherwise,
whereas for $u=u_C$, we put $\mathbf{w}_2(u)=|V(C)\setminus (\light\cup \ghost)|$.
In other words, $\mathbf{w}_1(u)$ and $\mathbf{w}_2(u)$ are characteristic functions of $\trms$ and of $V(G)\setminus (\light\cup \ghost)$, 
where all vertices contained in one island $C$ contribute to the weight of the collapsed vertex $u_C$.

By applying Lemma~\ref{lem:balsep} to these weight functions, we infer that there exist sets $Z_1$ and $Z_2$, each of size at most $4002\ctw\cdot \sqrt{k}\lg k+1\leq 4003\ctw\cdot\sqrt{k}\lg k$,
such that each connected component of $H-Z_1$ has $\mathbf{w}_1$-weight at most $|\trms|/2$, and every connected component of $H-Z_2$ has $\mathbf{w}_2(u)$-weight at most
$|V(G)\setminus (\light\cup \ghost)|/2$. We define 
\begin{eqnarray*}
Z & := & Z_1\cup Z_2\cup \{r\}\ ;\\
W & := & \iota^{-1}(Z).
\end{eqnarray*}
Clearly, $|Z|\leq 8006\ctw \sqrt{k}\lg k+1\leq 8007\ctw\sqrt{k}\lg k$.
Also, from the definition of weight functions $\mathbf{w}_1$ and $\mathbf{w}_2$, and the properties of $Z_1$ and $Z_2$, we immediately obtain the following.

\begin{claim}\label{cl:cc-after-rem}
For each connected component $D$ of $G'-W$, we have 
$$|V(D)\setminus (\light\cup \ghost)|\leq |V(G)\setminus (\light\cup \ghost)|/2\qquad \textrm{and}\qquad |V(D)\cap \trms|\leq |\trms|/2.$$
\end{claim}

Let us define
\begin{eqnarray*}
W_{\mathsf{isl}}& := & \iota^{-1}(Z\cap I);\\
W_{\mathsf{nrm}}& := & \iota^{-1}(Z\setminus I).
\end{eqnarray*}
In other words, $(W_{\mathsf{isl}},W_{\mathsf{nrm}})$ is the partition of $W$ with respect to which vertices are in the islands and which are not.
Since $\iota$ is identity on $M$, we have that $|W_{\mathsf{nrm}}|\leq |Z|\leq 8007\ctw\sqrt{k}\lg k$ and $r\in W_{\mathsf{nrm}}$.

The algorithm now branches into two cases based on a random decision, as follows.
With probability $1-1/k$ the algorithm assumes that $W_{\mathsf{isl}}$ is disjoint with the sought pattern $X$; that is, no island that is contracted onto a vertex of $Z$ intersects $X$.
In the remaining event, which happens with probability $1/k$, the algorithm assumes that $X$ intersects some island that is contracted onto a vertex of $Z$.
We now describe what steps are taken next in each of these cases, supposing that the algorithm made a correct assumption.
The success probability analysis will be explained at the end of each case.

\subsubsection{Case when $W_{\mathsf{isl}}$ is disjoint with the pattern}\label{sss:disjoint}

Here, we suppose that the algorithm assumed that $W_{\mathsf{isl}}\cap X=\emptyset$, and that this assumption is correct.
Hence, it is safe for the algorithm to restrict attention to the graph 
$$G'':=G'-W_{\mathsf{isl}}.$$
More precisely, it still holds that $X$ is contained in $G''$, and that every vertex of $X$ can be reached from $r$ in $G''$ by a path traversing only vertices of $X$ and ghost vertices;
the last claim is because vertices of islands neighboring $M$ have to be non-ghost vertices.
Note that since $G''$ is obtained from $G'$ only by removing some selection of islands, it still holds that $G''$ is connected.

We will consider the connected components of $G'-W=G''-W_{\mathsf{nrm}}$; let $\Comps:=\cc(G''-W_{\mathsf{nrm}})$ be their set.
Consider any $D\in \Comps$. By Claim~\ref{cl:cc-after-rem} we have that
$$|V(D)\cap \trms|\leq |\trms|/2\leq 8007\ctw\sqrt{k}\lg k+\lambda;$$
the last inequality follows from the invariant that $|\trms|\leq 16014\ctw\sqrt{k}\lg k+\lambda$.

Now, for each component $D\in \Comps$, we shall construct an instance $\Ii_D=(G_D,r,\light_D,\heavy_D,\ghost_D,\lambda)$.
We begin by defining the graph $G_D$, which will be constructed from $G''$ by a series of contractions.
Define $V_D=N_{G''}[V(D)]\cup \{r\}$, and consider the connected components of graph $G''-V_D$.
For each component $Q\in \cc(G''-V_D)$ that contains a neighbor of the root vertex $r$, contract the whole $Q$ onto $r$.
For each component $Q\in \cc(G''-V_D)$ that has no neighbor of $r$, contract the whole $Q$ into a new vertex $g_Q$ and declare it a ghost vertex.
We define $G_D$ to be the graph obtained from $G''$ by applying such contraction for each connected component $Q$ of $G''-V_D$. 
Thus, the set $\ghost_D$ of ghost vertices in graph $G_D$ consists of $\ghost\cap V_D$, plus we add $g_Q$ to $\ghost_D$ for each $Q\in \cc(G''-V_D)$ that has no neighbor of $r$.
Obviously, $G_D$ is still connected and contains the root vertex $r$.

Next, we need to define terminals in the instance $\Ii_D$.
For this, create an auxiliary graph $L$ from $G''$ by contracting every connected component $D\in \Comps$ into a single vertex $w_D$.
Obviously, $L$ is connected and $r$ persists in $L$; we treat $L$ as a graph without any ghost vertices, so all distances are computed normally.
Run a breadth-first search (BFS) in $L$ starting from $r$. For every vertex $v\in W_{\mathsf{nrm}}$, we say that $v$ is {\em{charged}} to component $D\in \Comps$ if the parent of $v$ in the BFS tree
exists and is equal to $w_D$. Note that a vertex of $W_{\mathsf{nrm}}$ can be not charged to any component if $v=r$ or the parent of $v$ in the BFS tree also belongs to $W_{\mathsf{nrm}}$.

With this definition, we are ready to define terminals in instance $\Ii_D$. We put:
\begin{eqnarray*}
\light_D & = & (\light\cap V_D)\cup \{v\in N_{G''}(V(D))\setminus (\light\cup \ghost)\colon \textrm{$v$ is not charged to $D$}\};\\
\heavy_D & = & (\heavy\cap V(D))\cup \{v\in N_{G''}(V(D))\setminus (\light\cup \ghost)\colon \textrm{$v$ is charged to $D$}\}.
\end{eqnarray*}
In other words, the terminals in the instance $\Ii_D$ comprise terminals that were originally in component $D$, plus we add also all non-ghost vertices from the boundary $N_{G''}(V(D))$ as terminals.
These new terminals are partitioned into light and heavy depending on whether they are charged to $D$. From the definition it readily follows that $\light_D$ and $\heavy_D$ are disjoint, as required.
Note that each vertex $v\in W_{\mathsf{nrm}}$ is charged to at most one component $D$,
so it can be declared as a heavy terminal in at most one instance $\Ii_D$.

From Claim~\ref{cl:cc-after-rem} and the facts that $|\trms|\leq 16014\ctw\sqrt{k}\lg k+\lambda$ and $N_{G''}(V(D))\subseteq W_{\mathsf{nrm}}$, we immediately obtain that
\begin{eqnarray*}
|\trms_D| & \leq & |\trms\cap V_D| + |W_{\mathsf{nrm}}| \leq |\trms|/2 + 8007\ctw\sqrt{k}\lg k\leq 16014\ctw\sqrt{k}\lg k+\lambda.
\end{eqnarray*}
Hence, invariant~\eqref{inv:terminals} is satisfied in the instance $\Ii_D$. We are left with invariant~\eqref{inv:dist}, whose satisfaction is proved in a separate claim.

\begin{claim}
In graph $G_D$, every vertex from $\light_D$ is at distance at most $3$ from $r$.
\end{claim}
\begin{proof}
Fix some $v\in \light_D$. Suppose first that $v\in \light\cap V_D$. 
Then, by the assumption on invariant~\eqref{inv:dist} holding in the instance $\Ii$, $v$ was already at distance at most $3$ from $r$ in graph $G$. 
During the construction of $G''$ from $G$ we removed only vertices at distance more than $2000\ctw\sqrt{k}\lg k$ from $r$, hence $v$ is still at distance at most $3$ from $r$ in $G''$. 
Graph $G_D$ was obtained from $G''$ by means of edge contractions, which can only decrease the distances. 
Hence, $v$ is at distance at most $3$ from $r$ also in $G_D$.

Consider now the remaining case when $v\in N_{G''}(V(D))\setminus (\light\cup \ghost)\subseteq W_{\mathsf{nrm}}\setminus (\light\cup \ghost)$.
If $v=r$ then we are done, so assume otherwise.
Let $P$ be the path from $r$ to $v$ in the BFS tree in graph $L$ used for the definition of charging.
Since $v\in \light_D$, we have that $v$ is not charged to $D$, and hence the parent of $v$ in this BFS tree is not equal to $w_D$.
As $w_D$ is a neighbor of $v$ in $L$, due to $v\in N_{G''}(V(D))$, this implies that $w_D$ cannot lie on path $P$.
Indeed, otherwise we could shortcut path $P$ by using the edge $w_Dv$, which contradicts the fact that $P$ is a shortest path from $r$ to $v$ in $L$ (due to using BFS for its construction).

Let $P'$ be the path $P$ lifted to graph $G''$ in the following manner: 
for every vertex $w_{D'}\in V(L)$ visited on $P$, we replace the visit of this vertex by traversal of an appropriate path within the connected component $D'$.
Then $P'$ is a path from $r$ to $v$ in $G''$ that does not traverse any vertex of $D$.
Let $v'$ be the first vertex on $P'$ that belongs to $N_{G''}(V(D))$; clearly $v'$ is well-defined because $v\in N_{G''}(V(D))$.
Observe that the prefix of $P$ from $r$ to $v'$, excluding $v'$, gets entirely contracted onto $r$ in the construction of $G_D$, so $v'$ is a neighbor of $r$ in $G_D$.

Finally, consider the suffix of $P$ from $v'$ to $v$. Observe that both $v$ and $v'$ are neighbors of $w_D$ in the graph $L$, so since $P$ is a shortest path,
we infer that the distance between $v'$ and $v$ on $P$ is at most $2$.
It is now easy to see that the suffix of $P'$ between $v$ and $v'$ gets contracted to a path of length at most $2$ in $G_D$.
Hence, we have uncovered a walk of length at most $3$ from $r$ to $v$ in $G_D$, which concludes the proof.
\cqed\end{proof}

Thus, we constructed an instance $\Ii_D=(G_D,r,\light_D,\heavy_D,\ghost_D,\lambda)$ for every connected component $D\in \Comps$,
and we verified that this instance satisfies the requested invariants.
Note that the constructed instances have the same credit $\lambda$ as the original one.
Finally, for the pattern $X$, we can define its {\em{projection}} $X_D$ onto instance $\Ii_D$ by simply putting:
$$X_D:=X\cap V_D.$$
Observe that each vertex of $X_D$ can be reached from $r$ in $G_D$ by a path that traverses only vertices from $X_D\cup \ghost_D$. 
This is because it could be reached from $r$ in $G''$ by some path traversing only $X\cup \ghost$, and parts of this path lying outside of $V_D$ has been either contracted
to edges, or replaced by ghost vertices during the construction of $G_D$ from $G''$.
Therefore, $X_D$ is a pattern in $\Ii_D$.

We now verify how the potentials behave in instances $\Ii_D$.

\begin{claim}\label{cl:recur-bounds}
The following holds:
\begin{eqnarray}
\label{eq:patsiz}\patsiz_\Ii(X) & \geq &\sum_{D\in \Comps} \patsiz_{\Ii_D}(X_D)\\
\label{eq:grasiz}\grasiz_\Ii & \geq & \sum_{D\in \Comps} \grasiz_{\Ii_D}\\
\label{eq:grasize-dec}\grasiz_\Ii/2 & \geq & \grasiz_{\Ii_D}\qquad\textrm{ for each }D\in \Comps\\
\label{eq:dstpot}\dstpot_\Ii(X) & \geq & \sum_{D\in \Comps} \dstpot_{\Ii_D}(X_D)
\end{eqnarray}
\end{claim}
\begin{proof}
Take any vertex $u\in V(G)\setminus (\light\cup \ghost)$. 
By the construction of instances $\Ii_D$, and in particular the mentioned fact that each vertex of $W_{\mathsf{nrm}}$ can be declared a heavy terminal in at most one instance $\Ii_D$,
it readily follows that there is at most one instance $\Ii_D$ for which $u\in V(G_D)\setminus (\light_D\cup \ghost_D)$.
From this observation we immediately obtain statements~\eqref{eq:patsiz} and~\eqref{eq:grasiz}.
Statement~\eqref{eq:dstpot} follows similarly, but one needs to additionally observe the following:
$G_D$ is obtained from $G''$ by means of edge contractions that can only decrease the distances from $r$, so if a vertex $u\in X_D$ is far in the instance $\Ii_D$, then it was also far in the original instance $\Ii$.
Finally,~\eqref{eq:grasize-dec} follows directly from Claim~\ref{cl:cc-after-rem}.
\cqed\end{proof}

Thus, Claim~\ref{cl:recur-bounds}\eqref{eq:grasize-dec} certifies that the graph potential drops significantly in each new instance. 
Intuitively, this drop will be responsible for amortizing the $(1-1/k)^2$ multiplicative factor in the success probability incurred by the preliminary clustering step, 
and by the random choice of the assumption on the considered case.

Note that, by the assumption that there is at least one non-ghost, non-terminal vertex, we have that $\grasiz_\Ii>0$. 
Hence, by Claim~\ref{cl:recur-bounds}, statement~\eqref{eq:grasize-dec} in particular, each of the instances $\Ii_D$ will have strictly fewer vertices than $\Ii$.

Therefore, we may apply the algorithm recursively to each instance $\Ii_D$. This yields subsets of vertices $\{A_D\colon D\in \cc(G''-W_{\mathsf{nrm}})\}$ with the following properties:
\begin{itemize}
\item $A_D\supseteq \light_D$ and $G_D[A_D]$ admits a tree decomposition $\Tt_D$ of width at most $24022\ctw\sqrt{k}\lg k$ with $\trms_D\cap A_D$ contained in the root bag.
\item The probability that $X_D\subseteq A_D$ is at least $\Monster(n_D, \patsiz_{\Ii_D}(X_D),\grasiz_{\Ii_D},\dstpot_{\Ii_D}(X_D))$, where $n_D=|V(G_D)|$.
\end{itemize}
Let us define now the set $A$.
\begin{itemize}
\item First, for every $D \in \Comps$, we put $V(D) \cap A_D$ into $A$.
\item Second, for every $v \in W_{\mathsf{nrm}} \setminus \ghost$, we put $v$ into $A$ if \emph{for every} $D \in \Comps$ such that $v \in N_{G''}(V(D))$
we have $v \in A_D$. In particular, if there is no $D\in \Comps$ for which $v \in N_{G''}(V(D))$, we also include $v$ in $A$.
\end{itemize}
Observe that for every $D \in \Comps$, we have $A \cap V_D \subseteq A_D$, but not necessarily $A \cap V_D = A_D$. 

To get more intuition about the above step, in particular the universal quantification in the second step, let us observe the following.
Consider a vertex $v \in W_{\mathsf{nrm}} \setminus \ghost$. This vertex is a terminal in every instance $\Ii_D$ for which $v \in N_{G''}(V(D))$; it is a heavy terminal
in at most one such instance, and a light terminal in all other such instances. If $v \in \light_D$, then clearly $v \in A_D$; thus, we have $v \notin A$
if and only if the (unique) instance $\Ii_D$ where $v \in \heavy_D$ exists and, furthermore, $v \notin A_D$ for this instance.
Hence, intuitively, we allow this particular instance $\Ii_D$ where $v \in \heavy_D$ to exclude the vertex $v$ from $A_D$ if it is deemed necessary; all other instances
are required to keep it in the set $A_D$. On the other hand, note that the vertex $v$ in the instance $\Ii_D$ where $v \in \heavy_D$ contributes to the potential
$\patsiz_{\Ii_D}(X_D)$, and does not contribute to this potential in the other instances.

We now formally verify that $A$ defined in this way has the required properties.
First, note that every vertex $v \in W_{\mathsf{nrm}} \setminus \ghost$ such that $v \notin \bigcup_{D \in \Comps} N_{G''}(V(D))$
(i.e., the universal quantification in the second step of the construction of the set $A$ is done over an empty set of choices)
is included in the set $A$. This, together with the fact that $\light \cap V_D \subseteq \light_D \subseteq A_D$ for every $D \in \Comps$, shows
that $\light \subseteq A$.
Moreover, from this and the way we defined $A$ it follows that if $X_D \subseteq A_D$ for every $D \in \Comps$, then we have $X \subseteq A$, as $X_D = X \cap V_D$ by definition.

We now check that $G[A]$ indeed admits a suitable tree decomposition.

\begin{claim}\label{cl:case-disjoint-tw}
The subgraph $G[A]$ admits a tree decomposition of width at most $24022\ctw\sqrt{k}$ with $\trms\cap A$ contained in the root bag.
\end{claim}
\begin{proof}
Create the root node with bag $A \cap (\trms \cup (W_{\mathsf{nrm}}\setminus \ghost))$ associated with it. 
Next, for each $D\in \Comps$, restrict the decomposition $\Tt_D$ to the vertices of $A \cap V_D$; 
that is, remove all the vertices of $A_D\setminus (A\cap V_D)$ from all the bags, thus constructing a tree decomposition $\Tt'_D$ of $G_D[A \cap V_D]$.
Then, for each $D\in \Comps$, attach the obtained decomposition $\Tt'_D$ below the root node by making its root a child of the root node.
Observe whenever a vertex $v \notin \ghost$ is shared between multiple instances $\Ii_D$, we have that 
$v$ is a vertex of $v \in W_{\mathsf{nrm}} \setminus \ghost$ that is a terminal
in all of them (i.e. belongs to $\trms_D$), and moreover $v \in A$ only if $v \in A_D$ for every instance $\Ii_D$ where $v$ is present.
Since the root bag of each $\Tt'_D$ contains $A \cap \trms_D$,
it is now easy to verify that in this manner we obtain a tree decomposition of $G[A]$.

Observe now that
\begin{eqnarray*}
|A\cap (\trms \cup (W_{\mathsf{nrm}}\setminus \ghost))| & \leq &|\trms \cup (W_{\mathsf{nrm}}\setminus \ghost)|\leq |\trms|+|W_{\mathsf{nrm}}|\\
& \leq & 16014\ctw\sqrt{k}\lg k+\lambda+8007\ctw\sqrt{k}\lg k=24022\ctw\sqrt{k}\lg k;
\end{eqnarray*}
here, the last inequality follows from the fact that $\lambda\leq \sqrt{k}/10$.
Since each $\Tt_D$ has width at most $24022\ctw\sqrt{k}\lg k$, it follows that the obtained tree decomposition of $G[A]$ has width at most $24022\ctw\sqrt{k}\lg k$.
\cqed\end{proof}

Finally, we analyze the success probability. For this, we fix $c_2:=2$.

\begin{claim}\label{cl:case-disjoint-error}
Supposing $X\cap W_{\mathsf{isl}}=\emptyset$, the algorithm outputs a set $A$ with $X\subseteq A$ with probability at least $\Monster(n,\patsiz_\Ii(X),\grasiz_\Ii,\dstpot_\Ii(X))$. 
This includes the $(1-1/k)$ probability of success of the preliminary clustering step, and $(1-1/k)$ probability that the algorithm makes the correct assumption that $X\cap W_{\mathsf{isl}}=\emptyset$.
\end{claim}
\begin{proof}
The preliminary clustering step is correct (that is, the set of removed vertices is disjoint with~$X$) with probability at least $1-1/k$.
Then, the algorithm makes the correct assumption with probability $1-1/k$.
We have already argued that if $X_D\subseteq A_D$ for each $D\in \Comps$, then also $X\subseteq A$.
The event $X_D\subseteq A_D$ holds with probability at least $\Monster(n_D, \patsiz_{\Ii_D}(X_D),\grasiz_{\Ii_D},\dstpot_{\Ii_D}(X_D))$, where $n_D\leq n$ is the number of vertices in instance $\Ii_D$.
Hence, we have
\begin{equation}\label{eq:est-indep}
\mathbb{P}(X\subseteq A)\geq \left(1-\frac{1}{k}\right)^2\cdot \prod_{D\in \Comps} \Monster(n_D,\patsiz_{\Ii_D}(X_D),\grasiz_{\Ii_D},\dstpot_{\Ii_D}(X_D)).
\end{equation}
From Claim~\ref{cl:recur-bounds}, equations~\eqref{eq:patsiz} and~\eqref{eq:dstpot}, and convexity of function $t\to t\lg t$ we infer the following:
\begin{eqnarray}
\label{eq:est-patsiz}\exp\left[-c_1\cdot\frac{\lg k+\lg \lg n}{\sqrt{k}}\cdot \patsiz_\Ii(X)\lg\patsiz_\Ii(X)\right] & \leq & \prod_{D\in \Comps} \exp\left[-c_1\cdot \frac{\lg k+\lg \lg n_D}{\sqrt{k}}\cdot \patsiz_{\Ii_D}(X_D)\lg\patsiz_{\Ii_D}(X_D)\right]\\
\label{eq:est-dstpot}\exp\left[-c_1\cdot\frac{\lg k+\lg \lg n}{\sqrt{k}}\cdot \dstpot(X)\right] & \leq & \prod_{D\in \Comps} \exp\left[-c_1\cdot \frac{\lg k+\lg \lg n_D}{\sqrt{k}}\cdot \dstpot(X_D)\right]
\end{eqnarray}
To estimate the last factor in each $\Monster(n_D,\patsiz_{\Ii_D}(X_D),\grasiz_{\Ii_D},\dstpot_{\Ii_D}(X_D))$, we use Claim~\ref{cl:recur-bounds}, and equations~\eqref{eq:patsiz} and~\eqref{eq:grasize-dec}.
Recall that for any $D\in \Comps$, we have
$$\grasiz_\Ii/2\geq \grasiz_{\Ii_D}.$$
Hence, if we pick $c_2=2$, then we obtain the following:
\begin{eqnarray*}
\left(1-1/k\right)^{c_2\patsiz_\Ii(X)\lg \grasiz_\Ii} & = & \left(1-1/k\right)^{c_2\patsiz_\Ii(X)}\cdot \left(1-1/k\right)^{c_2\patsiz_\Ii(X)\lg (\grasiz_\Ii/2)} \\
& \leq & \left(1-1/k\right)^{c_2\patsiz_\Ii(X)}\cdot \prod_{D\in \Comps} \left(1-1/k\right)^{c_2\patsiz_{\Ii_D}(X_D)\lg (\grasiz_\Ii/2)} \\
& \leq & \left(1-1/k\right)^{c_2\patsiz_\Ii(X)}\cdot \prod_{D\in \Comps} \left(1-1/k\right)^{c_2\patsiz_{\Ii_D}(X_D)\lg \grasiz_{\Ii_D}}
\end{eqnarray*}
Recall that we assumed that $\patsiz_{\Ii}(X)=\patsiz_{\Ii_0}>0$. Therefore, this yields:
\begin{equation}\label{eq:est-grasiz}
\left(1-1/k\right)^{c_2\patsiz_\Ii(X)\lg \grasiz_\Ii} \leq \left(1-1/k\right)^2\cdot\prod_{D\in \Comps} \left(1-1/k\right)^{c_2\patsiz_{\Ii_D}(X_D)\lg \grasiz_{\Ii_D}}.
\end{equation}
Putting~\eqref{eq:est-indep},~\eqref{eq:est-patsiz},~\eqref{eq:est-dstpot}, and~\eqref{eq:est-grasiz} together, we obtain that
\begin{eqnarray*}
\mathbb{P}(X\subseteq A) & \geq & \left(1-1/k\right)^2\cdot \prod_{D\in \Comps} \Monster(n_D,\patsiz_{\Ii_D}(X_D),\grasiz_{\Ii_D},\dstpot_{\Ii_D}(X_D))\\
& \geq & \exp\left[-c_1\cdot\frac{\lg k+\lg\lg n}{\sqrt{k}}\cdot (\patsiz_\Ii(X)\lg\patsiz_\Ii(X)+\dstpot_\Ii(X))\right]\cdot \left(1-1/k\right)^{c_2\patsiz_\Ii(X)\lg \grasiz_\Ii}\\
& = & \Monster(\patsiz_{\Ii}(X),\grasiz_{\Ii},\dstpot_{\Ii}(X)).
\end{eqnarray*}
This concludes the proof.
\cqed\end{proof}

Claims~\ref{cl:case-disjoint-tw} and~\ref{cl:case-disjoint-error} verify that the behaviour of the algorithm is as required, provided $X\cap W_{\mathsf{isl}}=\emptyset$.

\subsubsection{Case when $W_{\mathsf{isl}}$ intersects the pattern}\label{sss:not-disjoint}

We are left with describing the steps taken by the algorithm after taking the assumption that $W_{\mathsf{isl}}$ intersects the pattern $X$; recall that the algorithm takes this decision with probability $\frac{1}{k}$.
Henceforth, we assume that this assumption is correct.

First, the algorithm chooses uniformly at random a vertex $u_C$ of $Z\cap I$, and it assumes that $X$ intersects the island $C=\iota^{-1}(u_C)$. 
Since at least one such island is intersected by $X$ and $|Z\cap I|\leq |Z|\leq 8007\ctw\sqrt{k}\lg k$, the algorithm makes a correct choice with probability at least $(8007\ctw\sqrt{k}\lg k)^{-1}\geq k^{-7}$. 
Keeping this success probability in mind, from now on we assume that the choice was indeed correct; thus, the algorithm knows one island $C$ about which it can assume that $V(C)\cap X\neq \emptyset$.

Recall that, due to the preliminary clustering step, island $C$ has radius bounded by $9k^2\lg n$, where the radius is counted w.r.t. the distance measure that regards ghost vertices as traversed for free.
Select a non-ghost vertex $z$ of $C$ such that $\dist_C(u,z)\leq 9k^2\lg n$ for each $u\in V(C)$.
Let 
$$d:=\inf \{\dist_{C}(z,u)\colon u\in V(C)\cap X\}.$$
Since $V(C)\cap X\neq \emptyset$, we have that $0\leq d\leq d_{\max}$, where $d_{\max}=\inf \{\dist_{C}(z,u)\colon u\in V(C)\}$, which is not larger than $1+9k^2\lg n$.
The algorithm now samples an integer value between $0$ and $d_{\max}$, and assumes henceforth that the sampled value is equal to $d$.
This assumption holds with probability at least $1/(1+d_{\max})\geq (10k^2\lg n)^{-1}$.
Keeping this success probability in mind, we proceed with the assumption that the algorithm knows the correct value of $d$.

Let 
$$S:=\{u\in V(C)\setminus \ghost\colon \dist_C(u,z)<\max(d,1)\}\cup \{u\in V(C)\cap \ghost\colon \dist_C(u,z)<\max(d,1)-1\}.$$
That is, $S$ contains all the vertices of $C$ that are at distance less than $\max(d,1)$ from $z$, but we exclude ghost vertices at distance exactly $\max(d,1)-1$.
From the definition it readily follows that the induced subgraph $C[S]$ contains $z$ and is connected.
Construct graph $G''$ (this graph is different than $G''$ considered in the previous section) by contracting the whole subgraph $G'[S]$ onto $z$; the contracted vertex of $G''$ will be also denoted as $z$.
Note that if $d\leq 1$, then in fact no contraction has been made and $G''=G'$.
Observe that, provided the sampled value of $d$ is correct, the set $S$ is disjoint with $X$ (unless $d=0$ when $S = \{z\}$ and no contraction is made).
Thus $X\subseteq V(G'')$.
Moreover, observe that each vertex of $X$ can be still reached from $r$ in $G''$ by a path that uses only ghost vertices and vertices of $X$. 
Indeed, by the definition of $S$, if $S$ is disjoint with $X$, then $S$ also does not contain any ghost vertex traversed on the aforementioned path.
Hence, $X$ can be still regarded as a pattern in $G''$, where the ghost vertices in $G''$ are inherited from $G'$.
On the other hand, by the definition of $d$ it follows that some vertex of $X$ is at distance at most $1$ from $z$ in $G''$ (so either $z$, in case $d=0$ and $z\in X$, or a neighbor of $z$, or a neighbor via one ghost vertex).
Also, observe that since $C$ is disjoint with $M$, that is, all vertices of $C$ are at distance larger than $2000\sqrt{k}\lg k$ from $r$ in $G''$, we have the following
\begin{equation}\label{eq:still-far}
\dist_{G''}(r,z)>2000\sqrt{k}\lg k.
\end{equation}

Now we would like to apply the duality theorem, i.e., Theorem~\ref{thm:duality}.
Consider graph $L:=\torso{G''}{\ghost\cap V(G'')}$ (this is a different graph than $L$ considered in the previous section), a pair of vertices $(s,t)=(r,z)$ of $L$ and the following parameters: 
$$p=\lceil 120\sqrt{k}\lg k\rceil\qquad\textrm{and}\qquad q=k.$$
By applying Theorem~\ref{thm:duality} to these, in polynomial time we can compute one of the following structures:
\begin{enumerate}[(a)]
\item An $(r,z)$-separator chain $(C_1,\ldots,C_p)$ with $|C_j|\leq 2k$ for each $j\in [p]$.
\item A sequence $(Q_1,\ldots,Q_k)$ of $(s,t)$-paths with $|(V(Q_i)\cap \bigcup_{i'\neq i} V(Q_{i'}))\setminus \{s,t\}|\leq 4p$ for each~$i\in [k]$.
\end{enumerate}
The behavior of the algorithm now differs depending on which structure has been uncovered.
We start with the simpler case when the algorithm of Theorem~\ref{thm:duality} yielded a sequence of paths.

\paragraph*{Subcase: a sequence of radial paths.}
Suppose that the algorithm of Theorem~\ref{thm:duality} returned a sequence $(Q_1,\ldots,Q_k)$ of $(r,z)$-paths, 
where each path contains only at most $4p$ vertices that also belong to other paths, not including $z$ and $r$.
These are paths in graph $\torso{G''}{\ghost\cap V(G'')}$, but we can lift then to $(r,z)$-paths $P_1,\ldots,P_k$ in $G''$ in a natural manner as follows: 
whenever some $Q_i$ traverses an edge obtained from eliminating some vertex $g$, we replace the usage of this edge by a path of length $2$ traversing $g$.
If we obtain a walk in this manner, i.e., some ghost vertex is used more than once, we shortcut the segment of the walk between the visits of this ghost vertex; thus we obtain again a simple path.
All in all, we are obtain $(r,z)$-paths $P_1,\ldots,P_k$ in $G''$ with the following property:  for each~$i\in [k]$, there can be at most $4p$ non-ghost vertices on $P_i$ that are
traversed by some other paths $P_{i'}$ for $i'\neq i$. Note that every ghost vertex can be used by many paths.

By~\eqref{eq:still-far}, we have that each $P_i$ has length larger than $2000\sqrt{k}\lg k$ (measured with ghost vertices contributing $0$ to the length).
For $i\in [k]$, define
$$\Pub(P_i)=(V(P_i)\cap \bigcup_{i'\neq i} V(P_{i'}))\setminus (\ghost\cup \{r,z\})\qquad\textrm{and}\qquad\Prv(P_i)=(V(P_i)\setminus \bigcup_{i'\neq i} V(P_{i'}))\setminus (\ghost\cup \{r,z\}).$$
We have that $|\Pub(P_i)|\leq 4p$ for all $i\in [k]$ and, by definition, sets $\Prv(P_i)$ are pairwise disjoint.
Pattern $X$ has at most $k$ vertices, out of which one is equal to $r$.
Hence, there is at least one index $i\in [k]$ for which $\Prv(P_i)$ is disjoint with $X$.
The algorithm samples one index $i$ from $[k]$ and assumes henceforth that the sampled index has the property stated above.
Note that this holds with probability at least $1/k$; keeping this success probability in mind, we proceed with the assumption that the algorithm made a correct choice of $i$.

Now that $\Prv(P_i)$ is assumed to be disjoint with the sought pattern $X$, we can get rid of it in the following manner.
Consider the consecutive vertices of $P_i$, traversed in the direction from $r$ to $z$.
Let $v_0$ be the last light terminal on $P_i$; vertex $v_0$ is well defined because $r$ is a light terminal itself.
Let $P'$ be the suffix of $P_i$ from $v_0$ to $z$ (both inclusive).
Observe that since $v_0$, being a light terminal, is at distance at most $3$ from $r$, whereas $z$ is at distance more than $2000\sqrt{k}\lg k$ from $r$,
we have that $P'$ has length at most $2000\sqrt{k}\lg k-3\geq 1997\sqrt{k}\lg k$ (here, the distance is measured as the number of non-ghost vertices traversed, minus 1).

Let $v_0,v_1,\ldots,v_\ell=z$ be the vertices of $(\Pub(P_i)\cap V(P'))\cup \{v_0,z\}$ in the order of their appearance on $P'$.
Then clearly $\ell\leq |\Pub(P_i)|+1\leq 485\sqrt{k}\lg k$. 
For each $j=0,1,\ldots,\ell-1$, inspect the segment of $P'$ lying between $v_j$ and $v_{j+1}$.
If this segment contains some ghost vertex $g_j$, then contract it entirely onto $g_j$; in case there are multiple ghost vertices in the segment, select any of them as $g_j$.
Otherwise, if there are no ghost vertices within the segment, contract this whole segment onto vertex $v_{j}$; see Fig.~\ref{fig:contr} for a visualization.
Observe that, by the definition of $v_0$, no light terminal gets contracted in this manner.

\begin{figure}[htbp]
\begin{center}
\def\svgwidth{0.66\textwidth}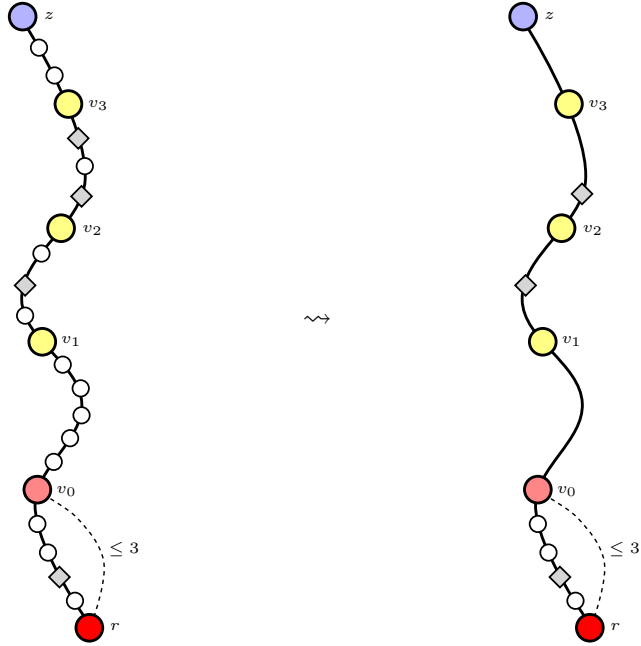
\end{center}
\caption{The contraction procedure applied on the path $P_i$. Vertices $v_1,v_2,\ldots,v_{\ell-1}$ are depicted in yellow, and small gray rhombi depict ghost vertices.}\label{fig:contr}
\end{figure}

Denote the obtained graph by $H$; this graph is equipped with a set $\ghost_H=\ghost\cap V(H)$ of ghost vertices naturally inherited from $G''$.
Since we assume that $X$ is disjoint with $\Prv(P_i)$, no contracted vertex belonged to $X$. 
Hence, it can be easily seen that $X$ is still a pattern in $H$.

Observe that $H$ has strictly less vertices than $G$ for the following reason: 
path $P'$ had at least $1997\sqrt{k}\lg k$ non-ghost vertices at the beginning, but after the contraction it got shortened to at most $485\sqrt{k}\lg k$ non-ghost vertices.
Observe also that, since $v_0$ is at distance at most $3$ from $r$ as a light terminal, we have that the distance between $r$ and $z$ in $H$ is at most $3+485\sqrt{k}\lg k\leq 488\sqrt{k}\lg k$.

Define a new instance $\Ii'=(H,r,\light_H,\heavy_H,\ghost_H,\lambda)$ by giving it the same credit $\lambda$, and taking $\light_H=\light$ and $\heavy_H=\heavy\cap V(H)$.
Observe here that we use the fact that during the construction of $H$ we did not contract any light terminal, and hence all vertices of $\light$ persist in $H$.
Clearly $\light_H$ and $\heavy_H$ are disjoint, and observe that invariants~\eqref{inv:dist} and~\eqref{inv:terminals} trivially hold for $\Ii'$, because $H$ was obtained from $G''$ by means of edge contractions.
Since no contracted vertex belongs to $X$, we have that
\begin{equation}\label{eq:case-intersect-paths-patsiz}
\patsiz_{\Ii'}(X)=\patsiz_{\Ii}(X).
\end{equation}
Observe that during the construction of $H$ we contracted at least $1511\sqrt{k}\lg k$ non-ghost vertices.
This means that the total number of vertices that are not ghost or light terminals, strictly decreases, so
\begin{equation}\label{eq:case-intersect-paths-grasiz}
\grasiz_{\Ii'}<\grasiz_{\Ii}.
\end{equation}
This also implies that $n_H<n$, where $n_H$ is the number of vertices of $H$.

We are left with analyzing the distance potential, which is factor on which we gain in this step.
More precisely, the crucial observation is that the performed contraction significantly reduces the number of far vertices.

\begin{claim}\label{cl:less-far}
The following holds:
$$\Far_{\Ii'}(X)\subseteq \Far_{\Ii}(X)\qquad\textrm{and}\qquad |\Far_{\Ii}(X)\setminus \Far_{\Ii'}(X)|\geq 511\sqrt{k}\lg k.$$
\end{claim}
\begin{proof}
Observe that each close vertex in $\Ii$ is contained in $M$, which remains intact in $G''$. 
Moreover, $H$ is obtained from $G''$ only by means of edge contractions, which can only decrease the distances.
Hence, every vertex of $X$ that is close in $\Ii$, remains close in $\Ii'$.
It follows that $\Far_{\Ii'}(X)\subseteq \Far_{\Ii}(X)$.

For the second claim, recall that in $G''$, there is some vertex $z'$ that belongs to $X$ and is at distance at most $1$ from $z$ (possibly $z'=z$). 
On the other hand, by~\eqref{eq:still-far} we have $\dist_{G''}(r,z)>2000\sqrt{k}\lg k$, so $\dist_{G''}(r,z')>1999\sqrt{k}\lg k$
Let $P$ be a path from $r$ to $z'$ whose vertices belong to $X$ or $\ghost$.
Since $\dist_{G''}(r,z')>1999\sqrt{k}\lg k$, $P$ contains at least $1999\sqrt{k}\lg k$ vertices of $X$, among which the last $511\sqrt{k}\lg k$ vertices have to belong to $\Far_{\Ii}(X)$ for the following reason:
their distance from $r$ is larger than $1999\sqrt{k}\lg k-511\sqrt{k}\lg k>1000\sqrt{k}\lg k$, by the triangle inequality.
However, in $\Ii'$ we have that the distance between $r$ and $z$ is shortened to at most $488\sqrt{k}\lg k$, and hence all of them become close.
\cqed\end{proof}

It follows that
\begin{equation}\label{eq:case-intersect-paths-dstpot}
\dstpot_{\Ii'}(X)\leq \dstpot_{\Ii'}(X)-511\sqrt{k}\lg k.
\end{equation}
Having analyzed the decrease in all the potentials, we are ready to finalize the case.

Apply the algorithm recursively to the instance $\Ii'=(H,r,\light_H,\heavy_H,\ghost_H,\lambda)$. 
As discussed earlier, $\Ii'$ satisfies the requested invariants and has strictly fewer vertices, so this recursive call is correctly defined.
The application of the algorithm yields a subset $A'$ of vertices of $H$ with the following properties:
\begin{itemize}
\item $A'\supseteq \light_H=\light$ and $H[A']$ admits a tree decomposition $\Tt'$ of width at most $24022\ctw\sqrt{k}\lg k$ with $\trms_H\cap A'$ contained in the root bag;
\item the probability that $X\subseteq A'$ is at least $\Monster(n_H,\patsiz_{\Ii'}(X),\grasiz_{\Ii'},\dstpot_{\Ii'}(X))$.
\end{itemize}
The algorithm returns the set $A:=A'$; we now verify that $A$ has all the required properties. 
Clearly we already have that $A=A'\supseteq \light_H=\light$, so let us check that $G[A]$ admits a suitable tree decomposition.

\begin{claim}\label{cl:case-intersect-paths-tw}
The subgraph $G[A]$ admits a tree decomposition of width at most $24022\ctw\sqrt{k}\lg k$ with $A\cap \trms$ contained in the root bag.
\end{claim}
\begin{proof}
We observe that decomposition $\Tt:=\Tt'$ is suitable. 
First, it can be easily verified that it is also a tree decomposition of $G[A]$, not only $H[A]$, because $G[A]$ is a graph on the same vertex set as $H[A]$ and every edge of $G[A]$ is also present in $H[A]$.
Second, from the recursive call we have that the root bag of $\Tt'$ contains $A\cap \trms_H$, but $A\cap \trms_H=A\cap \trms$, because $A$ contains only vertices that are present in $H$.
Consequently, the root bag of $\Tt$ contains $A\cap \trms$. Finally, from the recursive call we obtain that the width of $\Tt=\Tt'$ is at most $24022\sqrt{k}\lg k$.
\cqed\end{proof}

We are left with analyzing the success probability. For this, we assume $c_1\geq 1$.

\begin{claim}\label{cl:case-intersect-paths-error}
Assume $c_1\geq 1$. Supposing $X\cap W_{\mathsf{isl}}\neq \emptyset$ and the subroutine of Theorem~\ref{thm:duality} returned a sequence of paths, 
the algorithm outputs a set $A$ with $X\subseteq A$ with probability at least $\Monster(n,\patsiz(X),\grasiz,\dstpot(X))$. 
This includes the $(1-1/k)$ probability of success of the preliminary clustering step, the $1/k$ probability that the algorithm makes the correct assumption that $X\cap W_{\mathsf{isl}}=\emptyset$,
the $k^{-7}$ probability of correctly choosing the island $C$ that intersects the pattern,
the $(10k^2\lg n)^{-1}$ probability of choosing the right distance $d$, and the $1/k$ probability of choosing the right path index $i$.
\end{claim}
\begin{proof}
By the bound on the success probability of the recursive call, and the assumption that $k\geq 10$, we have that
\begin{eqnarray}
\mathbb{P}(X\subseteq A) & \geq & \left(1-\frac{1}{k}\right)\cdot k^{-8}\cdot (10k^2\lg n)^{-1}\cdot \Monster(n',\patsiz_{\Ii'}(X),\grasiz_{\Ii'},\dstpot_{\Ii'}(X))\nonumber \\
& \geq & k^{-12}\cdot (\lg n)^{-1}\cdot \Monster(n',\patsiz_{\Ii'}(X),\grasiz_{\Ii'},\dstpot_{\Ii'}(X)).\label{eq:to-subs}
\end{eqnarray}
By~\eqref{eq:case-intersect-paths-patsiz} and the fact that $n_H<n$, we have:
\begin{equation}\label{eq:est2-patsiz}
\exp\left[-c_1\cdot\frac{\lg k+\lg\lg n}{\sqrt{k}}\cdot \patsiz_\Ii(X)\lg\patsiz_\Ii(X)\right] \leq \exp\left[-c_1\cdot\frac{\lg k+\lg\lg n_H}{\sqrt{k}}\cdot \patsiz_{\Ii'}(X)\lg\patsiz_{\Ii'}(X)\right].
\end{equation}
By~\eqref{eq:case-intersect-paths-dstpot} and the facts that $c_2\geq 1$ and $n_H<n$, we have
\begin{eqnarray}
\exp\left[-c_1\cdot \frac{\lg k+\lg\lg n}{\sqrt{k}}\cdot \dstpot_\Ii(X)\right] & \leq & \exp\left[-c_1\cdot \frac{\lg k+\lg \lg n_H}{\sqrt{k}}\cdot \dstpot_{\Ii'}(X)\right]\cdot\nonumber \\
& & \exp\left[-c_1\cdot 511\lg k(\lg k+\lg \lg n)\right]\nonumber\\
& \leq & \exp\left[-c_1\cdot \frac{\lg k+\lg \lg n_H}{\sqrt{k}}\cdot \dstpot_{\Ii'}(X)\right]\cdot k^{12}\cdot \lg n.\label{eq:est2-dstpot}
\end{eqnarray}
Finally, by~\eqref{eq:case-intersect-paths-patsiz} and~\eqref{eq:case-intersect-paths-grasiz} we infer that
\begin{eqnarray}\label{eq:est2-grasiz}
\left(1-\frac{1}{k}\right)^{c_2\patsiz_\Ii(X)\lg \grasiz_\Ii} & \leq & 
\left(1-\frac{1}{k}\right)^{c_2\patsiz_{\Ii'}(X)\lg \grasiz_{\Ii'}}.
\end{eqnarray}
By multiplying~\eqref{eq:est2-patsiz},~\eqref{eq:est2-dstpot}, and~\eqref{eq:est2-grasiz}, and applying the obtained bound in~\eqref{eq:to-subs}, we infer that:
\begin{eqnarray*}
\mathbb{P}(X\subseteq A) & \geq & k^{-12}\cdot (\lg n)^{-1}\cdot \Monster(n,\patsiz_{\Ii}(X),\grasiz_{\Ii},\dstpot_{\Ii}(X))\cdot k^{12}\cdot \lg n\\
& = &\Monster(n,\patsiz_{\Ii}(X),\grasiz_{\Ii},\dstpot_{\Ii}(X)).
\end{eqnarray*}
This concludes the proof.
\cqed\end{proof}

Claim~\ref{cl:case-intersect-paths-tw} ensures that the output of the algorithm has the required properties, whereas Claim~\ref{cl:case-intersect-paths-error} yields the sought lower bound on the success probability.

\paragraph*{Subcase: nested chain of circular separators.}

In this case, the algorithm of Theorem~\ref{thm:duality} returned an $(r,z)$-separator chain $(C_1,\ldots,C_p)$ in $L=\torso{G''}{\ghost\cap V(G'')}$, 
where $p=\lceil 120\sqrt{k}\lg k\rceil$ and $|C_i|\leq 2k$ for each $i\in [p]$.
Obviously, by the definition of $L$ we have that $(C_1,\ldots,C_p)$ is also an $(r,z)$-separator chain in $G''$, and no vertex of any $C_i$ is a ghost vertex.
Recall that this means that all separators $C_i$ are pairwise disjoint and $\reach(r,G''-C_i)\subseteq \reach(r,G''-C_j)$ whenever $1\leq i<j\leq p$.
By invariant~\eqref{inv:dist}, at most $3$ first separators may include a light terminal, hence after excluding them we are left with at least $\lceil 117\sqrt{k}\lg k\rceil$ separators without any light terminals.
We restrict our attention to these separators.
Thus, by slightly abusing the notation, from now on we work with a $(r,z)$-separation chain $(C_1,\ldots,C_{p'})$, where $p'=\lceil 117\sqrt{k}\lg k\rceil$ 
such that each $C_i$ is disjoint with $\light\cup \ghost$ and, in fact, $\light\subseteq \reach(r,G''-C_1)$.

For $i\in [p']$, we define the following sets:
\begin{equation*}
V^\pin_i = \reach(r,G''-C_i) \qquad \textrm{and} \qquad V^\pout_i = V(G'')\setminus (C_i\cup V^\pin_i).
\end{equation*}
Thus, $(V^\pin_i,C_i,V^\pout_i)$ is a partition of $V(G'')$.
Without loss of generality we can assume that each separator $C_i$ is inclusion-wise minimal, which implies that each vertex of $C_i$ has a neighbor in $V^\pin_i$ and a neighbor in $V^\pout_i$.

We now prove that one of the separators $C_i$ has the property that it splits $X$ in a balanced way, relatively to the number of vertices of $X$ it contains.

\begin{claim}\label{cl:balanced-index}
There is an index $i\in [p']$ for which the following holds:
\begin{equation*}
10\sqrt{k}\cdot |X\cap C_i|\leq \min(|(X\cap V^\pin_i)\setminus \light|,|(X\cap V^\pout_i)\setminus \light|).
\end{equation*}
\end{claim}
\begin{proof}
For $i\in [p']$, let 
$$a_i=|(X\cap V^\pin_i)\setminus \light|,\qquad \textrm{and}\qquad b_i=|(X\cap V^\pout_i)\setminus \light|,\qquad \textrm{and}\qquad c_i=|X\cap C_i|.$$
Observe that since all light terminals are within $\reach(r,G''-C_1)$, for each $i\in [p']$ the following holds:
\begin{eqnarray}
a_i & \geq & \sum_{j<i} c_j;\label{eq:ai}\\
b_i & \geq & \sum_{j>i} c_j.
\end{eqnarray}
Observe that each separator $C_i$ has to contain a vertex of $X$, because $X$ contains a non-ghost vertex at distance at most $1$ from $z$, 
and this vertex can be reached from $r$ by a path that uses only ghost vertices and vertices of $X$.
We conclude that $c_i\geq 1$ for each $i\in [p']$.
Consequently, $a_i\geq 1$ for each $i\geq 2$, and $b_i\geq 1$ for each $i\leq p'-1$.

For the sake of contradiction, suppose that 
\begin{equation}\label{eq:contra}
c_i>\frac{\min(a_i,b_i)}{10\sqrt{k}}\qquad \textrm{for all $i$.}
\end{equation}
Obviously, the sequence $(a_i)_{i\in [p']}$ is non-decreasing and the sequence $(b_i)_{i\in [p']}$ is non-increasing,
Let $i_0$ be the smallest index such that $a_{i_0}>b_{i_0}$; possibly $i_0=p'+1$ if the condition $a_i\leq b_i$ is satisfied for all $i\in [p']$.
We claim that in fact $i_0\leq 53\sqrt{k}\lg k$; suppose otherwise.
By assumption~\eqref{eq:contra} and the definition of $i_0$ we have that $c_i>a_{i}/(10\sqrt{k})$ for all $i<i_0$, .
Therefore, by combining this with~\eqref{eq:ai}, we obtain that
$$a_i > \frac{1}{10\sqrt{k}}\sum_{j<i} a_j$$
for all $i<i_0$. Equivalently,
$$\sum_{j\leq i} a_j > \left(1+\frac{1}{10\sqrt{k}}\right)\cdot \sum_{j<i} a_j.$$
Since $a_2\geq 1$, we infer by a trivial induction that
$$\sum_{j<i} a_j\geq \left(1+\frac{1}{10\sqrt{k}}\right)^{i-2},$$
for all $2\leq i<i_0$. Therefore, we conclude that
\begin{eqnarray*}
a_{i_0-1} & > & \frac{1}{10\sqrt{k}}\cdot \left(1+\frac{1}{10\sqrt{k}}\right)^{53\sqrt{k}\lg k-3}\geq \frac{1}{10\sqrt{k}}\cdot \left(1+\frac{1}{10\sqrt{k}}\right)^{50\sqrt{k}\lg k}\\
& \geq & \frac{1}{10\sqrt{k}}\cdot e^{5\lg k}>k.
\end{eqnarray*}
This is a contradiction with $|X|\leq k$. Hence, we have that indeed $i_0\leq 53\sqrt{k}\lg k$.

By applying a symmetric reasoning for the last separators and numbers $b_i$, instead of the first and numbers $a_i$, 
we obtain that if $i_1$ is the largest index such that $a_{i_1}<b_{i_1}$, then $i_1\geq 64\sqrt{k}\lg k$.
However, this means that $i_0<i_1$, which is a contradiction with the fact that sequences $(a_i)_{i\in [p']}$ and $(b_i)_{i\in [p']}$ are non-increasing and non-decreasing, respectively.
\cqed\end{proof}
Observe that if an index $i$ satisfies the property given by Claim~\ref{cl:balanced-index}, then $|X\cap C_i|\leq \sqrt{k}/10$. Indeed, otherwise we would have that
$\min(|X\cap V^\pin_i|,|X\cap V^\pout_i|)>k$, which is a contradiction with $|X|\leq k$.

The algorithm performs random sampling as follows:
\begin{itemize}
\item First, it samples an index $i\in [p]$ uniformly at random, and assumes that this index $i$ satisfies the property given by Claim~\ref{cl:balanced-index}.
\item Then, it samples an integer $\alpha$ between $1$ and $\sqrt{k}/10$, and assumes that the sampled number $\alpha$ is equal to $|X\cap C_i|$.
\item Finally, the algorithm samples a subset $Q\subseteq C_i$ of size $\alpha$ uniformly at random, and assumes $Q$ to be equal to $X\cap C_i$.
\end{itemize}
As $|C_i|\leq 2k$, we observe that the assumptions stated above are correct with probability at least
$$\frac{1}{p}\cdot \frac{10}{\sqrt{k}}\cdot \frac{1}{\binom{|C_i|}{\alpha}}\geq \left(k^5\cdot \binom{2k}{\alpha}\right)^{-1}\geq k^{-2\alpha-5},$$
where $\alpha=|X\cap C_i|$. Keeping this success probability assumption in mind, we proceed further with the supposition that the sampled objects are indeed as assumed.

The algorithm now defines two subinstances $\Ii_\pout$ and $\Ii_\pin$ as follows.

First, we define $\Ii_\pout=(G_\pout,r,\light_\pout,\heavy_\pout,\ghost_\pout,\lambda+\alpha)$; note that the guessed size of $Q$ is added to the credit. 
Note that $\lambda+\alpha\leq \sqrt{k}/10+\sqrt{k}/10$, so the instance on which we shall recurse will have credit at most $\sqrt{k}/5$.
Define $G_\pout$ to be the graph constructed as follows: 
take $G''$, and contract the whole subgraph induced by $V^\pin_i\cup (C_i\setminus Q)$ onto $r$.
Observe that since $G''[V^\pin_i]$ is connected by definition, and each vertex of $C_i$ has a neighbor in $V^\pin_i$, the contracted subgraph is indeed connected.

The ghost vertices are just inherited from the original instance: we put $\ghost_\pout=\ghost\cap V(G_\pout)$.
The sets of light and heavy terminals $\light_\pout$ and $\heavy_\pout$ are defined as follows.
First, heavy terminals are inherited, but we remove all heavy terminals that reside in $Q$: we put $\heavy_\pout=\heavy\cap (V(G_\pout)\setminus Q)$.
Second, as light terminals we put $r$ plus the whole set $Q$: $\light_\pout=\{r\}\cup Q$.
Recall that $\light\subseteq V^\pin_1\subseteq V^\pin_i$, so all the light terminals of the original instance, apart from $r$, got contracted onto $r$ during the construction of $G^{\pout}$;
this is why we do not need to inherit any of them in $\Ii_\pout$.
Clearly, $\heavy_\pout$ and $\light_\pout$ defined in this manner are disjoint.
Note that we have that $|\trms_\pout|\leq |\trms|+|Q|\leq 16014\ctw\sqrt{k}\lg k+\lambda+|Q|$ and $|Q|=\alpha$, so we indeed have that $|\trms_\pout|\leq 16014\ctw\sqrt{k}\lg k+(\lambda+\alpha)$; 
that is, invariant~\eqref{inv:terminals} is satisfied in the new instance. Invariant~\eqref{inv:dist} is also satisfied, because all new light terminals are adjacent to the root $r$.

Finally, we define $X_{\pout}=X\cap V(G_\pout)$.
Since $G_\pout$ was obtained from $G''$ only by contracting some vertices onto the root, it still holds that every vertex of $X_{\pout}$ can be reached from $r$ by a path traversing only ghost vertices and
vertices of $X_{\pout}$. 
Observe also that Claim~\ref{cl:balanced-index} implies that at least $10\sqrt{k}\cdot \alpha$ vertices of $X$ that are not light terminals are contained
in $V^{\pin}_i$. These vertices do not remain in $X^{\pout}$, and hence:
$$|X_{\pout}|\leq |X|-10\sqrt{k}\cdot \alpha\leq k-10\sqrt{k}\cdot \lambda-10\sqrt{k}\cdot \alpha = k-10\sqrt{k}\cdot (\lambda+\alpha).$$
Therefore, we conclude that $X_{\pout}$ is a pattern in the instance $\Ii_\pout$. 

By applying the algorithm recursively to the instance $\Ii_\pout$, we obtain a subset of vertices $A_\pout$ with the following properties:
\begin{itemize}
\item $A_\pout\supseteq \light_\pout$ and $G_\pout[A_\pout]$ admits a tree decomposition $\Tt_\pout$ of width at most $24022\ctw \sqrt{k}\lg k$ with $A_\pout\cap \trms_\pout$ contained in the root bag.
\item The probability that $X_\pout\subseteq A_\pout$ is at least $\Monster(n_\pout,\patsiz_{\Ii_\pout}(X_\pout),\grasiz_{\Ii_\pout},\dstpot_{\Ii_\pout}(X_\pout))$, 
      where $n_\pout$ is the number of vertices on $G_\pout$.
\end{itemize}

Second, we define $\Ii_\pin=(G_\pin,r,\light_\pin,\heavy_\pin,\ghost_\pin,\lambda+\alpha)$; again the guessed size of $Q$ is added to the credit.
Note that, again, $\lambda+\alpha\leq \sqrt{k}/10+\sqrt{k}/10$, so the instance on which we shall recurse will have credit at most $\sqrt{k}/5$.
Graph $G_\pin$ is constructed from $G''$ as follows.
Inspect the connected components of the graph $G''-(V^\pin_i\cup Q)$. For each such component $D$, contract it onto a new vertex $g_D$ that is declared to be a ghost vertex.
That is, we define $\ghost_\pin$ to be $(\ghost\cap V^\pin_i)\cup \{g_D\colon D\in \cc(G''-(V^\pin_i\cup Q))\}$.

Next, we define the terminal sets $\light_\pin$, $\heavy_\pin$.
Recall that $\light\subseteq V^\pin_i$, so all the original light terminals persist in the graph $G_\pin$.
Hence, the light terminals are defined as simply inherited from the original instance: $\light_\pin=\light$.
For the heavy terminals, we take all the ones inherited from the original instance, plus we add all vertices of $Q$ explicitly: $\heavy_\pin=(\heavy\cap V(G_\pin))\cup Q$.
Note that $\light_\pin$ and $\heavy_\pin$ are disjoint, because there was no light terminal in $Q$; that is, $Q\cap \light=\emptyset$.
Again, we have that $|\trms_\pin|\leq |\trms|+|Q|\leq 16014\ctw\sqrt{k}\lg k+\lambda+|Q|$ and $|Q|=\alpha$, so we indeed have that $|\trms_\pin|\leq 16014\ctw\sqrt{k}\lg k+(\lambda+\alpha)$; 
that is, invariant~\eqref{inv:terminals} is satisfied in the new instance. Invariant~\eqref{inv:dist} is also satisfied, since edge contractions could only have made the light terminals closer to the root.

Finally, we take $X_{\pin}=X\cap V(G_\pin)$.
We observe that each vertex of $X_\pin$ can be reached from $r$ in $G_\pin$ by a path that uses only ghost vertices and vertices of $X_\pin$.
Indeed, there is such a path in $G''$, and its parts that lie outside of $V(G_\pin)$ must be contained in the connected components of $G''-(V^\pin_i\cup Q)$, 
so they can be replaced by the traversal of the ghost vertices into which these connected components are collapsed.
Next, from Claim~\ref{cl:balanced-index} we infer that
$$|X^{\pin}|\leq |X|-10\sqrt{k}\cdot \alpha\leq k-10\sqrt{k}\cdot \lambda-10\sqrt{k}\cdot \alpha = k-10\sqrt{k}\cdot (\lambda+\alpha).$$
Hence, we conclude that $X^{\pin}$ is a pattern in $\Ii_\pin$.

Again, we apply the algorithm recursively to instance $\Ii_\pin$, thus obtaining a subset of vertices $A_\pin$ with the following properties:
\begin{itemize}
\item $A_\pin\supseteq \light_\pin$ and $G_\pin[A_\pin]$ admits a tree decomposition $\Tt_\pin$ of width at most $24022\ctw \sqrt{k}\lg k$ with $A_\pin\cap \trms_\pin$ contained in the root bag.
\item The probability that $X_\pin\subseteq A_\pin$ is at least $\Monster(n_\pin,\patsiz_{\Ii_\pin}(X_\pin),\grasiz_{\Ii_\pin},\dstpot_{\Ii_\pin}(X_\pin))$, where $n_\pin$ is the number of vertices on $G_\pin$.
\end{itemize}

Observe that the sets of non-ghost vertices of $G_\pout$ and $G_\pin$ are contained in the vertex set of $G''$, and hence we can treat $A_\pout$ and $A_\pin$ also as subsets of non-ghost vertices of $G''$.
Hence, let us define $A:=(A_\pout \setminus Q) \cup A_\pin$ and declare that the algorithm returns $A$ as the answer.
Note that here, again as in the case of Section~\ref{sss:disjoint}, we formally 
allow the instance $\Ii_\pin$ to exclude the vertices of $Q$ from $A_\pin$, since they are heavy terminals there.

We now verify that $A$ has the required properties.
First, since we have that $\light_\pin=\light$, then $A\supseteq A_\pin\supseteq \light_\pin=\light$, so $A$ indeed covers all light terminals.
We now check that $G[A]$ admits a suitable tree decomposition.

\begin{claim}\label{cl:case-intersect-cycles-tw}
The subgraph $G[A]$ admits a tree decomposition of width at most $24022\ctw\sqrt{k}\lg k$ with $A\cap \trms$ contained in the root bag.
\end{claim}
\begin{proof}
Construct the root node and associate with it the bag $(\trms \cup Q) \cap A$. Since $|\trms|\leq 16014\ctw\sqrt{k}\lg k+\lambda$ and $\lambda,|Q|\leq\sqrt{k}/10$, it follows that this bag has size at most 
$16015\ctw\sqrt{k}\lg k$.
Let us restrict decomposition $\Tt_\pout$ to the vertex set $A\cap A_\pout$; that is, remove all vertices of $A_\pout\setminus A$ from all bags of $\Tt_\pout$, thus obtaining a tree decomposition
$\Tt'_\pout$ of $G_\pout[A_\pout\cap A]$.
We have that $A_\pin\subseteq A$, so there is no need of restricting decomposition $\Tt_\pin$.
Finally, attach decompositions $\Tt_\pin$ and $\Tt'_\pout$ as children of the root bag. 
It can be easily verified that in this manner we obtain a tree decomposition of $G[A]$, and
its width is clearly at most $24022\ctw\sqrt{k}\lg k$. Finally, the root bag contains $A\cap \trms$ by its definition.
\cqed\end{proof}

We are left with estimating the success probability. Before we proceed with the final calculation, let us analyze each of the potentials.
Note that graphs $G_\pin$ and $G_\pout$ intersect only at $Q\cup \{r\}$, and each vertex of $Q$ is a heavy terminal in $\Ii_\pin$ and a light terminal in $\Ii_\pout$. 
This observation will be crucial in the forthcoming analysis. Recall also that $Q$, as a subset of $C_i$, contains no original light terminal, i.e., $Q\cap \light=\emptyset$.

Firstly, in $\Ii_\pout$ we contracted all vertices of $V^\pin_i\cup (C_i\setminus Q)$, and in $\Ii_\pin$ we contracted all vertices of $V^{\pout}_i\cup (C_i\setminus Q)$. 
Among the vertices shared by the instances, being $\{r\}\cup Q$, $r$ is a light terminal in both instances, whereas the vertices of $Q$ are heavy terminals only in $\Ii_\pin$.
From this it immediately follows that
\begin{equation}\label{eq:dgrasiz}
\grasiz_\Ii\geq \grasiz_{\Ii_\pout}+\grasiz_{\Ii_\pin}.
\end{equation}
Observe that both $G_\pout$ and $G_\pin$ are constructed from $G$ by means of edge contractions only, which can only decrease the distances from $r$.
Hence, a vertex of $X$ that was close in the original instance $\Ii$, remains close in the instance $\Ii_\pin$ or $\Ii_\pout$ to which it belongs.
The vertices of $Q$ are adjacent to the root in $\Ii_\pout$, and hence none of them can be a far vertex of $X_\pout$. Hence it follows that
\begin{equation}\label{eq:ddstpot}
\dstpot_\Ii(X)\geq \dstpot_{\Ii_\pout}(X_\pout)+\dstpot_{\Ii_\pin}(X_\pin).
\end{equation}
Finally, the vertices of $Q$---shared among the instance---are declared light terminals in $\Ii_\pin$, and hence the same analysis yields that
\begin{equation}\label{eq:dpatsiz}
\patsiz_\Ii(X)\geq \patsiz_{\Ii_\pout}(X_\pout)+\patsiz_{\Ii_\pin}(X_\pin).
\end{equation}
We now perform a finer analysis of the behaviour of potential $\patsiz$. For this, we use Claim~\ref{cl:balanced-index} as follows.

\begin{claim}\label{cl:dpatsize}
The following holds:
\begin{equation}
\patsiz_\Ii(X)\lg \patsiz_\Ii(X)\geq \patsiz_{\Ii_\pout}(X_\pout)\lg \patsiz_{\Ii_\pout}(X_\pout)+\patsiz_{\Ii_\pin}(X_\pin)\lg \patsiz_{\Ii_\pin}(X_\pin) + 10\sqrt{k}\cdot \alpha.
\end{equation}
\end{claim}
\begin{proof}
Observe that
\begin{eqnarray*}
\patsiz_{\Ii_\pout}(X_\pout) & = & |(X\cap V^\pout_i)\setminus \light|;\\
\patsiz_{\Ii_\pin}(X_\pin) & = & |(X\cap V^\pin_i)\setminus \light|+|X\cap C_i|.
\end{eqnarray*}
Thus we have
\begin{equation}\label{eq:dbasic}
\patsiz_{\Ii_\pout}(X_\pout)+\patsiz_{\Ii_\pin}(X_\pin)=|(X\cap V^\pout_i)\setminus \light|+|(X\cap V^\pin_i)\setminus \light|+|X\cap C_i|=|X\setminus \light|=\patsiz_\Ii(X).
\end{equation}
Suppose first that $\patsiz_{\Ii_\pout}(X_\pout)\leq \patsiz_{\Ii_\pin}(X_\pin)$; the second case is symmetric.
Combining this with~\eqref{eq:dbasic} yields the following:
\begin{eqnarray*}
\patsiz_{\Ii_\pout}(X_\pout)\leq \patsiz_{\Ii}(X)/2\qquad \textrm{and}\qquad \patsiz_{\Ii_\pin}(X_\pin)\leq \patsiz_{\Ii}(X).
\end{eqnarray*}
By Claim~\ref{cl:balanced-index}, we infer that
\begin{equation*}
10\sqrt{k}\cdot \alpha \leq |(X\cap V^\pout_i)\setminus \light| \leq \patsiz_{\Ii_\pout}(X_\pout). 
\end{equation*}
Putting all these together, we observe that:
\begin{eqnarray*}
\patsiz_\Ii(X)\lg \patsiz_\Ii(X) & \geq & \patsiz_{\Ii_\pout}(X_\pout)\lg \patsiz_\Ii(X)+\patsiz_{\Ii_\pin}(X_\pin)\lg \patsiz_\Ii(X)\\
& \geq & \patsiz_{\Ii_\pout}(X_\pout)\cdot (1+\lg \patsiz_{\Ii_\pout}(X_\pout))+\patsiz_{\Ii_\pin}(X_\pin)\lg \patsiz_{\Ii_\pin}(X_\pin)\\
& \geq & \patsiz_{\Ii_\pout}(X_\pout)\lg \patsiz_{\Ii_\pout}(X_\pout)+\patsiz_{\Ii_\pin}(X_\pin)\lg \patsiz_{\Ii_\pin}(X_\pin) + \patsiz_{\Ii_\pout}(X_\pout)\\
& \geq  & \patsiz_{\Ii_\pout}(X_\pout)\lg \patsiz_{\Ii_\pout}(X_\pout)+\patsiz_{\Ii_\pin}(X_\pin)\lg \patsiz_{\Ii_\pin}(X_\pin) + 10\sqrt{k}\cdot \alpha.
\end{eqnarray*}
This is exactly the claimed inequality. As mentioned before, the case when $\patsiz_{\Ii_\pout}(X_\pout)\geq \patsiz_{\Ii_\pin}(X_\pin)$ is symmetric.
\cqed\end{proof}

Finally, we can proceed with the final success probability analysis. For this, we can take any $c_1\geq 2$.

\begin{claim}\label{cl:case-intersect-cycles-error}
Assume $c_1\geq 2$. Supposing $X\cap W_{\mathsf{isl}}\neq \emptyset$ and the subroutine of Theorem~\ref{thm:duality} returned a separator chain, 
the algorithm outputs a set $A$ with $X\subseteq A$ with probability at least $\Monster(n,\patsiz(X),\grasiz,\dstpot(X))$. 
This includes the $(1-1/k)$ probability of success of the preliminary clustering step, the $1/k$ probability that the algorithm makes the correct assumption that $X\cap W_{\mathsf{isl}}=\emptyset$,
the $k^{-7}$ probability of correctly choosing the island $C$ that intersects the pattern,
the $(10^2\lg n)^{-1}$ probability of choosing the right distance $d$, and $k^{-2\alpha-5}$ probability of correctly sampling the $i$, $\alpha$, and set $Q$.
\end{claim}
\begin{proof}
We denote by $n_\pin$ and $n_\pout$ the numbers of vertices in $G_\pin$ and $G_\pout$, respectively; note that $n_\pin,n_\pout\leq n$.
From the probability of the success of recursive calls, we infer that
\begin{eqnarray}
\mathbb{P}(X\subseteq A) & \geq & \left(1-\frac{1}{k}\right)\cdot k^{-8}\cdot (10k^2\lg n)^{-1}\cdot k^{-2\alpha-5}\cdot\nonumber\\
& & \Monster(n_\pout,\patsiz_{\Ii_\pout}(X_\pout),\grasiz_{\Ii_\pout},\dstpot_{\Ii_\pout}(X_\pout)) \cdot \Monster(n_\pin,\patsiz_{\Ii_\pin}(X_\pin),\grasiz_{\Ii_\pin},\dstpot_{\Ii_\pin}(X_\pin))\nonumber\\
& \geq & k^{-2\alpha-17}\cdot (\lg n)^{-1}\cdot\nonumber \\
& & \Monster(n_\pout,\patsiz_{\Ii_\pout}(X_\pout),\grasiz_{\Ii_\pout},\dstpot_{\Ii_\pout}(X_\pout)) \cdot \Monster(n_\pin,\patsiz_{\Ii_\pin}(X_\pin),\grasiz_{\Ii_\pin},\dstpot_{\Ii_\pin}(X_\pin))\label{eq:est3-main}
\end{eqnarray}
From~\eqref{eq:dgrasiz} and~\eqref{eq:dpatsiz} we infer that
\begin{eqnarray}\label{eq:est3-grasiz}
\left(1-\frac{1}{k}\right)^{c_2\patsiz_\Ii(X)\lg \grasiz_\Ii} & \leq & 
\left(1-\frac{1}{k}\right)^{c_2\patsiz_{\Ii_\pout}(X_\pout)\lg \grasiz_{\Ii_\pout}}\cdot 
\left(1-\frac{1}{k}\right)^{c_2\patsiz_{\Ii_\pin}(X_\pin)\lg \grasiz_{\Ii_\pin}}.
\end{eqnarray}
Similarly, from~\eqref{eq:ddstpot} and the fact that $n_\pin,n_\pout\leq n$, we infer that: 
\begin{eqnarray}
\exp\left[-c_1\cdot \frac{\lg k+\lg\lg n}{\sqrt{k}}\cdot \dstpot_\Ii(X)\right] & \leq & \exp\left[-c_1\cdot \frac{\lg k+\lg \lg n_\pout}{\sqrt{k}}\cdot \dstpot_{\Ii_\pout}(X_\pout)\right]\cdot\nonumber \\
& & \exp\left[-c_1\cdot \frac{\lg k+\lg \lg n_\pin}{\sqrt{k}}\cdot \dstpot_{\Ii_\pout}(X_\pin)\right].\label{eq:est3-dstpot}
\end{eqnarray}
Finally, from Claim~\ref{cl:dpatsize} we have that
\begin{eqnarray}
\exp\left[-c_1\cdot \frac{\lg k+\lg\lg n}{\sqrt{k}}\cdot \patsiz_\Ii(X)\lg \patsiz_\Ii(X)\right] & \leq & 
\exp\left[-c_1\cdot \frac{\lg k+\lg \lg n_\pout}{\sqrt{k}}\cdot \patsiz_{\Ii_\pout}(X_\pout)\lg \patsiz_{\Ii_\pout}(X_\pout)\right]\cdot\nonumber \\
& & \exp\left[-c_1\cdot \frac{\lg k+\lg \lg n_\pin}{\sqrt{k}}\cdot \patsiz_{\Ii_\pin}(X_\pin)\lg \patsiz_{\Ii_\pin}(X_\pin)\right]\cdot\nonumber \\
& & \exp\left[-c_1\cdot \frac{\lg k+\lg \lg n}{\sqrt{k}}\cdot 10\alpha\sqrt{k}\right].\label{eq:est3-patsiz}
\end{eqnarray}
Let us analyze the last factor of the right hand side of~\eqref{eq:est3-patsiz}, keeping in mind that $c_1\geq 2$.
\begin{eqnarray}
\exp\left[-c_1\cdot \frac{\lg k+\lg \lg n}{\sqrt{k}}\cdot 10\alpha\sqrt{k}\right] & = & \exp\left[-c_1\cdot 10\alpha\cdot \lg k- c_1\cdot 10\alpha\cdot \lg \lg n\right]\nonumber\\
& \leq & k^{-20\alpha}\cdot (\lg n)^{-1} \leq k^{-2\alpha-17}\cdot (\lg n)^{-1}.\label{eq:est3-patsize-aux}
\end{eqnarray}
Finally, by plugging~\eqref{eq:est3-grasiz},~\eqref{eq:est3-dstpot},~\eqref{eq:est3-patsiz}, and~\eqref{eq:est3-patsize-aux} into~\eqref{eq:est3-main} with $\Monster$ function expanded, 
and recognizing the expression $\Monster(n,\patsiz(X),\grasiz,\dstpot(X))$, we obtain
\begin{eqnarray*}
\mathbb{P}(X\subseteq A) & \geq & \Monster(n,\patsiz(X),\grasiz,\dstpot(X)),
\end{eqnarray*}
which is exactly what we needed to prove.
\cqed\end{proof}

Claim~\ref{cl:case-intersect-cycles-tw} ensures that the output of the algorithm has the required properties, whereas Claim~\ref{cl:case-intersect-cycles-error} yields the sought lower bound on the success probability.

\newcounter{qcount}

\newcommand{\question}[1]{
\medskip 
\noindent\begin{tabular}{ p{3em}@{}p{0.89\linewidth} }
{\bf{Q\stepcounter{qcount}\theqcount.}} & #1
\end{tabular}
\medskip
}

\section{Conclusions}\label{sec:conc}

In this work we have laid foundations for a new tool for obtaining subexponential parameterized algorithms for problems on planar graphs, and more generally on graphs that exclude a fixed apex graph as a minor.
The technique is applicable to problems that can be expressed as searching for a small, connected pattern in a large host graph.
Using the new approach, we designed, in a generic manner, a number of subexponential parameterized algorithms for problems for which the existence of such algorithms was open.
We believe, however, that this work provides only the basics of a new methodology for the design of parameterized algorithms on planar and apex-minor-free graphs. 
This methodology goes beyond the paradigm of bidimensionality and is yet to be developed.

An immediate question raised by our work is whether the technique can be derandomized. Note that our main result, Theorem~\ref{thm:maintheorem}, immediately yields the following combinatorial statement.

\begin{theorem}\label{thm:maincomb}
Let $\Cc$ be a class of graphs that exclude a fixed apex graph as a minor.
Suppose $G$ is an $n$-vertex graph from $\Cc$ and $k$ is a positive integer.
Then there exists a family $\Ff$ of subsets of vertices $G$, with the following properties satisfied:
\begin{enumerate}[(P1)]
\item For each $A\in \Ff$, the treewidth of $G[A]$ is at most $\Oh(\sqrt{k}\log k)$.
\item For each vertex subset $X\subseteq V(G)$ such that $G[X]$ is connected and $|X|\leq k$, there exists some $A\in \Ff$ for which $X\subseteq A$.
\item It holds that $|\Ff|\leq 2^{\Oh(\sqrt{k}\log^2 k)}\cdot n^{\Oh(1)}$.
\end{enumerate}
\end{theorem}
\begin{proof}
Let $f(n,k)\in 2^{\Oh(\sqrt{k}\log^2 k)}\cdot n^{\Oh(1)}$ be the inverse of the lower bound on the success probability of the algorithm of Theorem~\ref{thm:maintheorem}.
Repeat the algorithm of Theorem~\ref{thm:maintheorem} $f(n,k)\cdot 2k\ln n$ times, and consider the list of obtained vertex subsets as a candidate for $\Ff$.
Let us fix some $X\subseteq V(G)$ such that $|X|\leq k$ and $G[X]$ is connected, and consider the probability that there is some $A\in \Ff$ for which $X\subseteq A$.
For one particular run of the algorithm of Theorem~\ref{thm:maintheorem}, this holds with probability at least $f(n,k)^{-1}$.
As the runs are independent, the probability that no element of $\Ff$ covers $X$ is upper bounded by
$$\left(1-\frac{1}{f(n,k)}\right)^{f(n,k)\cdot 2k\ln n}\leq e^{-2k\ln n}=n^{-2k}.$$
As the number of $k$-vertex subsets of $V(G)$ is upper bounded by $n^k$, we infer that the expected value of the number of sets $X$ for which there is no element of $\Ff$ covering them, is upper bounded by
$n^{-k}<1$. Hence, there is a run of the described experiment for which this number is zero; this run yields the desired family $\Ff$.
\end{proof}

The above proof of Theorem~\ref{thm:maincomb} gives only a randomized algorithm constructing a family $\Ff$ that indeed covers all small patterns with high probability.
We conjecture that the algorithm can be derandomized; that is, that the family whose existence is asserted by Theorem~\ref{thm:maincomb} can be computed in time $2^{\Oh(\sqrt{k}\log^2 k)}\cdot n^{\Oh(1)}$.
So far we are able to derandomize most of the components of the algorithm, primarily using standard constructions based on splitters and perfect hash families~\cite{NaorSS95}.
However, one part of the reasoning with which we still struggle is the clustering step (Theorem~\ref{thm:clust}).

\question{Is it possible to construct a family with properties described in Theorem~\ref{thm:maincomb} in deterministic time $2^{\Oh(\sqrt{k}\log^c k)}\cdot n^{\Oh(1)}$, for some constant $c$?}

In Appendix~\ref{sec:minor-free} we attempt to generalize our technique to the case when the pattern is disconnected, and when the class only excludes some fixed (but arbitrary) graph $H$ as a minor.
In the case of disconnected patterns, we were able to prove a suitable generalization of Theorem~\ref{thm:maintheorem}, however the success probability of the algorithm depends
inverse-exponentially on number of connected components of the pattern. In the case of general $H$-minor-free classes, we needed to assume that the pattern admits a spanning tree of constant maximum degree.
So far we do not see any reason for any of these constraints to be necessary.

\question{Is it possible to prove Theorem~\ref{thm:maintheorem} without the assumption that the subgraph induced by $X$ has to be connected?}

\question{Is it possible to prove Theorem~\ref{thm:maintheorem} only under the assumption that all graphs from $\Cc$ exclude some fixed (but arbitrary) graph $H$ as a minor?}

Our next question concerns local search problems in the spirit of the \probLSVC problem considered in Section~\ref{sec:intro}. 
Apart from this problem, Fellows et al.~\cite{FellowsFLRSV12} designed FPT algorithms also for the local search for a number of other problems on apex-minor-free classes, including \probLSDS and its distance-$d$ generalization.
Here, we are given a dominating set $S$ in a graph $G$ from some apex-minor-free class $\Cc$, and we ask whether there exists a strictly smaller dominating set $S'$ that is at Hamming distance at most $k$ from $S$.
Again, the approach of Fellows et al.~\cite{FellowsFLRSV12} is based on the observation that if there is some solution, then there is also a solution $S'$ such that $S\triangle S'$ can be connected using at most
$k$ additional vertices. Thus, we need to search for a connected pattern of size $2k$, instead of $k$, in which suitable sets $S\setminus S'$ and $S'\setminus S$ are to be found.
Unfortunately, now the preprocessing step fails: vertices outside $A$ may require to be dominated from within $A$, which poses additional constraints that are not visible in the graph $G[A]$ only. 
Hence, we cannot just focus on the graph $G[A]$.
Observe, however, that the whole reasoning would go through if $A$ covered not just $S\triangle S'$, but also its neighborhood.
More generally, if the considered problem concerns domination at distance $d$, then we should cover the distance-$d$ neighborhood of $S\triangle S'$.
This motivates the following question.

\question{Fix some positive constant $d$. Is it possible to prove a stronger version of Theorem~\ref{thm:maintheorem}, where the sampled set $A$ is required to cover the whole
distance-$d$ neighborhood of the set $X$ with the same asymptotic lower bound on the success probability?}

Finally, so far we do not know whether the connectivity condition in Theorem~\ref{thm:SItwo} is necessary.

\question{Is it possible to solve \probSI on planar graphs in time $2^{\Oh(k/\log{k})}\cdot n^{\Oh(1)}$, even without the assumption that the pattern graph is connected?}

Note that a positive answer to Q2 implies a positive answer here as well, as the algorithm of Theorem~\ref{thm:bod-private} does not require the pattern graph to be connected.

\myparagraph{Acknowledgements.} The authors thank Felix Reidl, Fernando S\'anchez Villaamil, Somnath Sikdar, and Yngve Villanger for some preliminary discussions on the topic of this paper.

\bibliographystyle{abbrv}
\bibliography{pat-cov,book_pc}

\begin{thebibliography}{10}

\bibitem{AlberBFKN02}
J.~Alber, H.~L. Bodlaender, H.~Fernau, T.~Kloks, and R.~Niedermeier.
\newblock Fixed parameter algorithms for dominating set and related problems on
  planar graphs.
\newblock {\em Algorithmica}, 33(4):461--493, 2002.

\bibitem{AlonYZ}
N.~Alon, R.~Yuster, and U.~Zwick.
\newblock Color-coding.
\newblock {\em J. ACM}, 42(4):844--856, 1995.

\bibitem{BodlaenderCKN15}
H.~L. Bodlaender, M.~Cygan, S.~Kratsch, and J.~Nederlof.
\newblock Deterministic single exponential time algorithms for connectivity
  problems parameterized by treewidth.
\newblock {\em Inf. Comput.}, 243:86--111, 2015.

\bibitem{Bod-private}
H.~L. Bodlaender, J.~Nederlof, and T.~van~der Zanden.
\newblock Private communication, see also slides
  at~\url{http://www.lorentzcenter.nl/lc/web/2015/701/presentations/Bodlaender%
.pdf}.

\bibitem{borradaile_et_al:DR:2014:4427}
G.~Borradaile, P.~Klein, D.~Marx, and C.~Mathieu.
\newblock {Algorithms for Optimization Problems in Planar Graphs (Dagstuhl
  Seminar 13421)}.
\newblock {\em Dagstuhl Reports}, 3(10):36--57, 2014.

\bibitem{ChimaniMZ12}
M.~Chimani, P.~Mutzel, and B.~Zey.
\newblock Improved {S}teiner tree algorithms for bounded treewidth.
\newblock {\em J. Discrete Algorithms}, 16:67--78, 2012.

\bibitem{cygan2015parameterized}
M.~Cygan, F.~V. Fomin, L.~Kowalik, D.~Lokshtanov, D.~Marx, M.~Pilipczuk,
  M.~Pilipczuk, and S.~Saurabh.
\newblock {\em Parameterized Algorithms}.
\newblock Springer, 2015.

\bibitem{DemaineFHT05talg}
E.~D. Demaine, F.~V. Fomin, M.~Hajiaghayi, and D.~M. Thilikos.
\newblock Fixed-parameter algorithms for $(k, r)$-center in planar graphs and
  map graphs.
\newblock {\em ACM Transactions on Algorithms}, 1(1):33--47, 2005.

\bibitem{demaine_et_al:DR:2013:4013}
E.~D. Demaine, F.~V. Fomin, M.~Hajiaghayi, and D.~M. Thilikos.
\newblock {Bidimensional Structures: Algorithms, Combinatorics and Logic
  (Dagstuhl Seminar 13121)}.
\newblock {\em Dagstuhl Reports}, 3(3):51--74, 2013.

\bibitem{DemaineFHT05}
E.~D. Demaine, F.~V. Fomin, M.~T. Hajiaghayi, and D.~M. Thilikos.
\newblock Subexponential parameterized algorithms on bounded-genus graphs and
  {$H$}-minor-free graphs.
\newblock {\em J. {ACM}}, 52(6):866--893, 2005.

\bibitem{Demaine:2008mi}
E.~D. Demaine and M.~Hajiaghayi.
\newblock The bidimensionality theory and its algorithmic applications.
\newblock {\em Comput. J.}, 51(3):292--302, 2008.

\bibitem{DemaineH08}
E.~D. Demaine and M.~Hajiaghayi.
\newblock Linearity of grid minors in treewidth with applications through
  bidimensionality.
\newblock {\em Combinatorica}, 28(1):19--36, 2008.

\bibitem{DemaineH04}
E.~D. Demaine and M.~T. Hajiaghayi.
\newblock Equivalence of local treewidth and linear local treewidth and its
  algorithmic applications.
\newblock In {\em Proceedings of the 15th Annual ACM-SIAM Symposium on Discrete
  Algorithms, SODA 2004}, pages 840--849. {SIAM}, 2004.

\bibitem{dorn:LIPIcs:2010:2460}
F.~Dorn.
\newblock {Planar Subgraph Isomorphism revisited}.
\newblock In {\em Proceedings of the 27th International Symposium on
  Theoretical Aspects of Computer Science, STACS 2010}, volume~5 of {\em
  Leibniz International Proceedings in Informatics (LIPIcs)}, pages 263--274,
  Dagstuhl, Germany, 2010. Schloss Dagstuhl--Leibniz-Zentrum fuer Informatik.

\bibitem{Dorn:2013wd}
F.~Dorn, F.~V. Fomin, D.~Lokshtanov, V.~Raman, and S.~Saurabh.
\newblock Beyond bidimensionality: {P}arameterized subexponential algorithms on
  directed graphs.
\newblock {\em Information and Computation}, 233:60--70, 2013.

\bibitem{DornFT12}
F.~Dorn, F.~V. Fomin, and D.~M. Thilikos.
\newblock Catalan structures and dynamic programming in {$H$}-minor-free
  graphs.
\newblock {\em J. Computer and System Sciences}, 78(5):1606--1622, 2012.

\bibitem{Eppstein99}
D.~Eppstein.
\newblock Subgraph isomorphism in planar graphs and related problems.
\newblock {\em J. Graph Algorithms and Applications}, 3:1--27, 1999.

\bibitem{FellowsFLRSV12}
M.~R. Fellows, F.~V. Fomin, D.~Lokshtanov, F.~A. Rosamond, S.~Saurabh, and
  Y.~Villanger.
\newblock Local search: Is brute-force avoidable?
\newblock {\em J. Comput. Syst. Sci.}, 78(3):707--719, 2012.

\bibitem{FominLS12}
F.~V. Fomin, D.~Lokshtanov, and S.~Saurabh.
\newblock Bidimensionality and geometric graphs.
\newblock In {\em Proceedings of the 22nd Annual ACM-SIAM Symposium on Discrete
  Algorithms, SODA 2012}, pages 1563--1575. SIAM, 2012.

\bibitem{FominLS14}
F.~V. Fomin, D.~Lokshtanov, and S.~Saurabh.
\newblock Efficient computation of representative sets with applications in
  parameterized and exact algorithms.
\newblock In {\em Proceedings of the 24th Annual ACM-SIAM Symposium on Discrete
  Algorithms, SODA 2014}, pages 142--151. SIAM, 2014.

\bibitem{Grohe03}
M.~Grohe.
\newblock Local tree-width, excluded minors, and approximation algorithms.
\newblock {\em Combinatorica}, 23(4):613--632, 2003.

\bibitem{ImpagliazzoP01}
R.~Impagliazzo and R.~Paturi.
\newblock On the complexity of $k$-{SAT}.
\newblock {\em J. Computer and System Sciences}, 62(2):367--375, 2001.

\bibitem{KleinM14}
P.~N. Klein and D.~Marx.
\newblock A subexponential parameterized algorithm for subset {TSP} on planar
  graphs.
\newblock In {\em Proceedings of the 24th Annual ACM-SIAM Symposium on Discrete
  Algorithms, SODA 2014}, pages 1812--1830. {SIAM}, 2014.

\bibitem{LiptonT80}
R.~J. Lipton and R.~E. Tarjan.
\newblock Applications of a planar separator theorem.
\newblock {\em SIAM J. Computing}, 9:615--627, 1980.

\bibitem{MatousekT92}
J.~Matou\v{s}ek and R.~Thomas.
\newblock On the complexity of finding iso- and other morphisms for partial
  k-trees.
\newblock {\em Discrete Mathematics}, 108(1-3):343--364, 1992.

\bibitem{NaorSS95}
M.~Naor, L.~J. Schulman, and A.~Srinivasan.
\newblock Splitters and near-optimal derandomization.
\newblock In {\em Proceedings of the 36th {IEEE} Annual Symposium on
  Foundations of Computer Science, FOCS 1995}, pages 182--191. IEEE, 1995.

\bibitem{pilipczuk_et_al:LIPIcs:2013:3947}
M.~Pilipczuk, M.~Pilipczuk, P.~Sankowski, and E.~J. van Leeuwen.
\newblock Subexponential-time parameterized algorithm for steiner tree on
  planar graphs.
\newblock In {\em Proceedings of the 30th International Symposium on
  Theoretical Aspects of Computer Science, STACS 2013}, volume~20 of {\em
  Leibniz International Proceedings in Informatics (LIPIcs)}, pages 353--364,
  Dagstuhl, Germany, 2013. Schloss Dagstuhl--Leibniz-Zentrum fuer Informatik.

\bibitem{PilipczukPSL14}
M.~Pilipczuk, M.~Pilipczuk, P.~Sankowski, and E.~J. van Leeuwen.
\newblock Network sparsification for {S}teiner problems on planar and
  bounded-genus graphs.
\newblock In {\em Proceedings of the 55th {IEEE} Annual Symposium on
  Foundations of Computer Science, FOCS 2014}, pages 276--285. {IEEE} Computer
  Society, 2014.

\bibitem{RobertsonS3}
N.~Robertson and P.~D. Seymour.
\newblock Graph minors. {III}. {P}lanar tree-width.
\newblock {\em J. Combinatorial Theory Ser. B}, 36:49--64, 1984.

\bibitem{RobertsonS-GMXIII}
N.~Robertson and P.~D. Seymour.
\newblock Graph minors. {XIII}. {T}he disjoint paths problem.
\newblock {\em J. Combinatorial Theory Ser. B}, 63(1):65--110, 1995.

\bibitem{SauT10}
I.~Sau and D.~M. Thilikos.
\newblock Subexponential parameterized algorithms for degree-constrained
  subgraph problems on planar graphs.
\newblock {\em J. Discrete Algorithms}, 8(3):330--338, 2010.

\bibitem{Tazari2010}
S.~Tazari.
\newblock {\em Algorithmic graph minor theory: approximation, parameterized
  complexity, and practical aspects}.
\newblock PhD thesis, Humboldt University of Berlin, 2010.

\bibitem{Tazari12}
S.~Tazari.
\newblock Faster approximation schemes and parameterized algorithms on
  (odd-){$H$}-minor-free graphs.
\newblock {\em Theor. Comput. Sci.}, 417:95--107, 2012.

\end{thebibliography}

\newpage\appendix
\newcommand{\compot}{\Lambda}
\newcommand{\CMonster}{\widehat{\Monster}}

\section{Extension to multiple components of the pattern and to $H$-minor-free graphs}\label{sec:minor-free}

In this section we develop the following extension of Theorem~\ref{thm:maintheorem} for graph classes excluding a fixed minor,
at the cost of a bound on the maximum degree of the pattern. By a {\em{proper minor-closed graph class}} we mean a graph class that is minor-closed and does not contain all graphs.

\begin{theorem}\label{thm:maintheorem-minor}
Let $\Cc$ be a proper minor-closed graph class, and let $\Delta$ be a fixed constant.
Then there exists a randomized polynomial-time algorithm that, given an $n$-vertex graph $G$ from $\Cc$ and an integer $k$,
samples a vertex subset $A\subseteq V(G)$ with the following properties:
\begin{itemize}
\item The induced subgraph $G[A]$ has treewidth $\Oh(\sqrt{k}\log k)$.
\item For every vertex subset $X\subseteq V(G)$ with $|X|\leq k$ 
such that $G[X]$ is connected and has a spanning tree of maximum degree
$\Delta$,
the probability that $X$ is covered by $A$, that is $X\subseteq A$, is at least $(2^{\Oh(\sqrt{k}\log^2 k)}\cdot n^{\Oh(1)})^{-1}$.
\end{itemize}
\end{theorem}

To see why the above assumption seems necessary with our techniques, let us look at the following example. 
Let $G$ be a graph that contains a universal vertex $v_0$ (i.e., adjacent to all vertices of $V(G) \setminus \{v_0\}$)
such that $G-v_0$ is planar. It is easy to see that, since $G-v_0$ is $K_5$-minor-free,
we have that $G$ is $K_6$-minor-free. Let $H$ be a connected pattern in $G$: a connected subgraph on $k$ vertices.
If $H$ contains $v_0$, then $H-v_0$ is a not necessarily connected pattern (subgraph) of $G-v_0$.
Hence, finding a connected $k$-vertex pattern in $G$ boils down to finding a not necessarily connected $(k-1)$-vertex pattern in $G-v_0$.
However, if we bound the maximum degree of the pattern, the $(k-1)$-vertex pattern $H-v_0$ in the graph $G-v_0$ has bounded number of connected components,
making the situation much more similar to the connected planar (or apex-minor-free) case.

We do not know how to handle arbitrary disconnected patterns (subgraphs) with our techniques. As we show in this section, we are able
to do it in a limited fashion, namely we can handle
up to roughly $\sqrt{k} / \lg k$ connected components without increasing the asymptotic bound in the exponential factor
in the success probability.
The proof of the following generalization of Theorem~\ref{thm:maintheorem} is described in Section~\ref{ss:ext:components}.

\begin{theorem}\label{thm:maintheorem-comps}
Let $\Cc$ be a class of graphs that exclude a fixed apex graph as a minor.
Then there exists a randomized polynomial-time algorithm that, given an $n$-vertex graph $G$ from $\Cc$ and an integer $k$,
samples a vertex subset $A\subseteq V(G)$ with the following properties:
\begin{itemize}
\item The induced subgraph $G[A]$ has treewidth $\Oh(\sqrt{k}\log k)$.
\item For every vertex subset $X\subseteq V(G)$ with $|X|\leq k$ 
such that $G[X]$ has $\Oh(\sqrt{k}/\log k)$ connected components,
the probability that $X$ is covered by $A$, that is $X\subseteq A$, is at least $(2^{\Oh(\sqrt{k}\log^2 k)}\cdot n^{\Oh(1)})^{-1}$.
\end{itemize}
\end{theorem}

After proving Theorem~\ref{thm:maintheorem-comps},
in Section~\ref{ss:ext:minor-free} we show how to use this extension for a bounded number of connected components in order to handle 
connected patterns in graph classes excluding a fixed minor. To this end, we use the Robertson-Seymour decomposition theorem that provides
a tree decomposition for any graph that exclude a fixed minor.
Roughly speaking, in this decomposition every bag corresponds to a graph almost embeddable into a fixed surface, and every adhesion (intersection of neighboring bags) has bounded size.
By a result of Grohe~\cite{Grohe03}, one can delete a bounded number of vertices from an almost embeddable graph to get an apex-minor-free graph. 
If the pattern we are looking for is connected and has bounded degree,
deleting a bounded number of vertices can split it only into a bounded
number of connected components.
Thus, the algorithm for graph classes excluding a fixed minor
boils down to an application
of either Theorem~\ref{thm:maintheorem-comps}
or a simple Baker-style argument
to every bag, after turning it into an apex-minor-free graph.

\subsection{Extension to bounded number of components}\label{ss:ext:components}

In this section we prove Theorem~\ref{thm:maintheorem-comps}, that is, we extend Theorem~\ref{thm:maintheorem} to handle bounded
number of connected components of the pattern.
We describe it as a series of modifications to the proof of Theorem~\ref{thm:maintheorem}
from Section~\ref{sec:proof}.

As in Section~\ref{sec:proof}, in a recursive step we are given
an instance $\Ii$ consisting of
a minor $G$ of the input graph $G_0$, a root $r \in V(G)$,
two disjoint sets of light and heavy terminals $\light,\heavy \subseteq V(G)$
with $r \in \light$, a set $\ghost \subseteq V(G) \setminus \trms$
of ghost vertices representing connectivity in other parts of the input graph,
and credit $\lambda$.
We maintain the same invariants regarding terminals:
every light terminal is within distance at most $3$ from the root, 
and the number of terminals is bounded by $16014\ctw \sqrt{k}\lg k + \lambda$.

The first significant difference is with regards to the definition of a pattern.
We start with the following definition.
\begin{definition}
Let $X \subseteq V(G) \setminus \ghost$ be a set of vertices.
Two vertices $x,y \in X$ are \emph{connected} if they belong to the same connected component of $G[X \cup \ghost]$.
A \emph{component} of the set $X$ is an equivalence class in the relation of being connected
(i.e., a set of vertices from $X$ from a connected component of $G[X \cup \ghost]$ that contains at least one vertex of $X$).
A component $Y$ is \emph{rooted} if it contains a vertex within distance at most $3$ from the root,
and \emph{free} otherwise.
\end{definition}

For a set $X \subseteq V(G) \setminus \ghost$ in an instance $\Ii$,
we introduce the following \emph{component potential} as the fourth potential:

\begin{center}
\begin{tabular}{ p{5cm} p{9cm} }
{\em{Component potential}} & $\compot_\Ii(X) := \textrm{number of free components of } X$.
\end{tabular}
\end{center}

We can now formally define a pattern.
A set $X \subseteq V(G) \setminus \ghost$ is a pattern if $r \in X$
and
$$|X| \leq k - 10\sqrt{k} \cdot \lambda - 486\sqrt{k}\lg k \cdot \compot_\Ii(X).$$
That is,
we drop the assumption of the connectivity of $X$ (possibly with help of some ghost vertices), but every free component imposes a penalty on
the allowed size of the pattern.
Note that every pattern contains at least one rooted component (the one containing the root $r$), and an arbitrary number of free components.

\subsubsection{Potentials}

Let us now proceed to the description of the potentials.
Apart from introducing the component potential, 
we extend the set of \emph{far} vertices: every vertex
in a free component is far, regardless of its distance from the root $r$.
\begin{equation*}
\Far_\Ii(X) := \{u\in X\ \colon\ \dist_G(u,r)>1000\sqrt{k}\lg k\quad \textrm{or}\quad u\textrm{ is in a free component}\}
\end{equation*}
As before, every vertex of the pattern that is not far is called \emph{close}.

Intuitively, every free component of the pattern decreases the success probability
of the algorithm by a factor inverse-quasipolynomial in $k$ and $\lg n$. Formally,
we define
\begin{align*}
& \CMonster(n,\patsiz(X),\grasiz,\dstpot(X),\compot(X)) := \\
& \qquad 
  \Monster(n,\patsiz(X),\grasiz,\dstpot(X)) \cdot \exp\left[-c_3 \cdot \compot(X) \cdot  \left(\lg^2 k (\lg k + \lg \lg n) + \frac{\lg n \lg k}{\sqrt{k}} \right)\right]
\end{align*}
for some constant $c_3$ that will be determined later in this section.
Here, again, $n$ denotes the total number of vertices of the graph, 
and we omit the subscript $\Ii$ whenever the instance is clear from the
context.

Our goal is to compute a subset of non-ghost vertices $A \subseteq V(G) \setminus \ghost$ with the following properties:
\begin{enumerate}[(1)]
\item It holds that $\light\subseteq A$, and the graph $G[A]$ admits a tree decomposition of width at most $24022\ctw\sqrt{k}\lg k$, where $\trms\cap A$ is contained in the root bag.
\item For every pattern $X$ in instance $\Ii$, we require that
\begin{equation}\label{eq:comp-prob}
\mathbb{P}(X\subseteq A)\geq \CMonster(n,\patsiz(X),\grasiz,\dstpot(X),\compot(X)).
\end{equation}
\end{enumerate}

\subsubsection{Operations on the instance}\label{sss:ops}

One of the crucial property of the algorithm of Section~\ref{sec:proof}
is that it modifies the input graph in a limited fashion. Namely, every subinstance 
is created by means of the following operations:
\begin{enumerate}
\item Edge contraction. Furthermore, if one of the contracted vertices is a ghost vertex, the new vertex remains a ghost vertex or the contraction
is made onto the root.
\item Other modifications such as vertex/edge deletion/addition, but only involving
vertices within distance larger than $2000 \sqrt{k} \lg k$ from $r$,
and not involving vertices in the pattern $X$ nor ghost vertices essential for the connectivity relation within the pattern
(assuming that the algorithm made correct random choices).
\end{enumerate}
The analysis of Section~\ref{sec:proof} used the above properties to ensure that
the algorithm never turns a close vertex into a far vertex, assuming that
the algorithm makes correct random choices.
Here, we observe that neither of the above modifications can create a new component.
Furthermore, a component that is rooted remains rooted, and a vertex belonging to a rooted component remains in a rooted component.
As a result, these modifications cannot turn a close vertex into a far vertex under the new definition of the far vertices,
nor create a new free component.
In particular, whenever we construct a pattern in a subinstance by projecting
the original pattern in the natural way, the projection remains a pattern
in the new instance. This is because the $486\sqrt{k} \lg k \cdot \compot(X)$
penalty in the upper bound on the size of the pattern does not increase.

\subsubsection{Solving the general problem}

First, note that we can make the same assumptions \eqref{inv:ghost}--\eqref{inv:pat}
as in Section~\ref{sec:proof}.

The general structure and the main steps of the algorithm are the same
as in Section~\ref{sec:proof}: we define the margin $M$
to be the set of vertices of $G$ within distance at most
$2000\sqrt{k}\lg k$ from the root $r$, and apply the
clustering procedure to the graph $\torso{(G-M)}{R \setminus M}$.
Note that the clustering procedure does not use the assumption
of the connectivity of the pattern.
Thus, we can assume that every island --- every connected component
of $G-M$ --- has radius at most $9k^2 \lg n$ (where ghost vertices
are traversed for free), at the cost
of a $(1-1/k)$ multiplicative factor in the success probability.
By slightly abusing the notation, 
we redefine $G$ to be the graph obtained from the clustering procedure;
this graph was named $G'$ in Section~\ref{sec:proof}.

By the same arguments as in Section~\ref{sec:proof},
by locally bounded treewidth
we obtain sets $Z$ and $W \subseteq V(G)$ with the following properties:
\begin{enumerate}[(1)]
\item $Z$ consists of a selection of vertices of $M$ and islands of $G-M$, and we have that $r \in Z$;
\item $|Z| \leq 8007 \ctw \sqrt{k} \lg k$;
\item $W$ consists of vertices of $Z \cap M$, denoted $W_{\mathsf{nrm}}$, and the vertices
that belong to the islands contained in $Z$, denoted $W_{\mathsf{isl}}$;
\item every connected component of $G-W$ contains at most $|\trms|/2$
terminals and at most $|V(G) \setminus (\light \cup \ghost)|/2$ vertices
that are neither light terminals nor ghost vertices.
\end{enumerate}
As before, we randomly select a branch we pursue: with probability $(1-1/k)$ we assume that
the pattern is disjoint with $W_{\mathsf{isl}}$, and with the remaining
probability $1/k$ we assume otherwise.
Thus, we have two cases: when $W_{\mathsf{isl}}$ is assumed to be disjoint with the pattern,
and when we suppose that $W_{\mathsf{isl}}$ intersects the pattern.

\subsubsection{Case when $W_{\mathsf{isl}}$ is disjoint with the pattern}\label{sss:comp:disjoint}

The crux in this case is to observe that nothing new happens, 
mostly because the argumentation of Section~\ref{sss:disjoint} does not rely
on the connectivity of the pattern.
That is, we argue that the algorithm
as described in Section~\ref{sec:proof} works also in our setting.

Recall that in this case we first delete $W_{\mathsf{isl}}$ from $G$; let the obtained graph be named $G''$, as in Section~\ref{sec:proof}.
Then recurse into instances $\Ii_D$ created for every 
connected component $D$ of $G-W=G''-W_{\textrm{nrm}}$, defined as in Section~\ref{sec:proof}.
In the instance $\Ii_D$ we look for pattern
$X_D := X \cap V_D$, where $V_D = N_{G''}[D]\cup\{r\}$.
We denote by $\Comps$
the set of connected components of $G-W$.

We now need to analyze the behavior of the free components
in the process of recursion.
We start with the following observation that follows directly from the discussion of Section~\ref{sss:ops}.
\begin{claim}\label{cl:stdred-comp1}
Let $Y$ be a component of $X$ in $\Ii$, and let $D \in \Comps$ be such that $Y \cap V_D \neq \emptyset$. 
Then $Y \cap V_D$ is contained in a single component of $X_D$ in $\Ii_D$.
\end{claim}
Second, we observe that the rooted components cannot give rise to any new free components.
\begin{claim}\label{cl:stdred-comp2}
Let $Y$ be a rooted component of $X$ in $\Ii$, and let $D \in \Comps$ be such that $Y \cap V_D \neq \emptyset$.
Then $Y \cap V_D$ is contained in a rooted component of $X_D$ in $\Ii_D$.
\end{claim}
\begin{proof}
Let $w$ be a vertex of $Y$ that is within distance at most $3$ from the root $r$. By the definition of a component, there exists
a path $P$ in $G''[Y \cup \ghost]$ between $w$ and a vertex $v \in Y \cap V_D$ such that
no vertex of $P-\{v\}$ belongs to $V_D$, except for possibly a neighbor $v'$ of $v$ on $P$
if $v'$ is a ghost vertex in $N_{G''}(D)$. Consequently, in the process of construction of $\Ii_D$,
the path $P-\{v\}$ is contracted either onto the root or onto a ghost vertex. 
As the distance between $w$ and $r$ is at most three in $G$ (hence also in $G''$), and ghost vertices are traversed
for free in our distance measure, we have that $v$ is within distance at most $3$ from the root in $\Ii_D$.
Consequently, $Y \cap V_D$ is contained in a rooted component of $\Ii_D$.
\cqed\end{proof}
Third, we observe that a free component cannot split into multiple free components.
\begin{claim}\label{cl:stdred-comp3}
Let $Y$ be a free component of $X$ in $\Ii$. Then there exists a component $D_0 \in \Comps$
such that for every $D \in \Comps$ such that $D \neq D_0$ and $Y \cap V_D \neq \emptyset$,
the set $Y \cap V_D$ is contained in a rooted component of $X_D$ in $\Ii_D$.
\end{claim}
\begin{proof}
We say that a component $D \in \Comps$ is \emph{touched} if $Y \cap V_D \neq \emptyset$.
We consider all paths in $G''$ between the root $r$ and a vertex $w \in N_{G''}[D]$ for a touched component $D$,
and pick $Q_0$ to be a shortest such path. Let $w \in N_{G''}[D_0]$ be the second endpoint of $Q_0$, where $D_0$ is a touched component.
By the minimality of $Q_0$, no vertex of $Q_0-\{w\}$ belongs to $N_{G''}[D]$ for a touched component $D$.
Let $Q_1$ be a shortest path between $w$ and a vertex $v \in Y \cap V_{D_0}$ with all internal vertices in $D_0$;
such a path exists by the connectivity of $G''[D_0]$, and by the minimality of $Q_1$ no vertex of $Q_1-\{v\}$ belongs to $Y$.

Consider a touched component $D \in \Comps$ different than $D_0$. By the definition of a component,
there exists a path $Q_2$ in $G''[Y \cup \ghost]$ between $v$ and a vertex $u \in Y \cap V_D$ such that
no vertex of $Q_2-\{u\}$ belongs to $V_D$, except for possibly a neighbor $u'$ of $u$ on $Q_2$
that is a ghost vertex in $N_{G''}(D)$. 
Observe that from the walk being the concatenation of the paths $Q_0$, $Q_1$, and $Q_2$, only the root $r$ and the vertices $w$, $u$, and $u'$
may potentially belong to $V_D$. Consequently, in $\Ii_D$, the vertex $u$ is within distance at most $2$ from the root (recall that $u'$ is a ghost vertex
if it belongs to $V_D$).
We infer that $Y \cap V_D$ is contained in a rooted component of $X_D$ in $\Ii_D$.
This finishes the proof of the claim.
\cqed\end{proof}

Claims~\ref{cl:stdred-comp1}--\ref{cl:stdred-comp3} justify the following.
\begin{claim}\label{cl:recur-bounds-comp}
The following holds:
\begin{eqnarray}
\label{eq:compot}\compot_\Ii(X) & \geq & \sum_{D \in \Comps} \compot_{\Ii_D}(X_D).
\end{eqnarray}
\end{claim}

Furthermore, Claims~\ref{cl:stdred-comp1}--\ref{cl:stdred-comp3} ensure that a close vertex of $X$ in $\Ii$
cannot become a far vertex in any of the instances $\Ii_D$.
Consequently, the potential analysis of Claim~\ref{cl:recur-bounds} holds also in our case.

Using Claim~\ref{cl:recur-bounds-comp}, the following claim follows
along the same lines as Claim~\ref{cl:case-disjoint-error}, finishing
the analysis of this subcase.

\begin{claim}\label{cl:case-disjoint-error-comp}
Supposing $X\cap W_{\mathsf{isl}}=\emptyset$, the algorithm outputs a set $A$ with $X\subseteq A$ with probability at least $\CMonster(n,\patsiz_\Ii(X),\grasiz_\Ii,\dstpot_\Ii(X),\compot_\Ii(X))$. 
This includes the $(1-1/k)$ probability of success of the preliminary clustering step, and $(1-1/k)$ probability that the algorithm makes the correct assumption that $X\cap W_{\mathsf{isl}}=\emptyset$.
\end{claim}

\subsubsection{Case when $W_{\mathsf{isl}}$ intersects the pattern}

Just as in Section~\ref{sec:proof}, in the second case we
\begin{enumerate}[(a)]
\item guess (by sampling at random) the island $u_C \in Z$ for which $C$ intersects the pattern, 
\item take $z$ to be a non-ghost vertex of $C$ that is at distance at most $9k^2\lg n$ from all vertices of $C$ within $C$,
\item guess (by sampling at random) the distance $d$ from $z$ to the pattern
within the island $C$, and 
\item\label{pr:contr} contract the vertices of $C$ 
within distance less than $d$ from $z$ onto $z$.
\end{enumerate}
In step~\eqref{pr:contr}, we perform the same distinction as in Section~\ref{sec:proof} between non-ghost vertices (for which we use distance less than $\max(d,1)$)
and ghost vertices (for which we use distance less than $\max(d,1)-1$).

As a result, by incurring a multiplicative factor
$$\frac{1}{k} \cdot |Z|^{-1} \cdot (10k^2 \lg n)^{-1} \geq k^{-11} (\lg n)^{-1}$$
in the success probability, we can assume that we have a vertex $z \notin \ghost$
with $\dist_G(r,z) > 2000 \sqrt{k} \lg k$ such that there
exists a vertex of the pattern within distance at most $1$ from $z$.
By slightly abusing the notation, $G$ denotes the graph after the modifications.
In this case, as before,
we apply the duality theorem (Theorem~\ref{thm:duality})
to the graph $\torso{G}{\ghost}$, pair of vertices $(s,t) = (r,z)$,
and parameters:
$$p=\lceil 120\sqrt{k}\lg k\rceil\qquad\textrm{and}\qquad q=k.$$
The further behavior of the algorithm, as well as its analysis, depends on the output of the duality theorem.
Thus, we need to consider two subcases: either the duality theorem returns a family of paths or a nested chain of circular separators.

\paragraph*{Subcase: a sequence of radial paths.}
Following the argumentation of Section~\ref{sec:proof},
in this section we are working with the following objects:
\begin{itemize}
\item A vertex $z \in V(G) \setminus \ghost$ with $\dist_G(z,r) > 2000 \sqrt{k} \lg k$,
such that some vertex of $X$ is within distance at most $1$ from $z$.
\item A sequence $P_1,P_2,\ldots,P_k$ of $(r,z)$-paths in $G$, such that
for every $i \in [k]$ the set $V(P_i)$ can be partitioned as
as $$V(P_i) = \{r,z\} \uplus (V(P_i) \cap \ghost) \uplus \Pub(P_i) \uplus \Prv(P_i),$$
where the sets $\Prv(P_i)$ are pairwise disjoint
and $|\Pub(P_i)| \leq 480 \sqrt{k} \lg k$.
\end{itemize}
The success probability so far in this case is at least $k^{-11} (\lg n)^{-1}$.

As in Section~\ref{sec:proof}, we randomly pick an index $i \in [k]$
and assume further that
$X \cap \Prv(P_i) = \emptyset$.
Such an index $i$ exists as $|X| \leq k$,
the sets $\Prv(P_i)$ are pairwise disjoint,
and $r \in X$ but $r \notin \Prv(P_i)$ for every $i$.
Hence, the success probability of this step is at least $1/k$.

We reduce $P_i$ in the same way as in Section~\ref{sec:proof}.
Let $v_0$ be the last light terminal on $P_i$ (it exists as $r$
is a light terminal), let $P'$ be the suffix of $P_i$ from $v_0$
to $z$, and let $v_0,v_1,\ldots,v_\ell=z$ be the vertices
of $(\Pub(P_i) \cap V(P')) \cup \{v_0,z\}$ in the order of 
their appearance on $P'$. For every $0 \leq j < \ell$, we inspect
the segment of $P'$ between $v_j$ and $v_{j+1}$.
If this segment contains some ghost vertex $g_j$, then contract it entirely onto $g_j$; if there is more than one ghost vertex, choose an arbitrary one as $g_j$.
Otherwise, if there are no ghost vertices on the segment, contract the whole segment onto vertex $v_{j}$.

Let $H$ be the resulting graph, and construct the instance
$\Ii'$ as in Section~\ref{sec:proof}.
We have $\ell\leq |\Pub(P_i)|+1\leq 485\sqrt{k}\lg k$ and
$\dist_G(r,v_0) \leq 3$, thus $\dist_H(r,z) \leq 488\sqrt{k} \lg k$.

However, the proof of Claim~\ref{cl:less-far} fails if the vertex
of the pattern within distance at most $1$ from $z$ belongs to a free component.
The crux here is that if this is the case, then
we can turn this free component into
close by adding $\{v_0,v_1,\ldots,v_\ell\}$ to the pattern $X$,
providing a gain in the potential $\compot(X)$.
Let $X' := X \cup \{v_0,v_1,\ldots,v_\ell\}$.
\begin{claim}\label{cl:less-far-comp}
We have
$\Far_{\Ii'}(X) \subseteq \Far_{\Ii}(X)$ and
$\compot_{\Ii'}(X) \leq \compot_{\Ii}(X)$. Furthermore,
one of the following holds:
\begin{itemize}
\item $|\Far_{\Ii}(X)\setminus \Far_{\Ii'}(X)|\geq 511\sqrt{k}\lg k$, or
\item $X'$ is a pattern in $\Ii'$ and $\compot_{\Ii'}(X') < \compot_{\Ii}(X)$.
\end{itemize}
\end{claim}
\begin{proof}
The first part of the claim follows directly from the discussion of Section~\ref{sss:ops} and the fact $H$ is created from $G$ by means of edge contractions,
in the same manner as in Section~\ref{sec:proof}.
For the second part, let $v \in X$ be a vertex within distance at most $1$ from $z$ in $G$ (possibly $v=z$).

If $v$ belongs to a rooted component of $X$ in $\Ii$, the analysis
of Claim~\ref{cl:less-far} remains valid.
That is, $G[X \cup \ghost]$ contains a path
$P$ from $v$ to a vertex $w$
that is within distance at most $3$
from the root, and the first $511 \sqrt{k}\lg k$ vertices
of $X$ of this path belong to $\Far_{\Ii}(X)$.
These vertices become close in $H$, as $\dist_H(r,v_0) \leq 488 \sqrt{k} \lg k$.

Hence, we are left with the case when $v$ belongs to some free component $Y$
of $X$ in $\Ii$.
The crucial observation is that in $H$, the vertex $v$ belongs to a rooted component
of $X'$, as $v_0 \in \light$ (and thus is within distance at most $3$ from the root) and $v_0,v_1,\ldots,v_\ell$ is a path in $H$.
By the discussion in Section~\ref{sss:ops}, no new free component is created
in $\Ii'$, while $Y$ stops to be free and 
becomes part of a rooted component in $\Ii'$. Hence,
$\compot_{\Ii'}(X') < \compot_{\Ii}(X)$.
Furthermore,
\begin{align*}
|X'| &\leq (\ell+1) + |X| \\
&\leq 486\sqrt{k}\lg k + \left(k - 10\sqrt{k} \cdot \lambda - 486\sqrt{k}\lg k \cdot \compot_{\Ii}(X)\right) \\
&\leq k - 10\sqrt{k} \cdot \lambda - 486 \sqrt{k} \lg k \cdot \compot_{\Ii'}(X').
\end{align*}
Consequently, $X'$ is a pattern in $\Ii'$, and the claim is proven.
\cqed\end{proof}

The potential gains in Claim~\ref{cl:less-far-comp} allow us to conclude
with the analog of Claim~\ref{cl:case-intersect-paths-error}.
For its proof, we need the following estimate.
\begin{claim}\label{cl:dupa}
For every $x,y > 0$ it holds that
$$(x+y)\lg(x+y) - x \lg x \leq y (2 + \lg (x+y)).$$
\end{claim}
\begin{proof}
We have
$$(x+y)\lg(x+y) - x \lg x = y \lg (x+y) + x \lg(1+y/x) \leq y \lg(x+y) + 2y,$$
where in the last inequality we have used the fact that $\lg(1+t) \leq 2t$
for every $t > 0$.
\cqed\end{proof}

\begin{claim}\label{cl:case-intersect-paths-error-comp}
Assume $c_1\geq 1$ and $c_3$ is sufficiently large. Supposing $X\cap W_{\mathsf{isl}}\neq \emptyset$ and the subroutine of Theorem~\ref{thm:duality} returned a sequence of paths, 
the algorithm outputs a set $A$ with $X\subseteq A$ with probability at least $\CMonster(n,\patsiz(X),\grasiz,\dstpot(X),\compot(X))$. 
This includes the $(1-1/k)$ probability of success of the preliminary clustering step, the $1/k$ probability that the algorithm makes the correct assumption that $X\cap W_{\mathsf{isl}}=\emptyset$,
the $k^{-7}$ probability of correctly choosing the vertex $z$,
the $(10k^2\lg n)^{-1}$ probability of choosing the right distance $d$, and the $1/k$ probability of choosing the right path index $i$.
\end{claim}
\begin{proof}
The proof follows the same lines as the proof of Claim~\ref{cl:case-intersect-paths-error}.
Note that we have $\compot_{\Ii'}(X) \leq \compot_{\Ii}(X)$
and $\dstpot_{\Ii'}(X) \leq \dstpot_{\Ii}(X)$.
If the first option of Claim~\ref{cl:less-far-comp} happens
(i.e., the drop in the potential $\dstpot(X)$), then the analysis
is the same as in Claim~\ref{cl:case-intersect-paths-error}.
It remains to analyze the second case.
Note that here we need to focus on all potentials, as we will analyze pattern
$X'$ in the instance $\Ii'$.

Since $|X' \setminus X| \leq 486 \sqrt{k}\lg k$,
we have (using Claim~\ref{cl:dupa} and $|X'| \leq k$ for the first inequality):
\begin{align*}
\patsiz_{\Ii'}(X') \lg \patsiz_{\Ii'}(X') - \patsiz_\Ii(X) \lg \patsiz_\Ii(X) & \leq 486 \sqrt{k} \lg k \cdot (2+\lg k) \leq 972 \sqrt{k} \lg^2 k \\
\dstpot_{\Ii'}(X') - \dstpot_{\Ii}(X) & \leq 486\sqrt{k}\lg k \\
\patsiz_{\Ii'}(X') \lg \grasiz_{\Ii'} - \patsiz_{\Ii}(X) \lg \grasiz_{\Ii} &\leq 486 \sqrt{k} \lg k \lg n\\
\compot_{\Ii'}(X') - \compot_{\Ii}(X) & \leq -1.
\end{align*}
Thus, a straightforward computation shows that
\begin{align*}
\frac{\CMonster(n',\patsiz_{\Ii'}(X'),\grasiz_{\Ii'},\dstpot_{\Ii'}(X'),\compot_{\Ii'}(X'))}{\CMonster(n,\patsiz_{\Ii}(X),\grasiz_\Ii,\dstpot_\Ii(X),\compot_\Ii(X))}  & \geq
\exp\left[-1458 \cdot c_1 \cdot \lg^2 k(\lg k + \lg \lg n)\right]\cdot \\
&\quad \left(1-\frac{1}{k}\right)^{c_2 \cdot 486 \sqrt{k} \lg k \lg n}\cdot\\
&\quad \exp\left[c_3 \left(\lg^2 k (\lg k + \lg \lg n) + \frac{\lg k \lg n}{\sqrt{k}}\right)\right].
\end{align*}
For sufficiently large $c_3$, namely
$c_3 \geq 1458 \cdot c_1 + 486 \lg e \cdot c_2 + 12$, the right hand side in the inequality
above is at least $k^{12} \lg n$, required to offset the 
success probability of the random choices, similarly as in the proof of Claim~\ref{cl:case-intersect-paths-error}.
This finishes the proof of the claim and the analysis of this subcase.
\cqed\end{proof}

\paragraph*{Subcase: nested chain of circular separators.}
Following the argumentation of Section~\ref{sec:proof},
in this section we are working with the following objects:
\begin{itemize}
\item A vertex $z \in V(G) \setminus \ghost$ with $\dist_G(z,r) > 2000 \sqrt{k} \lg k$,
such that some vertex of $X$ is within at most distance $1$ from $z$.
\item An $(r,z)$-separator chain $(C_1,\ldots,C_p)$ in $\torso{G}{\ghost}$ with $|C_j|\leq 2k$ for each $j\in [p]$, where $p = \lceil 120\sqrt{k} \lg k \rceil$.
\end{itemize}
The success probability so far in this case is at least $k^{-11} (\lg n)^{-1}$.

As in Section~\ref{sec:proof}, we treat $(C_1,\ldots,C_p)$ as a separator
chain in $G$, and drop the first three separators.
In this manner, we can assume $p \geq \lceil 117 \sqrt{k} \lg k \rceil$,
every $C_i$ is disjoint with $\light \cup \ghost$,
and all vertices within distance at most $3$ from the root, including all light terminals, are in $\reach(r,G-C_i)$ for every $i$.

Recall that in this case the algorithm of Section~\ref{sec:proof}
samples an index $i \in [p]$, an integer $\alpha$ between $1$ and $\sqrt{k}/10$,
and a set $Q \subseteq C_i$ of size $\alpha$.
The intuition is that we hope for $Q = C_i \cap X$, and $C_i$
being a sparse separator in the sense of Claim~\ref{cl:balanced-index}.

In our case, we additionally allow the value $\alpha = 0$ in the above
sampling, and whenever there exists an index $i \in [p]$
with $C_i \cap X = \emptyset$, we require from the algorithm to sample
such an index $i$ and sample $\alpha = 0$. 

If no such index exists, the analysis of the algorithm of Section~\ref{sec:proof}
remains applicable: Claim~\ref{cl:balanced-index} still holds, 
and only the probability of choosing the correct $\alpha$ drops from
$(\sqrt{k}/10)^{-1}$ to $(1+\sqrt{k}/10)^{-1}$; however, in both cases
it is at least $k^{-1}$, and the total probability of making correct
random choices is at least $k^{-2\alpha-5}$, as in Section~\ref{sec:proof}.
Furthermore, we have the following with regards to the potential $\compot(X)$:
\begin{claim}\label{cl:cycles:compot1}
If for every separator $C_i$ we have $C_i \cap X \neq \emptyset$, and the algorithm made correct random choices, then the following holds:
\begin{align*}
\compot_\Ii(X) & \geq \compot_{\Ii_\pin}(X_\pin) + \compot_{\Ii_\pout}(X_\pout).
\end{align*}
\end{claim}
\begin{proof}
Let $Y$ be a component of $X$ in $\Ii_D$. 
As discussed in Section~\ref{sss:ops}, $Y \cap X_\pin$ is either empty or contained in a single component $Y_\pin$ 
of $X_\pin$ in $\Ii_\pin$.
Similarly, $Y \cap X_\pout$ is empty or contained in a single connected component $Y_\pout$ of $X_\pout$ in $\Ii_\pout$.
We show that if $Y$ is a rooted component of $X$ in $\Ii$, then both $Y_\pin$ and $Y_\pout$ are rooted in their respective instances if they exist,
and if $Y$ is free, then at most one of these components is free.

We first note that if $Y \cap C_i \neq \emptyset$, then $Y_\pout$ exists and is rooted,
as all vertices of $X \cap C_i$ are neighbors of the root in $\Ii_\pout$.

Assume first that $Y$ is a rooted component of $\Ii$.
Then $Y_\pin$ is a rooted component of $\Ii_\pin$, as $\Ii_\pin$ is created from $\Ii$ by edge contractions only.
Furthermore, if $Y \cap C_i = \emptyset$, then $Y \cap X_\pout = \emptyset$, and otherwise as discussed above
$Y_\pout$ is a rooted component of $\Ii_\pout$.

Now assume that $Y$ is a free component of $\Ii$.
If $Y \cap C_i = \emptyset$, then $Y \cap X_\pin$ or $Y \cap X_\pout$ is empty.
Otherwise, as already discussed, $Y_\pout$ is a rooted component of $\Ii_\pout$.
This finishes the proof of the claim.
\cqed\end{proof}

Thus, we are left with the case when there exists a separator $C_i$
that is disjoint with $X$, and we assume that the algorithm correctly 
guessed such an index $i$ and guessed $\alpha = 0$.
The probability of making a correct choice is at least
$$\frac{1}{p} \cdot (1+\sqrt{k}/10)^{-1} \geq k^{-5} = k^{-2\alpha-5}.$$

Recall that for $i \in [p]$ we defined the following partition $V(G) = V^\pin_i \uplus C_i \uplus V^\pout_i$:
\begin{equation*}
V^\pin_i = \reach(r,G-C_i) \qquad \textrm{and} \qquad V^\pout_i = V(G)\setminus (C_i\cup V^\pin_i).
\end{equation*}
The intuition is as follows: since $C_i \cap X = \emptyset$
and $C_i$ does not contain any ghost terminal, we can independently
recurse on $V^\pin_i$ and $V^\pout_i$.
In $V^\pin_i$, we need to replace the components of $G-V^\pin_i$ with
ghost vertices to keep the potentials bounded. On the other hand,
$V^\pout_i$ does not contain any vertex within distance at most $3$ from the root, and thus every
vertex of the pattern $X$ in $V^\pout_i$ is in a free component
and thus far. Hence, we can freely choose a new root in this subcase: by proclaiming
$z$ the new root, we make the component containing a vertex of $X$ within
distance $1$ from $z$ close, making a gain in the $\compot(X)$ potential.

Let us now proceed with formal argumentation.
The algorithm defines two subinstances $\Ii_\pout$ and $\Ii_\pin$ as follows.

The instance $\Ii_\pin$ is defined in the same way as in Section~\ref{sec:proof}
for $\alpha = 0$ and $Q = \emptyset$.
That is, we take $\Ii_\pin=(G_\pin,r,\light_\pin,\heavy_\pin,\ghost_\pin,\lambda)$.
Graph $G_\pin$ is constructed from $G$ as follows.
Inspect the connected components of the graph $G-V^\pin_i$. For each such component $D$, contract it onto a new vertex $g_D$ that is declared to be a ghost vertex.
That is, we define $\ghost_\pin$ to be $(\ghost\cap V^\pin_i)\cup \{g_D\colon D\in \cc(G-V^\pin_i)\}$.

For the terminal sets $\light_\pin$, $\heavy_\pin$,
recall that $\light\subseteq V^\pin_i$, so all the original light terminals persist in the graph $G_\pin$.
Thus we take $\light_\pin = \light$.
For the heavy terminals, we inherit them: $\heavy_\pin = \heavy \cap V(G_\pin)$.
Finally, we take $X_{\pin}=X\cap V(G_\pin)$.
By the same arguments as in Section~\ref{sec:proof},
$\Ii_\pin$ is a valid instance with pattern $X_\pin$.

Second, we define $\Ii_\pout=(G_\pout,z,\light_\pout,\heavy_\pout,\ghost_\pout,\lambda)$; that is, we take the vertex $z$ to be the new root in
the instance $\Ii_\pout$.
Recall that $V^\pout_i$ does not contain any light terminal.
We define $G_\pout = G[V^\pout_i]$, $\light_\pout = \{z\}$,
$\heavy_\pout = \heavy \cap V^\pout_i$, and
$\ghost_\pout = \ghost \cap V^\pout_i$.
In other words, we inherit all terminals and ghost vertices from $\Ii$,
and additionally proclaim $z$ the root and a light terminal.
Since $r \in \light$, we have $|\trms| \geq |\trms_\pout|$ and, consequently,
$\Ii_\pout$ is a valid instance.

Finally, we take $X_\pout = (X \cap V^\pout_i) \cup \{z\}$.
Note that $|X_\pout| \leq |X|$, as $r \in X \setminus X_\pout$.

Since $C_i \cap (X \cup \ghost) = \emptyset$, every component of $X$
in $\Ii$ lies either entirely in $V^\pin_i$ or entirely in $V^\pout_i$.
Furthermore, since every vertex within distance at most $3$ from the root is in $V^\pin_i$, every component of $X$
lying in $V^\pout_i$ is free, and, consequently, all vertices
of $V^\pout_i \cap X$ are far.
Let $x \in X$ be a vertex within distance at most $1$ from $z$; clearly,
$x \in V^\pout_i$. The component $Y$ of $X$ in $\Ii$ containing $x$ is free,
but $Y \cup \{z\}$ is contained in a rooted component of $X_\pout$ in $\Ii_\pout$.
This, together with the discussion of Section~\ref{sss:ops} applied
to the instance $\Ii_\pin$, proves the following claim.
\begin{claim}\label{cl:comps:pots}
The following holds:
\begin{align*}
\patsiz_\Ii(X) & \geq \patsiz_{\Ii_\pin}(X_\pin) + \patsiz_{\Ii_\pout}(X_\pout) \\
\grasiz_\Ii & \geq \grasiz_{\Ii_\pin} + \grasiz_{\Ii_\pout} \\
\dstpot_\Ii(X) & \geq \dstpot_{\Ii_\pin}(X_\pin) + \dstpot_{\Ii_\pout}(X_\pout) \\
\compot_\Ii(X) & \geq 1 + \compot_{\Ii_\pin}(X_\pin) + \compot_{\Ii_\pout}(X_\pout).
\end{align*}
\end{claim}
In particular, the last inequality of Claim~\ref{cl:comps:pots} shows
that $X_\pout$ is a pattern in $\Ii_\pout$.

We apply the algorithm recursively to instances $\Ii_\pin$ and $\Ii_\pout$,
obtaining sets $A_\pin$ and $A_\pout$.
We have $\light = \light_\pin \subseteq A_\pin$ and $z \in A_\pout$.
We take $A := A_\pin \cup A_\pout$. Clearly, $\light \subseteq A$.

The fact that $G[A]$ admits a tree decomposition of width at most
$24022\ctw \sqrt{k} \lg k$ with $A \cap \trms$ in the root bag is straightforward:
since $A_\pin$ and $A_\pout$ are separated by $C_i$, there are no edges
between these two sets, and we can just take a root bag $A \cap \trms$,
and attach as children the decompositions of $G[A_\pin]$ and
$G[A_\pout]$.

We are left with analyzing the success probability.
All the necessary observations have already been made in Claim~\ref{cl:comps:pots},
so we can conclude with the following claim.

\begin{claim}\label{cl:case-intersect-cycles-error-comps}
Assume $c_1\geq 2$ and $c_3 \geq 17$. Supposing $X\cap W_{\mathsf{isl}}\neq \emptyset$ and the subroutine of Theorem~\ref{thm:duality} returned a separator chain, 
the algorithm outputs a set $A$ with $X\subseteq A$ with probability at least $\CMonster(n,\patsiz(X),\grasiz,\dstpot(X),\compot(X))$. 
This includes the $(1-1/k)$ probability of success of the preliminary clustering step, the $1/k$ probability that the algorithm makes the correct assumption that $X\cap W_{\mathsf{isl}}=\emptyset$,
the $k^{-7}$ probability of correctly choosing the island $C$ that intersects the pattern,
the $(10k^2\lg n)^{-1}$ probability of choosing the right distance $d$, and $k^{-2\alpha-5}$ probability of correctly sampling the $i$, $\alpha$, and set $Q$.
\end{claim}
\begin{proof}
The case $\alpha > 0$ has been already discussed, and is the same as in Section~\ref{sec:proof},
with the help of Claim~\ref{cl:cycles:compot1} to control the split of the potential $\compot(X)$. For $\alpha = 0$, Claim~\ref{cl:comps:pots}
ensures that
\begin{align*}
& \CMonster(n_\pin,\patsiz_{\Ii_\pin}(X_\pin),\grasiz_{\Ii_\pin},\dstpot_{\Ii_\pin}(X_\pin),\compot_{\Ii_\pin}(X_\pin)) \cdot 
   \CMonster(n_\pout,\patsiz_{\Ii_\pout}(X_\pout),\grasiz_{\Ii_\pout},\dstpot_{\Ii_\pout}(X_\pout),\compot_{\Ii_\pout}(X_\pout)) \\
& \qquad \geq 
   \CMonster(n,\patsiz_{\Ii}(X),\grasiz_{\Ii},\dstpot_{\Ii}(X),\compot_{\Ii}(X)) \cdot \exp\left[c_3 \lg^2 k (\lg k + \lg \lg n)\right].
\end{align*}
Hence,
\begin{eqnarray*}
\mathbb{P}(X\subseteq A) & \geq & \left(1-\frac{1}{k}\right)\cdot k^{-8}\cdot (10k^2\lg n)^{-1}\cdot k^{-5}\cdot\\
& & \CMonster(n_\pout,\patsiz_{\Ii_\pout}(X_\pout),\grasiz_{\Ii_\pout},\dstpot_{\Ii_\pout}(X_\pout),\compot_{\Ii_\pout}(X_\pout)) \cdot \\
& & \CMonster(n_\pin,\patsiz_{\Ii_\pin}(X_\pin),\grasiz_{\Ii_\pin},\dstpot_{\Ii_\pin}(X_\pin),\compot_{\Ii_\pin}(X_\pin))\\
& \geq & k^{-17}\cdot (\lg n)^{-1}\cdot \exp\left[c_3 \lg^2 k(\lg k + \lg \lg n)\right]
   \cdot \CMonster(n,\patsiz_{\Ii}(X),\grasiz_{\Ii},\dstpot_{\Ii}(X),\compot_{\Ii}(X)) \\
& \geq & \CMonster(n,\patsiz_{\Ii}(X),\grasiz_{\Ii},\dstpot_{\Ii}(X),\compot_{\Ii}(X))
\end{eqnarray*}
This finishes the proof of the claim and concludes the description
of the third and last subcase.
\cqed\end{proof}

\subsubsection{Wrap-up: a multiple-component version of Theorem~\ref{thm:maintheorem}}\label{sss:minor:comp-wrapup}

Let us now take a step back and use the developed recursive algorithm to obtain a multiple-component version of
Theorem~\ref{thm:maintheorem}, namely Theorem~\ref{thm:maintheorem-comps}, following the lines of the reasoning of Section~\ref{ss:proof:recur}. 

Assume we are given an $n$-vertex graph $G_0$ from a minor closed graph class $\Cc$ that excludes some apex graph, and we are looking for a pattern $X \subseteq V(G)$ of size $k$ such that $G[X]$ has $d$ connected components.
Similarly as in Section~\ref{ss:proof:recur}, we can guess (by sampling at random) one vertex $x \in X$
and create an instance $\Ii$
of our recursive problem with $G = G_0$, $r_0 = x$, $\light = \{x\}$,
and $\heavy = \ghost = \emptyset$. 
Since $G[X]$ has $d$ connected components, we have $\compot_\Ii(X) \leq d-1$.

However, $X$ may not be a pattern in $\Ii$ due to the size penalty 
incurred by multiple connected components. 
Instead, provided we assume that $d \leq c \cdot \sqrt{k}/ \lg k$ for some constant $c\geq 1$, we
can invoke the recursive subproblem with the parameter
$k' := (10^5\cdot c)^3 k = \Oh(k)$. Then, as $\sqrt{\alpha k} \lg (\alpha k) \leq \alpha^{2/3} \sqrt{k} \lg k$ for every $k \geq 10$ and $\alpha \geq 10^5$, we have:
$$k' - d \cdot 486 \cdot \sqrt{k'} \lg{k'} \geq (10^5\cdot c)^3 k - 486 \cdot 10^{10} c^3 k > k \geq |X|.$$
Consequently, $X$ is a pattern in $\Ii$ for the parameter $k' = \Oh(k)$.

Let us now look at the contribution of the term including the potential
$\compot_\Ii(X)$ into the success probability.
Using the estimate on $\Monster(n,\patsiz(X),\grasiz,\dstpot(X))$
from Section~\ref{ss:proof:recur}, we obtain that
\begin{align*}
& \CMonster(n,\patsiz(X),\grasiz,\dstpot(X),\compot(X)) = \\
& \qquad 
  = \Monster(n,\patsiz(X),\grasiz,\dstpot(X)) \cdot \exp\left[-c_3 \cdot \compot(X) \cdot  \left(\lg^2 k' (\lg k' + \lg \lg n) + \frac{\lg n \lg k'}{\sqrt{k'}} \right)\right]\\
& \qquad \geq \left(2^{\Oh(\sqrt{k}\lg^2 k)} n^{\Oh(1)}\right)^{-1} \cdot 
\exp\left[-c_3 \cdot c \cdot \left(\sqrt{k'} \lg k' (\lg k' + \lg \lg n) + \lg n\right)\right] \\
& \qquad \geq \left(2^{\Oh(\sqrt{k}\lg^2 k)} n^{\Oh(1)}\right)^{-1}.
\end{align*}
Note that we used Claim~\ref{cl:loglog} in the last inequality.
This finishes the proof of Theorem~\ref{thm:maintheorem-comps}.

\subsection{The excluded minor case}\label{ss:ext:minor-free}

Using the multiple-components variant of the previous section, we now prove Theorem~\ref{thm:maintheorem-minor}.
Let us fix a graph $G$ and an integer $k$ as in the theorem statement. Furthermore, let $X$ be a pattern
in $G$ such that $G[X]$ admits a spanning tree $S$ of maximum degree $\Delta$, for a fixed constant $\Delta$.

\paragraph{The Robertson-Seymour decomposition theorem.}
As announced at the beginning of this section, we use the decomposition
theorem of Robertson and Seymour for graphs excluding a fixed minor.
To make use of locally bounded treewidth, we will use the variant 
of Grohe~\cite{Grohe03}.

To formulate this decomposition, we need some notation.
Recall that we use a notation $\Tt = (T,\beta)$ for a tree decomposition,
where $T$ is a tree and $\beta: V(T) \to 2^{V(G)}$ are bags.

The set $\beta(t) \cap \beta(t')$ for an edge $tt' \in E(T)$
is called an \emph{adhesion} of $tt'$.
For a node $t \in V(T)$, the \emph{torso} of the node $t$,
denoted by $\torsoo(t)$,
is the graph $G[\beta(t)]$ with every adhesion $\beta(t) \cap \beta(t')$
for $t' \in N_T(t)$ turned into a clique.

Recall that we consider rooted tree decompositions; that is, the tree
$T$ is rooted in one node. For a non-root node $t \in V(T)$,
by $\parent(t)$ we denote the parent of $t$ in $T$.
Furthermore, we denote:
\begin{align*}
\adh(t) & = \begin{cases} \emptyset & \text{if t is the root of }T \\ \beta(t) \cap \beta(\parent(t)) & \text{otherwise.}\end{cases}
\end{align*}

We are now ready to formulate the variant of the Robertson-Seymour decomposition of Grohe that we use.
\begin{theorem}[\cite{Grohe03}]\label{thm:grohe}
For every proper minor-closed graph class $\Cc$ 
there exists a constant $h$ and an apex-minor-free graph class $\Cc'$
such that the following holds.
Given a graph $G \in \Cc$, one can in polynomial time compute
a tree decomposition $\Tt = (T,\beta)$ of $G$ together
with a family of sets $(Z_t)_{t \in V(T)}$ such that
every adhesion of $\Tt$ has size at most $h$, and for every $t \in V(T)$
the set $Z_t$ is a subset of $\beta(t)$ of size at most $h$, and the graph
$\torsoo(t)-Z_t$ belongs to $\Cc'$.
\end{theorem}

Furthermore, we need the following variant of Baker's shifting technique.
\begin{theorem}[\cite{Grohe03}]\label{thm:baker-apex}
Let $\Cc$ be a proper minor-closed graph class.
Given a graph $G \in \Cc$ and an integer $\ell$, one can in polynomial time compute a
partition of $V(G)$ into $\ell$ sets $L_1,L_2,\ldots,L_\ell$, such that for every $1 \leq i \leq \ell$
the graph $G-L_i$ is of treewidth $\Oh(\ell)$.
\end{theorem}

\paragraph{The algorithm.}
Let us now proceed to the description of the algorithm.
Given a graph $G \in \Cc$, we compute its tree decomposition $\Tt = (T,\beta)$
and sets $(Z_t)_{t \in V(T)}$ using Theorem~\ref{thm:grohe}.
Paying $1/n$ in the success probability,
we guess an arbitrary vertex $r \in X$, and root the decomposition $\Tt$
in a bag $t_r$ such that $r \in \beta(t_r)$. 
By slightly abusing the notation, we proclaim $\adh(t_r) = \{r\}$.
Since we can assume that $G$ is connected, we can henceforth assume
that every adhesion is nonempty.

For every $t \in V(T)$, we create an instance $\Ii_t = (G_t,r_t,\light_t,\heavy_t,\ghost_t,\lambda_t)$ of the recursive problem as follows.
We first take $G_t = \torsoo(t)-Z_t$. If $\adh(t) \not \subseteq Z_t$, we define $r_t$ to be an arbitrary vertex of $\adh(t) \setminus Z_t$;
otherwise, we create a new vertex $r_t$ and make it adjacent to an arbitrary vertex of $G_t$.
We set $\light_t = \{r_t\} \cup (\adh(t) \setminus Z_t)$, $\heavy_t = \emptyset$, $\ghost_t = \emptyset$, and $\lambda_t = 0$.

Furthermore, we define $X_t^\circ = X \cap \beta(t)$ and $X_t = (X_t^\circ \setminus Z_t) \cup \{r_t\}$.
Note that $\torsoo(t)[X_t^\circ]$ is connected, as $G[X]$ is connected.
Furthermore, we claim that $\torsoo(t)[X_t^\circ]$ admits a spanning tree of bounded degree.
\begin{claim}\label{cl:torsodeg}
There exists a spanning tree of $\torsoo(t)[X_t^\circ]$ of maximum degree $2\Delta$.
\end{claim}
\begin{proof}
We construct a connected spanning subgraph $S_t$ of $\torsoo(t)[X_t^\circ]$ of maximum degree $2\Delta$ as follows.
First, we take $V(S_t) = X_t^\circ$ and $E(S_t) = E(S) \cap E(\torsoo(t)[X_t^\circ]$. Second, for every $t' \in N_T(t)$, we perform the following operation.
Let $\textrm{leg}(t,t')$ be the set of those vertices $v \in X_t^\circ \cap \beta(t) \cap \beta(t')$ 
for which $v$ is incident to an edge $uv \in E(S)$ with $u \in \beta(t') \setminus \beta(t)$.
If $|\textrm{leg}(t,t')| \geq 2$, we add to $E(S_t)$ edges of an arbitrary path on vertex set $\textrm{leg}(t,t')$;
such a path exists in $\torsoo(t)$ as the adhesion $\beta(t) \cap \beta(t')$ is turned into a clique.

Let us now show that $S_t$ is connected. To this end, consider a maximal path $P$ in $S$ between two vertices of $X_t^\circ$
such that no edge or internal vertex of $P$ belongs to $\torsoo(t)$. Let $v_1,v_2$ be the endpoints of $P$.
By the properties of a tree decomposition, there exists a node $t' \in N_T(t)$ such that
$v_1,v_2 \in \beta(t) \cap \beta(t')$ and the first and last edges of $P$ are $v_1u_1$ and $v_2u_2$ with $u_1,u_2 \in \beta(t') \setminus \beta(t)$ 
(possibly $u_1 = u_2$). However, then $v_1,v_2 \in \textrm{leg}(t,t')$, and they remain connected in $S_t$. This shows that $S_t$ is connected.

To bound the maximum degree of $S_t$, note that for every $t' \in N_T(t)$ and every $v \in \textrm{leg}(t,t')$, at most two edges
incident to $v$ are added to $S_t$ when considering $t'$. These two edges can be charged to the edge $vu\in E(S)$ with $u \in \beta(t') \setminus \beta(t)$
that certifies that $v \in \textrm{leg}(t,t')$. We have $vu \in E(S)\setminus E(S_t)$, and every edge $vu$ can be charged at most once.
Since $S$ has maximum degree at most $\Delta$ by assumption, the bound on the maximum degree of $S_t$ follows.
\cqed\end{proof}

We define $k' := (10^5\cdot \Delta h)^3 \cdot k = \Oh(k)$.
We will use the machinery of Section~\ref{ss:ext:components}, in particular all the potentials, using the parameter $k'$ instead of the input parameter $k$.
The reason for this is that we need to pay the penalty in the size of the pattern for multiple connected components
of $X_t$ in $\Ii_t$. The following claim verifies that it suffices to inflate $k$ by a constant factor.
\begin{claim}\label{cl:minor:cc}
The graph $G_t[X_t]$ has at most $2\Delta h+1$ connected components and $\compot_{\Ii_t}(X_t) \leq 2\Delta h$.
Furthermore, the set $X_t$ is a pattern in $\Ii_t$ with respect to the parameter $k'$.
\end{claim}
\begin{proof}
Since the maximum degree of $S_t$ is at most $2\Delta$ and $|Z_t| \leq h$, there are at most $2\Delta h$ connected components
of $\torsoo(t)[X_t^\circ]-Z_t$. Consequently, $G_t[X_t]$ has at most $2\Delta h + 1$ connected components and
$$\compot_{\Ii_t}(X_t) \leq 2\Delta h.$$
Then, as $\sqrt{\alpha k} \lg (\alpha k) \leq \alpha^{2/3} k$ for every $k \geq 10$ and $\alpha \geq 10^5$, we have:
$$k' - 486 \sqrt{k'} \lg k' \cdot \compot_{\Ii_t}(X_t) \geq (10^5\cdot \Delta h)^3 k - 486 \cdot (10^5\cdot \Delta h)^2 \cdot 2\Delta h > 2k \geq |X_t|.$$
Thus, $X_t$ satisfies the size bound for a pattern in $\Ii_t$.
\cqed\end{proof}
For every node $t \in V(T)$, we are going to look for the pattern $X_t$ in the instance $\Ii_t$, using $k'$ as the parameter. 
Observe that the first three potentials partition well between the instances.
\begin{claim}\label{cl:minor:pots}
The following holds, if we measure the potentials with respect to the parameter $k'$:
\begin{equation*}
k \geq \sum_{t \in V(T)} \patsiz_{\Ii_t}(X_t) \qquad \textrm{and} \qquad
n \geq \sum_{t \in V(T)} \grasiz_{\Ii_t} \qquad \textrm{and} \qquad
k \geq \sum_{t \in V(T)} \dstpot_{\Ii_t}(X_t).
\end{equation*}
\end{claim}
\begin{proof}
The crucial observation is that for an edge between $t$ and $\parent(t)$ in $T$, every vertex $v \in \adh(t) = \beta(t)\cap \beta(\parent(t))$
is either in $Z_t$ or a light terminal in $\Ii_t$. Consequently, every vertex $v \in V(G)$ is present but \emph{not} a light terminal in at most
one instance $\Ii_t$.
\cqed\end{proof}
However, the potentials $\compot_{\Ii_t}(X_t)$ do not behave as nicely as the other potentials
in Claim~\ref{cl:minor:pots}: they are bounded by $2\Delta h$ by Claim~\ref{cl:minor:cc},
but may be positive in $\Omega(k)$ instances. Thus, we cannot afford to apply the algorithm of the previous section to every 
instance $\Ii_t$ separately.

Instead, for every instance $\Ii_t$, we make a random choice.
With probability $1/k$, we proclaim $\Ii_t$ \emph{interesting} and apply the recursive algorithm to $\Ii_t$ (that belongs to an apex-minor-free graph class $\Cc'$ promised by Theorem~\ref{thm:grohe}),
obtaining a set $A_t \subseteq (\beta(t) \setminus Z_t) \cup \{r_t\}$ such that $G_t[A_t]$ is of treewidth $\Oh(\sqrt{k'} \lg k') = \Oh(\sqrt{k}\lg k)$.
Furthermore, note that $\adh(t) \subseteq A_t \cup Z_t$.

With the remaining probability, we proclaim $\Ii_t$ \emph{standard}, and proceed as follows.
First, we apply the algorithm of Theorem~\ref{thm:baker-apex} to the graph $G_t$
with $\ell := \lceil c_3 \sqrt{k'} \lg k' \rceil$ obtaining a partition $L_1^t,L_2^t,\ldots,L_\ell^t$. 
Second, we pick a random index $1 \leq i_t \leq \ell$. Third, we set $A_t := (V(G_t) \setminus L_{i_t}^t) \cup \{r_t\} \cup (\adh(t) \setminus Z_t)$.
Note that in this case also $G_t[A_t]$ has treewidth $\Oh(\sqrt{k}\lg k)$ as $\ell = \Oh(\sqrt{k}\lg k)$ and $|\adh(t)| \leq h = \Oh(1)$.
Furthermore, we have again $\adh(t) \subseteq A_t \cup Z_t$.

We define
$$A := \{r\} \cup \bigcup_{t \in V(T)} (A_t \cup Z_t) \setminus (\adh(t) \cup \{r_t\}).$$
We claim that $A$ satisfies the desired properties. The treewidth bound is easy.
\begin{claim}\label{cl:minor:tw}
$G[A]$ is of treewidth $\Oh(\sqrt{k} \lg k)$.
\end{claim}
\begin{proof}
Since $|Z_t| \leq h = \Oh(1)$ for every $t \in V(T)$, we have that $\torsoo(t)[A_t \cup Z_t]$ is of treewidth $\Oh(\sqrt{k} \lg k)$.
Since $\adh(t) \subseteq A_t \cup Z_t$ for every $t \in V(T)$, 
we have that $G[A]$ can be constructed from $\torsoo(t)[A_t \cup Z_t]$ by vertex deletions and clique sums along cliques of size at most $h$, and the claim follows.
\cqed\end{proof}
Finally, we check the probability that $X \subseteq A$.
To this end, we need the following simple estimate.
\begin{claim}\label{cl:minor:dupa}
Let $a,b$ be positive integers, and $a_1,a_2,\ldots,a_r$ be integers such that $0 \leq a_i < a$ for every $1 \leq i \leq r$ and $\sum_{i=1}^r a_i \leq b$.
Then
$$\prod_{i=1}^r \left(1-\frac{a_i}{a}\right) \geq a^{-2b/a - 1}.$$
\end{claim}
\begin{proof}
We use the following local improvement argument: whenever we have two indices $1 \leq i < j \leq r$ such that $a_i+a_j < a$,
we can replace $a_i$ and $a_j$ with $a_i+a_j$, since
$$\left(1-\frac{a_i}{a}\right) \left(1-\frac{a_j}{a}\right) \geq 1-\frac{a_i+a_j}{a}.$$
Thus, we can assume that for every $1 \leq i < j \leq r$ we have that $a_i + a_j \geq a$. In particular,
every index $i$ satisfies $a_i \geq a/2$, apart from at most one. We infer that $r \leq 2b/a + 1$. Since $1-\frac{a_i}{a}\geq a^{-1}$ for each $i$ due to $a,a_i$ being integers, the claim follows.
\cqed\end{proof}
\begin{claim}
The probability that $X \subseteq A$ is at least $(2^{\Oh(\sqrt{k} \lg^2 k)}\cdot n^{\Oh(1)})^{-1}$.
\end{claim}
\begin{proof}
Note that we have the following partition:
$$V(G) = \{r\} \uplus \biguplus_{t \in V(T)} \beta(t) \setminus \sigma(t).$$
By the definition of the set $A$, we have $r \in A$, and for every $t \in V(T)$ and for every $v \in \beta(t) \setminus \sigma(t)$
it holds that $v \in A$ if and only if $v \in A_t \cup Z_t$. Consequently, by the definition of $X_t$ and $G_t$,
if for every $t \in V(T)$ we have $X_t \subseteq A_t$, then we have $X \subseteq A$.
In what follows we will argue that with sufficient probability it holds that for every $t \in V(T)$ we have $X_t \subseteq A_t$.

First, observe that this assertion is clearly true for every node $t$ where $X_t \subseteq \{r_t\} \cup (\adh(t) \setminus Z_t)$ as
both in standard and interesting nodes we have $\adh(t) \subseteq A_t \cup Z_t$ and $r_t \in A_t$.

If this is not the case for a node $t$ (i.e., $X_t \not\subseteq \{r_t\} \cup (\adh(t) \setminus Z_t)$),
we call the node $t$ \emph{touched}.
Note that we have $\patsiz_{\Ii_t}(X_t) > 0$ for a touched node $t$.
Hence, there are at most $k$ touched nodes.
We require that a touched node $t$ is proclaimed interesting if $\patsiz_{\Ii_t}(X_t) \geq c_3 \sqrt{k'} \lg k'$, and standard otherwise.
Note that Claim~\ref{cl:minor:pots} implies that we require at most $\sqrt{k'} / (c_3 \lg k')$ nodes to be interesting and at most $k$ nodes to be standard.
Consequently, the probability that we proclaim touched nodes as requested is $2^{-\Oh(\sqrt{k})}$.

In every standard touched node $t$
we have $X_t \subseteq A_t$ if $X_t \cap L_{i_t}^t \subseteq \light_t$, as $\light_t = \{r_t\} \cup (\adh(t) \setminus Z_t)$.
We have $X_t \cap L_{i_t}^t \subseteq \light_t$ with probability at least $1-|X_t \setminus \light_t|/\ell = 1-\patsiz_{\Ii_t}(X_t) / \ell$.
Recall that $\ell = \lceil c_3 \sqrt{k'} \lg k' \rceil$ but $\patsiz_{\Ii_t}(X_t) < \ell$ in a standard node $t$.
Consequently, by Claim~\ref{cl:minor:dupa}, the probability that in every standard node we have $X_t \subseteq A_t \cup \{r_t\}$ is at most
$$\ell^{-2k'/\ell-1} = 2^{-\Oh(\sqrt{k})}.$$

Let us now consider an interesting node $t$, that is, a node $t$ with $\patsiz_{\Ii_t}(X_t) \geq c_3 \sqrt{k'} \lg k'$.
Let $W_\mathsf{int}$ be the set of these nodes; note that $|W_\mathsf{int}| \leq \sqrt{k'} / (c_3 \lg k')$.
Since $X_t$ is a pattern in $\Ii_t$, we have $X_t \subseteq A_t$ with probability at least
$\CMonster(n,\patsiz_{\Ii_t}(X_t),\grasiz_{\Ii_t},\dstpot_{\Ii_t}(X_t),\compot_{\Ii_t}(X_t))$ (with respect to the parameter $k'$).
By Claims~\ref{cl:minor:cc} and~\ref{cl:minor:pots} we have that
\begin{align*}
&\prod_{t \in W_\mathsf{int}} \CMonster(n,\patsiz_{\Ii_t}(X_t),\grasiz_{\Ii_t},\dstpot_{\Ii_t}(X_t),\compot_{\Ii_t}(X_t)) \\
& \qquad \leq \prod_{t \in W_\mathsf{int}} \Monster(n,\patsiz_{\Ii_t}(X_t),\grasiz_{\Ii_t},\dstpot_{\Ii_t}(X_t)) \\
& \qquad\qquad \cdot \prod_{t \in W_\mathsf{int}} \exp\left[-c_3 \cdot \compot_{\Ii_t}(X_t) \cdot  \left(\lg^2 k' (\lg k' + \lg \lg n) + \frac{\lg n \lg k'}{\sqrt{k'}} \right)\right] \\
& \qquad \leq \left(2^{\Oh(\sqrt{k} \lg^2 k)} n^{\Oh(1)}\right)^{-1} \\
&\qquad \qquad \cdot \exp\left[-|W_\mathsf{int}| \cdot c_3 \cdot 2\Delta h \cdot  \left(\lg^2 k' (\lg k' + \lg \lg n) + \frac{\lg n \lg k'}{\sqrt{k'}} \right)\right] \\
& \qquad \leq \left(2^{\Oh(\sqrt{k} \lg^2 k)} n^{\Oh(1)}\right)^{-1} \cdot \exp\left[-2\Delta h\left(\sqrt{k'} \lg k' (\lg k' + \lg \lg n) + \lg n\right)\right]\\
& \qquad \leq \left(2^{\Oh(\sqrt{k} \lg^2 k)} n^{\Oh(1)}\right)^{-1}. 
\end{align*}
Here, we estimated the product of the terms $\Monster(n,\patsiz_{\Ii_t}(X_t),\grasiz_{\Ii_t},\dstpot_{\Ii_t}(X_t))$ using Claim~\ref{cl:minor:pots}
as in Section~\ref{sec:proof}, and in the last inequality we used Claim~\ref{cl:loglog}.
\cqed\end{proof}
This concludes the proof of Theorem~\ref{thm:maintheorem-minor}.

\end{document}